\documentclass{siamart0516}

\usepackage{lipsum}
\usepackage{amsfonts}
\usepackage{graphicx}
\usepackage{epstopdf}
\usepackage{algorithmic}
\usepackage{overpic}
\ifpdf
  \DeclareGraphicsExtensions{.eps,.pdf,.png,.jpg}
\else
  \DeclareGraphicsExtensions{.eps}
\fi

\newcommand{\TheTitle}{Noisy Interactive Quantum Communication} 
\newcommand{\TheAuthors}{G. Brassard, A. Nayak, A. Tapp, D. Touchette, and F. Unger}

\headers{\TheTitle}{\TheAuthors}

\title{{\TheTitle}\thanks{Submitted to the editors 12 October 2016. An extended abstract of this work appeared in FOCS'14~\cite{BNTTU14}.
\funding{G.B.~is supported in part by the Natural 
Sciences and Engineering Research Council of Canada (NSERC), the Canada 
Research Chair program, the Canadian Institute for Advanced Research
(CIFAR) and the Institute for Theoretical Studies of ETH~Zurich. 
A.N.'s research was conducted in part at Perimeter Institute
and supported in part by NSERC Canada, CIFAR, an ERA (Ontario), 
QuantumWorks, MITACS, and ARO (USA). Research at Perimeter Institute
for Theoretical Physics is supported in part by the Government of Canada
through Industry Canada and by the Province of Ontario through MRI.
A.T.~was supported by NSERC and CIFAR. D.T.'s research was conducted while he was a student 
in~D\'epartement d'informatique et de recherche op\'erationnelle, Universit\'e de Montr\'eal, and~was supported by a Fonds de Recherche 
Qu\'ebec -- Nature et Technologies B2 Doctoral research scholarship.
F.U.'s research was conducted in part at UC Berkeley.}}}

\author{
  Gilles Brassard\thanks{D\'epartement d'informatique et de recherche op\'erationnelle, Universit\'e de Montr\'eal
    (\email{brassard@iro.umontreal.ca}, \email{alain.tapp@gmail.com}).}
  \and
  Ashwin Nayak\thanks{Department of Combinatorics and Optimization, and 
Institute for Quantum Computing, University of Waterloo (\email{anayak@uwaterloo.ca},
    \email{touchette.dave@gmail.com}).}
  \and
  Alain Tapp\footnotemark[2]
  \and
  Dave Touchette\footnotemark[3] \thanks{Perimeter Institute for Theoretical Physics, Waterloo.}
  \and
  Falk Unger\thanks{No Affiliation (\email{falk.unger@gmail.com}).}
}

\usepackage{amsopn}

\usepackage{amsmath}
\usepackage{amssymb}
\usepackage[mathscr]{euscript}

\renewcommand{\epsilon}{\varepsilon}
\newcommand{\ket}[1]{\mathop{\left|#1\right>}\nolimits}
\newcommand{\bra}[1]{\mathop{\left<#1\,\right|}\nolimits}

\newcommand{\kb}[2]{| #1\rangle\!\langle #2 |}
\newcommand{\Tr}[2]{\mathop{{\mathrm{Tr}}_{#1}} (#2) }
\newcommand{\Tra}[1]{\mathop{{\mathrm{Tr}}_{#1}} }
\newcommand{\wt}[1]{\mathop{{\mathrm{wt}}} (#1) }

\def\N{\mathcal{N}}
\def\L{\mathcal{L}}
\def\A{\mathcal{A}}
\def\M{\mathcal{M}}
\def\T{\mathcal{T}}
\def\E{\mathcal{E}}
\def\F{\mathcal{F}}
\def\D{\mathcal{D}}
\def\H{\mathcal{H}}
\def\X{\mathcal{X}}

\def\R{\mathcal{R}}
\def\sA{{\mathsf A}}
\def\sAM{\mathsf{AM}}
\def\sAD{\mathsf{AD}}
\def\sB{{\mathsf B}}
\def\sBM{\mathsf{BM}}
\def\sBD{\mathsf{BD}}

\newcommand{\set}[1]{\left\{ #1 \right\}}
\newcommand{\size}[1]{\left| #1 \right|}
\newcommand{\tensor}{\otimes}
\newcommand{\rI}{{\mathrm I}}
\newcommand{\rX}{{\mathrm X}}
\newcommand{\rZ}{{\mathrm Z}}
\newcommand{\rC}{{\mathrm C}}
\newcommand{\rE}{{\mathrm E}}
\newcommand{\rH}{{\mathrm H}}
\newcommand{\rg}{{\mathrm g}}
\newcommand{\rb}{{\mathrm b}}
\newcommand{\rB}{{\mathrm B}}
\newcommand{\re}{{\mathrm e}}
\newcommand{\rs}{{\mathrm s}}
\newcommand{\ra}{{\mathrm a}}
\newcommand{\rd}{{\mathrm d}}
\newcommand{\rT}{{\mathrm T}}
\newcommand{\rS}{{\mathrm S}}
\newcommand{\rQ}{{\mathrm Q}}
\newcommand{\rO}{{\mathrm O}}
\newcommand{\IP}{{\mathrm{IP}}}
\newcommand{\complexi}{{\mathrm i}}
\newcommand{\msA}{{\mathscr A}}
\newcommand{\msE}{{\mathscr E}}
\newcommand{\msD}{{\mathscr D}}
\newcommand{\msB}{{\mathscr B}}
\newcommand{\suppress}[1]{}

\ifpdf
\hypersetup{
  pdftitle={\TheTitle},
  pdfauthor={\TheAuthors}
}
\fi

\begin{document}

\maketitle

% REQUIRED
\begin{abstract}
We~study the problem of simulating protocols in a quantum communication 
setting over noisy channels. This problem falls at the intersection of 
quantum information theory and quantum communication complexity, and it will be
of importance for eventual real-world applications of interactive 
quantum protocols, which can be proved to have exponentially lower 
communication costs than their classical counterparts for some problems. 
These are the first results concerning the 
quantum version of this problem, originally studied by Schulman in a classical 
setting (FOCS '92, STOC '93). We simulate a length \boldmath{$N$} quantum 
communication protocol by a length \boldmath{$O(N)$} protocol with arbitrarily small 
error. 
%Our simulation strategy has a far higher communication rate than 
%a naive one that encodes separately each particular round of communication to 
%achieve comparable success. 
%Such a strategy would have a 
%communication rate going to \boldmath{$0$} in the worst interaction case as the 
%length of the protocols increases, in contrast to our strategy, which has 
%a communication rate proportional to the capacity of the channel used. 
Under adversarial noise, our strategy can withstand, for arbitrarily 
small \boldmath{$\epsilon > 0$}, error rates as high as \boldmath{$1/2 -\epsilon$} when parties 
pre-share perfect entanglement, but the classical channel is noisy.
We show that this is optimal. 
%Note that in this model, the naive strategy
%would not work for any 
%constant fraction of errors. 
We provide extension of these results in several other models of communication, 
including when also the entanglement is noisy, 
and when there is no pre-shared entanglement but 
communication is quantum and noisy. We also study 
the case of random noise, for which we provide
simulation protocols with positive communication rates
and no pre-shared entanglement over some quantum channels with quantum capacity \boldmath{$C_Q=0$}, 
proving that \boldmath{$C_Q$} is in general not the right characterization of a channel's capacity 
for interactive quantum communication.
Our results are stated for a general quantum communication 
protocol in which Alice and Bob collaborate, and these results hold in 
particular in the quantum communication complexity settings of the Yao and 
Cleve--Buhrman models.
\end{abstract}

% REQUIRED

% REQUIRED

\section{Introduction}
Quantum information theory is well developed for information transmission 
over noisy quantum channels, dating back to the work of Holevo in the 
1970's~\cite{Hol72, Hol73} for the transmission of classical information 
\cite{Hol98, SW97} and quantum information \cite{Lloyd97, Shor02, Dev05}, 
and even for cases allowing pre-shared entanglement between sender and 
receiver \cite{BSST99, BSST02}. It~describes the ultimate limits for 
(unidirectional) data transmission over noisy quantum channels without 
concern for explicit, efficient construction of codes. Closely related is 
the area of quantum coding theory, which takes a more practical approach 
toward the construction of quantum error correcting codes 
\cite{Shor95, Ste95} by providing explicit and efficient constructions 
\cite{CS96, Ste95, Gott96, CRSS98} and by providing bounds on their 
existence \cite{CRSS98, FM04, Rai99}.

Quantum communication complexity has also been studied in depth since 
Yao's paper introduced the field in 1993~\cite{Yao93}. It is an 
idealized setting in  which local computation is deemed free and 
communication is noiseless but expensive. Two parties want to 
compute a classical function of their joint input while minimizing the 
number of qubits they have to exchange. Exponential separations have been 
shown for some promise problems between their classical and quantum 
communication complexity \cite{BCW98}---even in cases allowing bounded error 
\cite{Raz99}. Moreover, for both classical and quantum communication 
complexity, interaction has been proved to be a powerful resource: 
exponential separations in the communication complexity of some functions 
have also been established between protocols restricted to~$k$ messages, 
and protocols with $k + 1$ messages \cite{NW91, KNTZ07}. In 1997, Cleve and 
Buhrman \cite{CB97} defined an alternative model for communication 
complexity in a quantum setting, in which the players are allowed to 
pre-share an arbitrary entangled state but transmit classical rather than 
quantum bits. 
They proved the first separation between such a quantum model and the 
classical model of communication complexity (for a three party task).
This model is at least as powerful as Yao's (up 
to a factor of $2$), since entanglement can be used to teleport~\cite{BBCJPW93} 
the message qubits with twice as many classical bits. It 
is still open whether the two models are essentially equivalent, since no 
good bound on the amount of entanglement required in the Cleve--Buhrman 
model is known.

With the ubiquity of 
distributed computing nowadays, it has become increasingly important to 
develop an information and coding theory for interactive protocols. In 
the realm of classical communication, Schulman initiated the field with 
his pioneering works \cite{Sch92, Sch93, Sch96}, showing that it is 
possible to simulate any protocol defined over a noiseless channel with a 
noisy channel with exponentially small probability of error while only 
dilating the protocol by a constant factor. This multiplicative dilation 
factor, in the case of a binary symmetric channel, is proportional to the 
inverse of the capacity, as in the data transmission case. However, the 
hidden constant of proportionality does not go to $1$ asymptotically. For 
adversarial errors, Schulman also shows how to withstand corruption up to 
a rate of \smash{$\frac{1}{240}$}. Recent work by Braverman and Rao 
\cite{BR11} shows how to withstand error rates of $\smash{\frac{1}{4}} - 
\epsilon$ in the case of an adversarial channel, and they also show that this 
is optimal in their model of noisy communication. Even more recently, 
Franklin, Gelles, Ostrovsky, and Schulman \cite{FGOS12} were able to show 
that in an alternative model in which Alice and Bob are allowed to share 
a secret key unknown to the adversary Eve, they can withstand error rates 
up to $\smash{\frac{1}{2}} - \epsilon$, which is also shown to be optimal 
in this model. 

All of the above simulations use \emph{tree codes}, which 
were introduced by Schulman. Tree codes exist for various parameters, but 
no efficient construction is known. A relaxation of the tree code 
condition still strong enough for most applications in interactive coding 
was proposed by Gelles, Moitra and Sahai \cite{GMS12}, and they provided 
an efficient randomized construction for these so-called potent
tree codes. Using these in a random error model leads to efficient 
decoding on average hence to efficient simulation protocols (of 
course when given black-box access to the original protocol, which might be 
inefficient in itself). In a worst-case adversarial scenario, the 
decoding might still take exponential time with potent tree codes. It was 
only recently that an alternative coding strategy, developed by Brakerski 
and Kalai \cite{BK12}, was able to address the adversarial error case 
efficiently. Their strategy is to cleverly split the communication into 
blocks of logarithmic length in which tree encoding is used. In addition, 
they send, in between the blocks, some history information that enables 
efficient decoding. This construction was further improved by Brakerski 
and Naor \cite{BN13}. A survey article by Braverman \cite{Bra12} 
provides a good overview of results and open questions in the area of 
classical interactive communication circa 2011, though some of the 
important questions raised there have been addressed since. In 
particular, the question of interactive capacity of binary symmetric 
channels was recently investigated by Kol and Raz \cite{KR13}. For this 
channel they find that indeed, in the low noise regime, the communication 
capacity behaves differently in the asymptotic limit of long interactive 
protocols than in the data transmission case.

Quantum communication, even more so than classical communication, is 
prone to transmission errors in the real world. 
The approach taken in all of the above is inherently classical and does 
not generalize well to the quantum setting. In~particular, the fact that 
classical information can be copied and resent multiple times is 
implicitly used, and therefore the fact that the information in the 
communication register can be destroyed by noise is inconsequential. In 
contrast, the no-cloning theorem of quantum theory \cite{Dieks82, WZ82} 
rules out copying of quantum messages. As a result, if the information in 
some communication register is destroyed, it cannot be resent. A naive 
strategy, which applies in the quantum as well as the classical case, 
would be to encode each round separately. However, in a random error 
model, a constant dilation of each round would not be sufficient to 
achieve constant fidelity in the worst case of one-qubit transmission per 
round, and a super-constant dilation leads to a communication rate of 
zero asymptotically. Moreover, in the case of adversarial errors, no 
constant rate of error can be withstood with such a strategy unless the 
number of rounds is constant: the adversary can always disrupt a whole 
block. 

The properties of classical information made it possible for 
Schulman and his successors to design clever classical simulation 
protocols that can withstand constant error rates at constant 
communication rates and that can succeed in simulating classical protocols 
designed for noiseless channels over noisy channels by reproducing the 
whole transcript of the noiseless protocol. However, it was not immediately 
obvious 
that it is possible, given an arbitrary protocol designed for a noiseless 
bidirectional quantum channel, to simulate it over noisy quantum channels 
with constant error rate at a constant communication rate. Even for 
protocols in the Cleve--Buhrman model, in which the communication is 
classical, it is not clear whether we can achieve results similar to those 
for classical protocols. Indeed, a quantum measurement is in general 
irreversible. If such a measurement is performed on the shared entangled 
state and the players later realize that the measurement was based on 
wrong classical information, the naive adaptation of the classical 
simulation to the Cleve--Buhrman model fails.

\section{Overview of Results}

We~show that despite the above obstacles, it is indeed possible to 
simulate arbitrary quantum protocols over noisy quantum channels with 
good communication rates. We~consider two models for interaction over 
noisy channels. One is analogous to Yao's model, and all communication 
in it is over noisy quantum channels, but the parties do not pre-share 
entanglement. The other is analogous to the Cleve--Buhrman model, and all 
communication in it is over noisy classical channels and parties are 
allowed to pre-share noiseless entanglement. We~call these models the 
\emph{quantum\/} and \emph{shared entanglement models\/}, respectively. 
We also consider a further variation on the shared entanglement model in which entanglement is also noisy.

Our main focus is on the model with perfect shared entanglement but adversarial noise on the classical communication. 
In such a context, the number of errors is defined to be the Hamming distance between the transcript of sent messages 
and the transcript of possibly corrupted received messages. Messages are over a constant size alphabet, and the error rate 
is the ratio between the number of errors introduced by the adversary in the worst case and
the number of such messages sent, i.e.~the transcript length. Note that in this model, it is
possible for the honest parties to generate a secret key unknown to the adversary by measuring their shared entanglement. 
Details about the other models of communication appear in 
%Section~
\cref{sec:oth}.
Most of our technical contributions involve showing the following result,
which is stated more formally as Theorem~\ref{th:optach} later.
\begin{theorem}
A constant 
dilation factor on the communication suffices to withstand an
 adversarial error rate of $\frac{1}{2} - \epsilon$ in the shared
 entanglement model, for arbitrarily small $\epsilon > 0$. 
\end{theorem}
This is optimal and matches 
the highest tolerable error rate in the analogous shared secret key model 
for classical interactive communication \cite{FGOS12}.

The results in the other models are consequences of this main theorem. For the quantum communication model in which parties do not pre-share entanglement, but have access to a noisy quantum channel, we first distribute a linear amount of entanglement using standard quantum information and coding theory techniques. We can tolerate 
any adversarial error rate less than $\frac{1}{6}$ in that case (Theorem~\ref{th:advstand}),
close to the best achievable for quantum 
data transmission with zero error at $\frac{1}{4}$. 
This is better than the factor of two drop that might be expected if 
we compare classical interactive coding to unidirectional coding.
We can also adapt our techniques for an adversarial error model to the case of a random error model. Then, dilation factors proportional to $\frac{1}{C_Q}$ for a depolarizing channel of
quantum capacity $C_Q$ in the quantum model (Theorem~\ref{th:iidstand}), and proportional to $\frac{1}{C}$ for a binary symmetric channel of capacity $C$ in the shared entanglement model (Theorem~\ref{th:iidshent}), are sufficient. We also show that the result in the shared entanglement model is asymptotically optimal: there exists a family of binary functions for which a dilation factor proportional to $\frac{1}{C}$ is necessary (Theorem~\ref{th:iidshentopt}). We further extend the study in the shared entanglement model to consider noisy entanglement in the form of noisy Einstein-Podolsky-Rosen
(EPR) pairs in the so-called Werner states. For any non-separable Werner state, we give simulation protocols with linear noisy classical communication and noisy EPR pair consumption. Perhaps surprisingly, similar techniques can be used to show that 
the use of depolarizing channels in both directions enables the 
simulation to succeed whenever the quantum capacity with two-way 
classical communication, $C_Q^2$, is strictly positive (Theorem~\ref{th:iidstandopt}). 
For some range of the depolarizing parameter, $C_Q = 0$ but $C_Q^2 > 0$, so 
this proves that $C_Q$ does not characterize a quantum channel's 
capacity for interactive quantum communication.

Due to the use of tree codes, the protocols presented in this paper are 
not computationally efficient. However, it is possible to extend 
classical results on efficient interactive coding tolerating maximum error to noisy quantum 
communication. The representation of noisy protocols mentioned above is 
quite powerful and could be used to adapt classical 
results on computationally efficient interactive computation over 
adversarial channels \cite{BK12} and on the interactive capacity of
random noise channels \cite{KR13} to the quantum regime.

There are two main components that establish our main result.

\subsection{First Component: Teleportation and Active Rewinding}

First, we need to establish a framework for simulating 
quantum protocols over noisy channels. To avoid losing quantum information, the approach we take is to teleport 
\cite{BBCJPW93} the quantum communication register back and forth. When 
the register is in some party's possession, this party tries to evolve the 
simulation by applying one of his unitary operations in the noiseless protocol, or 
one of its inverses if he realizes at some point he applied it wrongly 
before. The important point is that all operations on the quantum 
registers are reversible, being a sequence of noiseless protocol 
unitary operators and random (but known) Pauli operators.
Of particular importance to our work is the notion of tree codes as 
introduced by Schulman, which the players use to transmit classical 
information.

As described in a recent paper on efficient interactive coding 
\cite{BN13}, the high-level logic of all solutions proposed until now for 
classical protocol simulation can be summarized as follows: the parties 
try to evolve the protocol, and if they later realize there has been some 
error, they try to go back to the point where they last agreed (in a 
protocol tree representation, this would be their least common ancestor). 
In our approach for quantum protocols, the parties try to follow roughly 
the same idea, but for two reasons are not able to do this passively. 
First, there is no underlying transcript (or protocol tree) that the 
parties try to synchronize, except for their wish to evolve the correct 
sequence of unitary operations. By the no-cloning theorem \cite{Dieks82, WZ82}, 
the parties cannot restart with a copy of the quantum information 
received up to some earlier point. Instead they have to actively rewind 
previous unitary operators and wrong teleportation decodings until a suitable 
point in the protocol. Second, when they try to synchronize in this 
manner, they actively teleport, potentially leading to more errors on the 
joint quantum register. 

An important ingredient in our 
simulation is the representation for noisy quantum protocols that we 
develop. As said before, in quantum protocols there is no direct 
analogue of a protocol tree representation that enables one to keep track exactly 
and explicitly of the evolution of the noiseless protocol simulation. 
The cleaned-up form \cref{eq:coll} of our representation provides 
in some sense a quantum analogue of a protocol tree representation. 
As the classical representation, it enables an exact and explicit 
assessment of the evolution of the noiseless protocol simulation,
as well as such an assessment of the departure from it due to noise.

At this point, it might look like we have reduced our problem
to the classical case, since the parties only transmit 
classical information---the teleportation measurement 
outcomes. This enables us to reuse tools from classical interactive 
coding, most notably tree codes, but the design of the quantum simulation 
protocol needs extra care. Unlike in the classical case, agreement 
by the two parties on a common classical transcript is not sufficient.
This transcript consists mostly of random teleportation measurement 
outcomes and is useless by itself. Additionally, we need to maintain a
joint quantum state that eventually evolves according to the original
protocol.

Once we realize the importance of teleportation in the context of noisy
communication, and carefully design the simulation protocol, 
it may not come as a surprise that the simulation incurs only a constant
factor overhead. The need for backtracking in the quantum simulation,
however, seems to impose serious constraints on the 
tolerable error rate. \emph{A priori\/} it is entirely unclear whether
we could hope to circumvent the low error tolerance seen in simulations
with backtracking.

\subsection{Second Component: Simulation via Blueberry Codes}

The second part of our main contribution is to develop 
the necessary techniques to prove that we \emph{can} tolerate an error rate 
as high as $\frac{1}{2} - \epsilon$. These techniques are indeed novel, 
and could be used to improve on previously known 
\emph{classical} results.

Indeed, all recent classical schemes tolerating high error rates have the 
property that the parties always go forward with the communication by 
using the tree structure of classical protocols. In comparison, in the 
original Schulman scheme based on tree codes there is some form of 
backtracking, due to which the scheme could only tolerate a much lower 
adversarial error rate of \smash{$\frac{1}{240}$}. This is due to the 
fact that in a protocol with backtracking \cite{Sch96}, fort he
simulation to succeed the fraction of 
good rounds, in which both players correctly decode the tree code 
transmission, must be higher than in a 
protocol that always goes forward by transmitting edges of a pointer 
jumping problem \cite{BR11, FGOS12}.
There also is some form of backtracking in the outer level of the 
computationally efficient protocol of Ref.~\cite{BK12}, thus limiting the
 overall error rate that can be tolerated to a fourth of that of the 
inefficient protocol used at the inner level.
Hence, computationally efficient protocols in the shared secret key
communication model prior to this work could only tolerate error rates
less than $\frac{1}{8}$~\cite{FGOS12}. In light of these results, it is
clear that previously used techniques would not suffice to tolerate
error rates as high as $\smash{\frac{1}{2}} -\epsilon$ for our protocol,
which requires backtracking.
The new techniques we develop are thus necessary.

To achieve higher error tolerance, we follow 
Ref.~\cite{FGOS12} and use a \emph{blueberry code\/} to effectively turn most 
adversarial errors into erasures. Concatenating such a code on top of a 
tree code yields a tree code with an erasure symbol. Since general 
transmission errors are twice as harmful as erasures for the tree code 
condition, which is stated in terms of Hamming distance, it was shown in 
Ref.~\cite{FGOS12} that if the error rate is below $\smash{\frac{1}{2}} - 
\epsilon$, then the large number of rounds in which both parties correctly 
decode a long enough prefix is sufficient to imply success of the simulation. 
Once again due to backtracking, this condition is not sufficient for our purpose
and in particular blueberry codes by themselves are not 
sufficient to improve error tolerance up to $\frac{1}{2}$ here. 
For us,
the number of rounds in 
which both parties correctly decode even the whole string could be high, 
but if these rounds alternate with rounds in which at least one of the 
parties makes a decoding error, then the protocol could stall, and 
simulation would fail. To circumvent this possibility, we need to
bound the number of rounds with bad tree code decoding. Previously known bounds on this \cite{Sch96} can be used
to show the success of our simulation but are far from enabling us to tolerate
error rates up to $\frac{1}{2}$.
We develop a new bound on tree codes with an erasure symbol, 
(see Lemma~\ref{lem:optcor}), which might be 
of independent interest for classical interactive coding. 
This bound enables us to tightly control the number of rounds with bad decoding. 
Once we control this quantity,
 it is also important to ensure that even when there is corruption 
detected as an erasure in a round, as long as there is no bad
 decoding, the protocol will not need to spend a good round to correct 
for this previous erasure round. 

In fact, the techniques that we develop are not just powerful enough to prove that
our quantum protocol can tolerate the maximum error rate of 
$\smash{\frac{1}{2}} - \epsilon$. 
%Lemma 
\Cref{lem:optcor} can be used to 
obtain a strengthening of the theorem of Ref.~\cite{FGOS12} in the \emph{classical} shared
 secret key model, and then our techniques can be applied with
 this strengthened theorem and the techniques of Ref.~\cite{BK12} to obtain 
computationally efficient simulation protocols in this model that can also
 tolerate any error rate less than $\frac{1}{2}$~\cite{Tou13}.
This demonstrates the power of our techniques.
However, this result has been superseded by slightly adapting a result from Ref.~\cite{GH14}, which uses different techniques; there
the authors
obtain computationally efficient simulation protocols at a maximum error 
up to $\frac{1}{4}$ in the model without a shared secret key.

%
%\emph{Organization}: 
\subsection{Organization}

The paper is structured as follows: in 
%section~
\cref{sec:not}, we set up the
notation and state the relevant definitions, in particular for 
the different models of communication. 
In 
%section~
\cref{sec:bas},
we state and prove a simpler version of our main result for the 
adversarial case in the shared entanglement model.
In 
%section~
\cref{sec:opt}, we state 
and prove our main result for the adversarial case in the shared 
entanglement model. 
%Section~
\Cref{sec:oth} shows how to adapt the result 
of the previous section to obtain various interesting results, in 
particular for the quantum model, for the noisy shared entanglement model,
and in the case of a random error model. 
We conclude with a discussion of our results and further research directions.

\section{Preliminaries}
	\label{sec:not}
	\label{sec:prel}

%------------------------------------------------------------
%------------------------------------------------------------
\subsection{Quantum Mechanics}

We briefly review the quantum formalism for finite dimensional systems, 
mainly to set notation; for a more thorough treatment, we refer the 
interested reader to the following good introductions in a quantum information theory 
context \cite[Chapter 2]{NC00}, \cite[Chapter 2]{Wat08} \cite[ Chapters 3, 4, 5]{Wilde11}.

\subsubsection{Quantum States and Quantum Evolution}

To every quantum system $A$ we 
associate a finite dimensional Hilbert space, which by abuse of notation 
we also denote by $A$. The state of quantum system $A$ is represented by 
a density operator $\rho^A$, a positive semi-definite operator over the 
Hilbert space $A$ with unit trace. We denote by $\D (A)$ the set of all 
density operators representing states of system $A$. Composite quantum 
systems are associated with the (Kronecker) tensor product space of the 
underlying spaces, i.e., for systems $A$ and $B$, the allowed states of 
the composite system $A \otimes B$ are (represented by) the density 
operators in $\D(A \otimes B)$. We sometimes use the shorthand $AB$ for 
$A \otimes B$. The evolution of a quantum system $A$ is represented by a 
completely positive, trace preserving linear map (CPTP map) $\N^A$ such 
that if the state of the system is $\rho \in \D(A)$ before evolution 
through $\N^A$, the state of the system is $\N^A (\rho) \in \D(A)$ after. 
If the system $A$ is clear from the context, we might drop the superscript.
We 
refer to such maps as quantum channels, and to the set of all channels 
acting on $A$ as $\L(A)$. An important quantum channel that we consider 
is the qubit depolarizing channel $\T_{\epsilon}$ with depolarizing 
parameter $0 \leq \epsilon \leq 1$: it takes as input a qubit $\rho$ and 
outputs a qubit $\T_{\epsilon}(\rho) = (1 - \epsilon) \rho + \epsilon 
\tfrac{\rI}{2}$, i.e., with probability $1 - \epsilon$ it outputs $\rho$, and 
with complementary probability $\epsilon$ it outputs a completely mixed 
state. We also consider quantum channels with different input and output 
systems; the set of all quantum channels from a system $A$ to a system 
$B$ is denoted $\L(A, B)$. 
An example of such a channel that we consider is the qubit erasing channel $\R_{\epsilon}$ with erasing 
parameter $0 \leq \epsilon \leq 1$: it takes as input a qubit $\rho$ and 
outputs a qutrit $\R_{\epsilon}(\rho) = (1 - \epsilon) \rho + \epsilon 
\kb{e}{e}$, i.e., with probability $1 - \epsilon$ it outputs $\rho$, and 
with complementary probability $\epsilon$ it outputs an orthogonal erasure flag~$\ket{e}$.
Another important operation on a composite 
system $A \otimes B$ is the partial trace $\Tr{B}{\rho^{AB}}$ which 
effectively derives the \emph{reduced\/} or marginal state of the $A$
subsystem from the quantum state $\rho^{AB}$.
Fixing an orthonormal basis $\{ \ket{i} \}$ for $B$, the 
partial trace is given by~$\Tr{B}{\rho^{AB}} = \sum_i (\rI \tensor \bra{i} )
\rho (\rI \tensor \ket{i})$, and this is a valid quantum channel in $\L(A\otimes B, 
A)$.
Note that the action of $\Tra{B}$ is independent of the choice of basis chosen to represent it,  so we unambiguously write $\rho^A = \Tr{B}{\rho^{AB}}$.

An important special case for quantum systems comprises pure states, whose 
density operators have a special form:  rank-one projectors 
$\kb{\psi}{\psi}$. In such a case, a more convenient notation is provided 
by the pure state formalism: a state is represented by the unit vector 
$\ket{\psi}$ (up to an irrelevant complex phase) upon which the density operator 
projects. We denote by $\H(A)$ the set of all such unit vectors (up 
to equivalence of global phase) in system $A$.

%
%\subsubsection{Quantum evolution.}
%

Pure state evolution is represented by a 
unitary operator $U^A$ acting on $\ket{\psi}^A$, denoted $U 
\ket{\psi}^A$. Evolution of the $B$ register of a state $\ket{\psi}^{AB}$ 
under the action of a unitary operator $U^B$ is represented by $(\rI^{A} \otimes 
U^{B})\ket{\psi}^{AB}$, for $\rI^{A}$ representing the identity operator 
acting on the $A$ system, and is denoted by the shorthand $U^{B} 
\ket{\psi}^{AB}$ for convenience. We occasionally drop the superscripts when the 
systems are clear from the context. The evolution under consecutive action of 
unitary operators $U_j$'s is denoted by
\begin{equation}
	\left( \prod_{j=1}^\ell  U_j \right) \ket{\psi} = U_\ell \dotsc U_1 \ket{\psi}.
\end{equation}

We represent a classical random variable $X$ with probability density 
function $p_X$ by a density operator $\sigma^X$ that is diagonal in a 
fixed (orthonormal) basis $\{ \ket{x} \}_{x \in \X}$: $\sigma^X = \sum_{x \in \X} 
p_X(x) \kb{x}{x}^X$. For a quantum system $A$ classically correlated with 
a random variable $X$, we represent the corresponding classical-quantum 
state by the density operator $\rho^{XA} = \sum_{x \in \X} p_X(x) 
\kb{x}{x}^X \otimes \rho_x^A$, in which $\rho_x^A$ is the state of system 
$A$ conditioned on the random variable $X$ taking value $x \in \X$. The 
extraction of classical information from a quantum system is represented 
by quantum instruments: classical-quantum CPTP maps that take 
classical-quantum states on a composite system $X \otimes A$ to 
classical-quantum states. Viewing classical random variables as a special 
case of quantum systems, quantum instruments can be viewed as a special 
case of quantum channels.

\subsubsection{Pauli Operators}

When considering a quantum system $A$ of dimension $q$, we fix
an orthonormal basis $\{ \ket{i} \}_{i \in \{0, 1, \dotsc, q-1 \}}$ for $A$ and use the
following generalizations of Pauli operators: for $j, k \in \{0, 1, \dotsc, q-1 \}$,
$\rX^j \ket{k} = \ket{(k+j) \mod{q}}$ and $\rZ^j \ket{k} = 
e^{\complexi 2 \pi \tfrac{jk}{q}} \ket{k}$.
The operators in the set $\{\rX^j \rZ^k \}_{j, k \in \{0, 1, q-1  \}}$ 
are known as the Heisenberg-Weyl operators and form a basis for the 
linear vector space of operators on $A$, and the operators in
\begin{equation}
\label{eq:FqN}
	\F_{q, N} = \{\rX^{j_1} \rZ^{k_1} \otimes \dotsb \otimes \rX^{j_N} \rZ^{k_N} \}_{
		j_\ell k_\ell \in \{0, 1, \dotsc, q-1 \}^2, \ell \in [N]}
\end{equation}
form a basis for the space of operators on $A^{\otimes N}$. For $E \in \F_{q, N}$,
we denote by $\wt{E}$ the weight of $E$, i.e., the number of $A$ subsystems
on which $E$ acts non-trivially. For~$\delta \in [0,1]$, the set
\begin{equation}
\label{eq:Edelta}
	\E_{\delta, q, N} = \{E \in \F_{q, N} : \wt{E} \leq \delta N \}
\end{equation}
is the subset of elements of $\F_{q, N}$ of weight less than or equal to $\delta N$.

\subsubsection{Teleportation}

Our simulation protocols make heavy use of the teleportation protocol 
between Alice and Bob \cite{BBCJPW93}, which uses the following resource 
state shared by Alice and Bob, called an EPR pair: $\ket{\Phi^+}^{T_\sA 
T_\sB} = \tfrac{1}{\sqrt{2}} (\ket{00} + \ket{11})$, with the qubit  in the 
$T_\sA$ register held by Alice, and the qubit in the $T_\sB$ register held by 
Bob. The teleportation protocol then uses one of these resource states to 
teleport one qubit either from Alice to Bob, or from Bob to Alice. If 
Alice wants to teleport a qubit $\ket{\psi}$ in the register $C$ to Bob, 
with whom she shares an EPR pair, she applies a joint Bell measurement, 
which can perfectly distinguish the Bell states $\{\ket{\Phi_{xz}} = 
\tfrac{1}{\sqrt{2}} (\ket{0x} + (-1)^z \ket{1\bar{x}}) \}_{x, z \in \{0, 1 
\}}$, to the registers $C T_\sA$ she holds, and obtains uniformly random 
measurement outcomes $xz \in \{0, 1 \}^2$. After this measurement, the 
state in the $T_\sB$ register is $\rX^x \rZ^z \ket{\psi}$, for $\rX$ and $\rZ$ the 
Pauli operators corresponding to bit flip and phase flip in the 
computational ($\rZ$) basis, respectively. If Alice transmits the two bits 
$xz$ to Bob, he can then decode the state $\ket{\psi}$ on the $T_\sB$ 
register by applying $(\rX^x \rZ^z)^{-1} = \rZ^z \rX^x$. Teleportation from Bob 
to Alice is performed similarly (EPR pairs are symmetric).

\subsubsection{Pseudo-Measurements}

Another technique we use is that of making classical operations coherent: 
measurements and classically controlled operations are replaced by 
corresponding unitary operators (and ancilla register preparation). We call the
coherent version of a measurement a \emph{pseudo-measurement\/}. Without
loss in generality, it suffices to consider the measurement of a single
qubit in the standard basis~$\set{\ket{0}, \ket{1}}$.
This measurement corresponds to the
instrument~$\N$ defined by~$\N(\rho) = \bra{0} \rho \ket{0} \, \kb{0}{0} +
\bra{1} \rho \ket{1} \, \kb{1}{1}$. We replace this with the action of
the CNOT operation~$\kb{0}{0} \tensor \rI + \kb{1}{1} \tensor \rX$
on the qubit and a fresh ancillary qubit prepared in state~$\ket{0}$,
i.e., with the CPTP map~$\N'$ defined by~$\N'(\rho) = U (\rho \tensor 
\kb{0}{0}) U^*$, where~$U$ is the CNOT operation. The ancilla qubit may
now be transmitted instead of sending the classical outcome of the
measurement~$\N$. Provided all further operations on the two
qubits are only controlled unitary operations (in which the two qubits
may only be control qubits), each separately behaves like the classical 
measurement outcome.
\suppress{
to obtain a binary classical outcome in $\{0, 1 \}$ 
with some probability $p_0$ and $p_1$, respectively, a classical value that 
can be distributed classically among two parties (or more), a 
pseudo-measure is applied to the quantum state, leaving it in a pure 
quantum state in which the two qubits will then act the same as if they 
had been measured, provided they do not further interact. The technique 
to do so uses a controlled-$\rX$ gate, i.e., a gate mapping $\ket{x}^S 
\ket{b}^T$ to $\ket{x}^S \ket{b \oplus x}^T$ for some source qubit  in 
register $S$ to be kept at the previous measurement point, and some 
target qubit in register $T$ to be distributed. Then a pseudo-measure is 
done by preparing a fresh ancilla qubit in state $\ket{0}$ in some 
register $T$, and if we relabel the register of the previously measured 
qubit by register $S$, we simply use the controlled-$\rX$ gate as described 
above. Then, if the $S$ and $T$ registers are left as is, a subsequent 
measurement on any one of these registers will still output $0$ or $1$ 
with probability $p_0$ and $p_1$, respectively, and if both registers are 
measured, both outcomes are perfectly correlated.
}
The advantage of this substitution is that unlike measurements, they 
are \emph{reversible\/}. If it is later realized that a qubit should not have 
been measured, the pseudo-measurement can be undone.

\subsubsection{Distance Measures}

To measure the success of the simulation, we use the trace distance $\| 
\rho - \sigma \|_1^A$ between two arbitrary states $\rho^A$ and 
$\sigma^A$, in which $\| O \|_1^A = \Tr{}{(O^\dagger O)^{\tfrac{1}{2}}}$ is 
the trace norm for operators on system $A$. We might drop the $A$ 
superscript if the system is clear from the context. The trace distance has the 
operational interpretation  to be (four times) the best possible bias to 
distinguish between the two states $\rho^A$ and $\sigma^A$, given a 
single unknown copy of one of these two states~\cite[Chapter 3]{Wat08}. To distinguish between 
quantum channels, we first consider the induced norm for quantum channels 
$\N \in \L(A, B)$: $\|\N \| = \max{ \{\|\N(\sigma) \|_1^B : \sigma \in 
\D(A) \}}$. Correlations with another quantum system can help distinguish 
between quantum channels, so an appropriate norm to use to account for 
this is the completely bounded trace norm \cite{AKN97}: $\| \N \|_\diamond = \|\N 
\otimes \rI^R \|$ for some reference system $R$ of the same dimension as 
the input system $A$~\cite[Chapter 3]{Wat08}. For two quantum channels $\N$, $\M \in \L(A, 
B)$, $\| \N - \M \|_\diamond$ has a useful operational interpretation:
it is (four times) the best possible bias with which we can identify 
a uniformly random (unknown) channel out of the two, when we are allowed only
one use of the channel.

%------------------------------------------------------------
%------------------------------------------------------------
\subsection{Quantum Communication Model} 
	\label{sec:qucomm}

%------------------------------------------------------------
\subsubsection{Noiseless Communication Model}
	\label{sec:nslss}

In the \emph{noiseless quantum communication model} that we want to  
simulate, there are five quantum registers: the $A$ register held by 
Alice; the $B$ register held by Bob; the $C$ register, which is the 
communication register exchanged back and forth between Alice and Bob and 
initially held by Alice,  the $E$ register held by a potential adversary Eve; and finally the $R$ register, a reference system which purifies the 
initial (and then also the final) state of the $ABCE$ registers. The initial state 
$\ket{\psi_\mathrm{init}}^{ABCER} \in \H (A \otimes B \otimes C \otimes E \otimes R)$ is chosen arbitrarily 
from the set of possible inputs and is fixed at the outset 
of the protocol, but it is possibly unknown (totally or partially) to Alice and 
Bob. Note that to allow for composition of quantum protocols in an 
arbitrary environment, we consider arbitrary quantum states as input, 
which may be entangled with systems $RE$. A protocol $\Pi$ is then 
defined by the sequence of unitary operations $U_1, U_2, \dotsc , U_{N + 
1}$, with $U_{i}$ for odd~$i$ known at least to Alice (or given to her 
in a black box) and acting on registers~$AC$, and $U_{i}$ for even~$i$ 
known at least to Bob (or given to him in a black box) and acting on
registers~$BC$. For simplicity, we assume that $N$ is even. We can
modify any protocol to satisfy this property, while increasing
the total cost of communication by at most one communication of the $C$ 
register. 
The unitary operations of protocol $\Pi$ can be assumed to be public information
and known to Eve.
On a particular input state $\ket{\psi_\mathrm{init}}$, the 
protocol generates the final state $\ket{\psi_\mathrm{final}}^{ABCER} = 
U_{N + 1} \cdots U_1 \ket{\psi_\mathrm{init}}^{ABCER}$, for which at the 
end of the protocol the $A$ and $C$ registers are held by Alice, the $B$ 
register is held by Bob, and the $E$ register is held by Eve. The reference register $R$ is left untouched throughout the protocol.
The output state of the protocol is the $ABC$ part, i.e.,  $\Pi (\ket{\psi_\mathrm{init}}) =
\Tr{ER}{\kb{\psi_\mathrm{final}}{\psi_\mathrm{final}}^{ABCER}}$, and by  a slight
abuse of notation we also represent the induced quantum channel from $ABCE$ 
to $ABC$ simply by $\Pi$. This is depicted in Figure~\ref{fig:int_mod}.
Note that while the protocol only acts on $ABC$, we wish to maintain correlations with the reference system $R$, while we simply disregard
what happens on the $E$ system assumed to be in Eve's hand.
Since we consider local computation to be free, 
the sizes of $A$ and $B$ can be arbitrarily large, but still of finite 
size, say $m_A$ and $m_B$ qubits, respectively. We restrict ourselves to
the case of a single-qubit communication register~$C$, which is the worst 
case for noisy interactive communication. Every protocol can be converted 
into such a form by increasing the communication by a factor of at most two
but possibly at the expense of much more interaction: if a party has to speak 
when it is not his turn, he sends a qubit in state $\ket{0}$.
Note that both the Yao and the
Cleve--Buhrman models of quantum communication complexity can be recast 
in this framework; see Section~\ref{sec:qucc}. 

		\begin{figure}
		\begin{overpic}[width=1\textwidth]{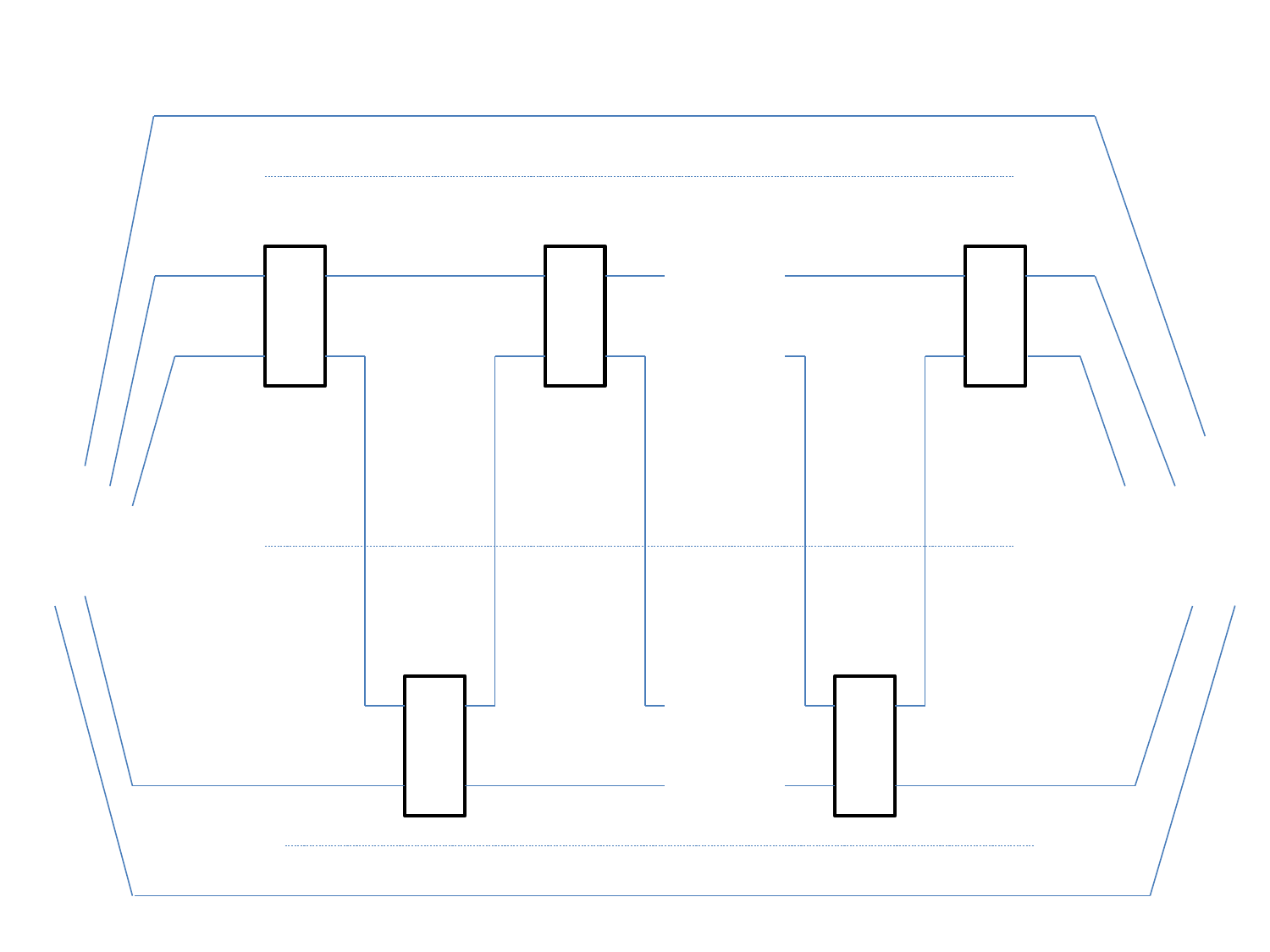}
		  \put(0,66.5){Reference}
		  \put(0,45){Alice}
		  \put(0,22){Bob}
		  \put(0,4.8){Eve}
		  \put(6,31){ $\ket{\psi_{\mathrm{init}}}$}
		  \put(15,66.5){\footnotesize$R$}
		  \put(15,4.8){\footnotesize$E$}
		  \put(15,54){\footnotesize$A$}
		  \put(15,48){\footnotesize$C$}
		  \put(15,13.7){\footnotesize$B$}
		  \put(22.1,49.5){\footnotesize$U_1$}
		  \put(26.2,54){\footnotesize$A$}
		  \put(26.2,48){\footnotesize$C$}
		  \put(33,15.5){\footnotesize$U_2$}
		  \put(37.2,17.4){\footnotesize$C$}
		  \put(37.2,13.7){\footnotesize$B$}
		  \put(44.2,49.5){\footnotesize$U_3$}
		  \put(48.2,54){\footnotesize$A$}
		  \put(48.2,48){\footnotesize$C$}
		  \put(55,33){\footnotesize$\cdots$}
		  \put(59.5,54){\footnotesize$A$}
		  \put(59.5,48){\footnotesize$C$}
		  \put(59.5,13.7){\footnotesize$B$}
		  \put(66.2,15.5){\footnotesize$U_{M}$}
		  \put(73.5,23){\footnotesize$C$}
		  \put(73.5,13.7){\footnotesize$B$}
		  \put(77.3,49.5){\footnotesize$U_{f}$}
		  \put(81.3,54){\footnotesize$A$}
		  \put(81.3,48){\footnotesize$C$}
		  \put(92,31){ $\ket{\psi_{\mathrm{final}}}$}
		\end{overpic}
		  \caption{Depiction of a quantum protocol in the noiseless communication model, adapted from the long version of~\cite[Figure 1]{Tou15}.}
		  \label{fig:int_mod}
		\end{figure}

We later embed length $N$ protocols into others of larger length 
$N^\prime > N$. To perform such \emph{noiseless protocol embedding}, we 
define some dummy registers $\tilde{A}$, $\tilde{B}$, $\tilde{C}$ isomorphic 
to $A$, $B$, $C$, respectively. $\tilde{A}$ and $\tilde{C}$  are part of 
Alice's scratch register and $\tilde{B}$ is part of Bob's scratch 
register. Then, for any isomorphic quantum registers $D$, $\tilde{D}$, let 
SWAP$_{D \leftrightarrow \tilde{D}}$ denote the unitary operation that 
swaps the $D, \tilde{D}$ registers. Recall that~$N$ is assumed to be
even. In a noiseless protocol embedding, 
for $i \in \{1, 2, \dotsc, N-1 \}$, we leave $U_i$ untouched. We replace 
$U_{N}$ by (SWAP$_{B \leftrightarrow \tilde{B}} U_{N})$ and $U_{N + 
1}$ by (SWAP$_{AC \leftrightarrow \tilde{A} \tilde{C}} U_{N + 1})$. 
Finally, for $i \in \{N+2, N+3, \dotsc, N^\prime + 1 \}$, we define $U_i = 
\rI$, the identity operator.
This embedding is important in the setting of interactive quantum coding for the following reasons:
first, adding these $U_i$ for $i > N$ makes the protocol well defined for $N^\prime +1$ steps.
Then,  swapping the important registers into the safe registers $\tilde{A}$, $\tilde{B}$, $\tilde{C}$
ensures that the important registers are never affected by noise arising after the first $N+1$ steps have been applied.
Hence, in our simulation, as long as we succeed in implementing the first $N+1$ steps without errors, the simulation will succeed since the
$\tilde{A}$, $\tilde{B}$, $\tilde{C}$ registers will then contain the output of the simulation, with no error acting on these registers.

We refer later to the \emph{unidirectional model}, consisting of one-way protocols; in this noiseless 
model, we allow for large local registers $A^\prime$, $B^\prime$ and for a 
large communication register $C^\prime$ that is used only once, either 
from Alice to Bob or from Bob to Alice, depending on the protocol. These 
registers can be further decomposed such that when used for simulation, 
the $A$ and $C$ registers of the protocol to be simulated are subsystems of 
$A^\prime$, and $B$ is one of $B^\prime$. 
%We also 
%allow for classical registers $X$, $Y$ held by Alice and Bob, respectively. 
For concreteness we consider here the case of communication from Alice to 
Bob; the other case is symmetric. A simulation protocol $U$ in the 
unidirectional model is defined by two quantum instruments $\M_1^{
A^\prime C^\prime}$, $\M_2^{ B^\prime C^\prime}$, and the output of the 
protocol on input $\ket{\psi} \in \H (A \otimes B \otimes C \otimes E \otimes R)$ 
is the state of the $ABC$ subsystem of $\M_2 \M_1 (\ket{\psi})$ and is denoted 
$U(\ket{\psi})$. By abuse of notation, the induced quantum channel from 
$ABCE$ to $ABC$ is also denoted $U$.

\subsubsection{Noisy Communication Model}

There are many possible models for noisy communication. We consider two 
in particular: one analogous to the Yao model with no shared entanglement 
but noisy quantum communication, which we call the \emph{quantum model}, 
and one analogous to the Cleve--Buhrman model with noiseless pre-shared 
entanglement but noisy classical communication, which we call the 
\emph{shared entanglement model}. A further variation on the shared 
entanglement model in which the entanglement is also noisy is considered 
in 
%section 
\cref{sec:nsent}. For simplicity, we formally define in this 
section what we sometimes refer to as \emph{alternating} communication models, 
in which Alice and Bob take turns transmitting the 
communication register to each other, and this is the model in which most of our protocols are 
defined. Our definitions easily adapt to somewhat more general models 
which we call \emph{oblivious} communication models, following Ref.~\cite{BR11}.
In these models, Alice and Bob do not necessarily 
transmit their messages in alternation, but nevertheless in a fixed order and
of fixed sizes
known to all (Alice, Bob, and Eve) depending only on the round and not on 
the particular input or the actions of Eve. 
Communication models with a dependence on inputs or 
actions of Eve are called \emph{adaptive\/} communication models.

\paragraph{Quantum Model}

We give formal definitions for the quantum model in Appendix~\ref{sec:nqucommqu}. Let us give an informal description here.

In the \emph{quantum model}, Alice has workspace $A^\prime$, Bob has workspace $B^\prime$, adversary Eve has workspace $E^\prime$, and there is some quantum communication register $C^\prime$ of some fixed size $q$, exchanged back and forth between them $N^\prime$ times, passing through Eve's hand each time. Alice and Bob can perform arbitrary local processing between each transmission, whereas Eve's processing when the $C^\prime$ register passes through her hand is limited by the noise model as described below. The input registers $ABCE$ are shared between Alice ($AC$), Bob ($B$) and Eve ($E$) and the output registers $\tilde{A} \tilde{B} \tilde{C}$ are shared between Alice ($\tilde{A} \tilde{C}$) and Bob ($\tilde{B}$). The reference register $R$ containing the purification of the input is left untouched throughout. Alice and Bob also possess registers $C_\sA$ and $C_\sB$, respectively, acting as virtual communication register $C$ from the original protocol $\Pi$ of length $N$ to be simulated. The communication rate of the simulation is given by the ratio $\frac{N}{N^\prime \log q}$.

We are interested in two models of errors, adversarial and random noise. In the \emph{adversarial} noise model, we are mainly interested in adversary Eve with a bound $\delta N^\prime$ on the number of errors that she introduces on the quantum communication register $C^\prime$ that passes through her hand. The fraction $\delta$ of corrupted transmissions is called the error rate, and is assessed by requiring that there exists a representation of the global action of Eve on the $N^\prime$ quantum communication registers with Kraus operators of weight at most $\delta N^\prime$.

In the random noise model, we consider $N^\prime$ independent and identically distributed uses of a noisy quantum channel acting on register $C^\prime$, half the time in each direction. Eve's workspace register $E^\prime$ (including her input register $E$) can be taken to be trivial in this noise model. 

For both noise models, we say that the simulation succeeds with error $\epsilon$ if for any input, the output in register $\tilde{A} \tilde{B} \tilde{C}$ corresponds to that of running protocol $\Pi$ on the same input, while also maintaining correlations with system $R$, up to error $\epsilon$ in trace distance.

\paragraph{Shared Entanglement Model}

We give formal definitions for the shared entanglement model in Appendix~\ref{sec:nqucommsh}. Let us give an informal description here.

In the \emph{shared entanglement model}, Alice has workspace $A^{ \prime}$, Bob has workspace $B^{ \prime}$, adversary Eve has workspace $E^{ \prime}$, and there is some classical communication register $C^{ \prime \prime}$ of some fixed size $q$, exchanged back and forth between them $N^\prime$ times, passing through Eve's hand each time. Alice and Bob also pre-share noiseless entanglement in register $T_\sA T_\sB$. Alice and Bob can perform arbitrary local processing between each transmission, whereas Eve's processing when the $C^{ \prime \prime}$ register passes through her hand is limited by the noise model as described below. The input registers $ABCE$ are shared between Alice ($AC$), Bob ($B$) and Eve ($E$) and the output registers $\tilde{A} \tilde{B} \tilde{C}$ are shared between Alice ($\tilde{A} \tilde{C}$) and Bob ($\tilde{B}$). The reference register $R$ containing the purification of the input is left untouched throughout. Alice and Bob also possess registers $C_\sA$ and $C_\sB$, respectively, acting as virtual communication register $C$ from the original protocol $\Pi$ of length $N$ to be simulated. The communication rate of the simulation is given by the ratio $\frac{N}{N^\prime \log q}$.

We are interested in two models of errors, adversarial and random noise. In the \emph{adversarial} noise model, we are mainly interested in an adversary Eve with a bound $\delta N^\prime$ on the number of errors that she introduces on the classical communication register $C^{\prime \prime}$ that passes through her hand. The fraction $\delta$ of corrupted transmissions is called the error rate, and is assessed by requiring that the global action of Eve on the $N^\prime$ classical communication registers introduces errors of Hamming weight at most $\delta N^\prime$.

In the random noise model, we consider $N^\prime$ independent and identically distributed uses of a noisy classical channel acting on register $C^{\prime \prime}$, half the time in each direction. Eve's workspace register $E^{ \prime}$ (including her input register $E$) can be taken to be trivial in this noise model. 

For both noise models, we say that the simulation succeeds with error $\epsilon$ if for any input, the output in register $\tilde{A} \tilde{B} \tilde{C}$ corresponds to that of running protocol $\Pi$ on the same input, while also maintaining correlations with system $R$, up to error $\epsilon$ in trace distance.

Notice that adversaries in the quantum model and shared entanglement model are incomparable. In the quantum model, the adversary can inject fully quantum errors since the messages are quantum, while errors in the shared entanglement model are restricted to be modifications of classical symbols. On the other hand, in the shared entanglement model the adversary can read all the classical messages without the risk of corrupting them, whereas in the quantum model, any attempt to ``read'' messages will result in an error in general on some quantum message.

%------------------------------------------------------------
%------------------------------------------------------------
\subsection{Quantum Communication Complexity} 
	\label{sec:qucc}

We discuss how standard models for quantum communication complexity fit into our model for noiseless quantum communication.
In the Yao model for quantum communication complexity \cite{Yao93}, Alice 
is given a classical input $x \in X$ and Bob is given a classical input $y \in Y$, 
and they want to compute a classical function $f: X \times Y \rightarrow 
Z$  of their joint input (often $X = Y = \{0, 1 \}^n, Z = \{0, 1 \}$) by 
communicating as few quantum bits as possible, but without regard to the 
local computation cost. Often, we are only interested in $x \in X$, $y \in 
Y$ satisfying some promise $P : X \times Y \rightarrow \{0, 1 \}$. A 
global quantum system is split into three subsystems: the $A$ register 
held by Alice, the $B$ register held by Bob, and 
the $C$ register, which is the communication register initially held by Alice
and exchanged back and forth by Alice and Bob in each round. Our formal 
description of the protocols in this model is based upon the one given in 
Ref.~\cite{Kre95}. 

A length $N$ protocol is defined by a sequence of unitary operators $U_1$, $\dotsc$, $ 
U_{N + 1}$ in which for $i$ odd, $U_i$ acts on the $AC$ register, and for 
$i$ even, $U_i$ acts on the $BC$ register. 
We need $N+1$ unitary operators in order to have $N$ messages 
since a first unitary operation is applied before the first message is sent and a last one is
applied after the final message is received.
Initially, all the qubits in 
the $A$, $B$, $C$ registers are set to the all $\ket{0}$ state, except for 
$n$ qubits in the $A$ register initially set to $x \in X$, and $n$ in the 
$B$ register set to $y \in Y$. The number of qubits $m_A$, $m_B \in 
\mathbb{N}$ in the $A$ and $B$ registers is arbitrary (of course, $m_A$, $
m_B \geq n$) and is not taken into account in the cost of the protocol.
The complexity of the $U_i$'s is also immaterial, since local computation is 
deemed free. However, the number of qubits $c$ in the $C$ register is 
important and is taken into account in the communication cost, which is 
$N \cdot c$. The outcome of the protocol is obtained by measuring 
an appropriate number of qubits of registers $A$ and~$B$ of Alice
and Bob, respectively, after the application of $U_{N + 1}$. The 
protocol succeeds if the outcomes of both measurements equal $f(x, y)$ with 
good probability, usually required to be a constant greater than~$1/2$,
for any $x$, $y$ satisfying the promise.

Another model  for quantum communication complexity was introduced 
by Cleve and Buhrman~\cite{CB97}. In their model, communication is 
classical, but parties are allowed to pre-share an arbitrary entangled 
quantum state at the outset of the protocol. We can view protocols in 
this model as a modification on those of Yao's model in which the initial 
state $\ket{\psi}$ on the $ABC$ register is arbitrary except for $n$ 
qubits in each of the  $A$, $B$ registers initialized to $x$, $y$,
respectively. Also, each qubit in the $C$ register is measured in the 
computational basis, and it is the outcome of these measurements that is 
communicated to the other party. Note that by using pseudo-measurements 
instead of actual measurements in each round, the parties can use quantum 
communication instead of classical communication. Then the two models 
become almost identical, except for the initial state, which is arbitrary 
in the Cleve--Buhrman model, and fixed to the all $0$ state in the Yao model 
(not including each party's classical input). Since our simulation 
protocols consider general unitary local processing but do not assume any 
particular form for the initial state, they work on this slight 
adaptation of the Cleve--Buhrman model as well as on the Yao model of 
quantum communication complexity.

Hence, both the Yao and the
Cleve--Buhrman models of quantum communication complexity can be recast 
in our framework for noiseless communication by making all operations coherent: put the initial 
classical registers into quantum registers, replace classically 
controlled operations by quantumly controlled operations, also replace 
measurements by pseudo-measurements, and then replace any classical 
communication by quantum communication. In particular, this gets rid of 
the problem of the non-reversibility of measurements
in the Cleve--Buhrman model.

%------------------------------------------------------------
%------------------------------------------------------------
\subsection{Classical Communication}
	\label{sec:clcomm}

\subsubsection{History}

Our simulation protocols contain an important classical component. In our 
setting, we are interested in protocols in which each party sends a 
message from some message set $[d] = \{1, 2, \dotsc, d-1, d \}$ of size 
$d$ in alternation, for some fixed number of rounds $N^\prime$ (actually, 
$\tfrac{N^\prime}{2}$ in our protocols). A round consists of Alice sending 
a message to Bob and then Bob sending a message back. Parties only have 
access to some noisy channels, so they need to encode these messages in 
some way. The codes used to do so in an interactive setting are described 
in the next subsection. For the moment, let us focus on the 
messages the parties wish to transmit, without the coding.

In round $i$, Alice transmits a message $a_i \in [d]$ to Bob, and then 
Bob sends back a message $b_i \in [d]$. These messages depend on the 
messages $a_1$, $a_2$, $\dotsc$, $a_{i-1} \in [d]$ and $b_1$, $b_2$, $
\dotsc$, $b_{i-1} \in~[d]$ that Alice and Bob sent in the previous rounds, 
respectively. We refer to these sequences of 
messages (at the end of round $i$) as Alice's history $s_\sA = a_1 \cdots a_i 
\in [d]^i$ and Bob's history $s_\sB = b_1 \cdots b_i \in [d]^i$, 
respectively. Note that these histories are updated in each round, and that 
each history, at the end of round $i$, can be represented as a node at 
depth $i$ in some $d$-ary tree of depth $N^\prime$. This tree is called a 
history tree. The whole (noiseless) communication can be extracted from the 
information in these two histories.

When the communication is noisy, in some rounds one party makes errors 
when trying to determine the other party's history. When comparing the 
history $s = s_1 \cdots s_i \in [d]^i$ of a party in round $i$ of the
protocol without coding, with the other party's best guess $s^i = s_1^i 
\cdots s_i^i \in [d]^i$ for that history, 
% based on the communication he received up to that point,
the least 
common ancestor of $s$ and $ s^i$ is the node at depth $i-\ell$ such that 
$s_1 \cdots s_{i-\ell} = s_1^i \cdots s_{i-\ell}^i$ but $s_{i-\ell+1} 
\not= s_{i-\ell+1}^i$. We call $\ell$ the \emph{magnitude} of the error 
of such a guess $s^i$, and in general for two histories $s, s^i \in [d]^i$ 
satisfying the above (with least common ancestor at depth $i-\ell$) we 
write $L(s, s^i) = \ell$. Note that we can compute $\ell$ as $i - \max{ 
\{t: (\forall j \leq t) [s_j = s_j^i] \}}$.

%------------------------------------------------------------
%\subsection{Online Classical Codes}

%------------------------------------------------------------
\subsubsection{Tree Codes}
	\label{sec:tc}

Standard error correcting codes are designed for data transmission and 
therefore are not particularly well suited for interactive communication 
over noisy channels. In his breakthrough papers 
% on interactive communication
\cite{Sch93, Sch96}, Schulman defined tree codes, which are 
particular codes designed for such interactive communication. Indeed, 
these tree codes can perform encoding and decoding round by round (following 
Ref.~\cite{FGOS12}, we refer to such codes as online codes), such that for 
each round, a message from the message set $[d]$ is transmitted, but even 
if there is some decoding error in this round, for each additional round 
that we perform (without transmission error), the more likely it is that this 
previous decoding error is correctly decoded. We describe this self-healing property 
in more detail after formally defining tree codes. We use the 
following for our definition. Given a set $A$ and its $k$-fold Cartesian 
product $A^k = A \times \cdots \times A$ ($k$-times), we denote, for any 
$n \in \mathbb{N}$, $A^{\leq n} = \cup_{k=1}^n A^k$. Also, given a 
transmission alphabet $\Sigma$ and two words $\bar{e} = e_1 \cdots e_t 
\in \Sigma^t$ and $\bar{e}^\prime = e_1^\prime \cdots e_t^\prime \in 
\Sigma^t$ over this alphabet, we denote by $\Delta(\bar{e}, 
\bar{e}^\prime)$ (the Hamming distance) the number of different symbols, 
i.e., $\Delta(\bar{e}, \bar{e}^\prime) = |\{i: e_i \not= e_i^\prime \}|$.

\begin{definition}
	\emph{(Tree codes \cite{Sch96})} Given a message set $[d]$ of size 
	$d > 1$, a number of rounds of communication $N^\prime \in 
	\mathbb{N}$, a distance parameter $0 < \alpha < 1$ and a 
	transmission alphabet $\Sigma$ of size $|\Sigma| > d$, a $d$-ary 
	\emph{tree code} of depth $N^\prime$ and distance parameter 
	$\alpha$ over alphabet $\Sigma$ is defined by an encoding 
	function $\msE: [d]^{\leq N^\prime} \rightarrow \Sigma$, and a
decoding function~$\msD: \Sigma^{\leq N^\prime} \rightarrow [d]^{\leq
N^\prime}$.

Let~$\bar{\msE}: [d]^{\leq N^\prime} \rightarrow 
\Sigma^{\leq N^\prime}$ denote the extension of~$\msE$ to strings, i.e., for
any~$t \leq N^\prime$ and~$a = a_1 \cdots a_t \in [d]^t$,
\begin{align*}
%       (\forall t \leq N^\prime) & (\forall a = a_1 \cdots a_t \in [d]^t)
%       \notag \\
%       [\bar{\msE}(a) = \msE(a_1) \msE(a_1 a_2) & \cdots \msE(a_1 \cdots a_{t-1})
%       \msE(a_1 \cdots a_t) \in \Sigma^t].
        \bar{\msE}(a) = \msE(a_1) \, \msE(a_1 a_2)  \cdots \msE(a_1 \cdots a_{t-1})
        \, \msE(a_1 \cdots a_t) \enspace,
\end{align*}
which is a string in~$\Sigma^t$.

The encoding function satisfies the following distance property, called 
the tree code property. For any~$t \le N^\prime$, and~$a, a^\prime \in
[d]^t$,
\begin{align*}
%	(\forall t \leq N^\prime) & (\forall a, a^\prime \in [d]^t) \notag \\
%	[L(a, a^\prime) = \ell & \rightarrow  \Delta(\bar{e}, \bar{e}^\prime) 
%	%= \Delta( \bar{e}_{t-\ell+1} & \cdots \bar{e}_t, 
%	\bar{e}_{t-\ell+1}^\prime \cdots \bar{e}_t^\prime) 
%	\geq \alpha \cdot \ell],
	L(a, a^\prime) = \ell  \quad \implies \quad \Delta(\bar{\msE}(a), 
	\bar{\msE}(a^\prime)) \geq \alpha \cdot \ell \enspace.
\end{align*}
In other words, if the least common ancestor of~$a,a^\prime$ is at
depth~$t-\ell$, then the corresponding codewords are at distance at
least~$\alpha \ell$.

The decoding function satisfies the property that for any~$t
\leq N^\prime$, and~$\bar{e} \in \Sigma^t$,
\begin{align*}
%	(\forall t \leq N^\prime) & (\forall \bar{e}^\prime \in \Sigma^t) 
%	\notag \\
%	[\msD(\bar{e}^\prime) \in \{a: a \in [d]^t & \text{\ minimizes\ } 
%	\Delta(\bar{\msE}(a), \bar{e}^\prime) \}].
	\msD(\bar{e}) \in \{a: a \in [d]^t  \text{\ minimizes\ } 
	\Delta(\bar{\msE}(a), \bar{e}) \} \enspace.
\end{align*}
\end{definition}

See Appendix~\ref{sec:apptreecode} for a depiction of tree codes.

We later consider decoding of tree codes with an erasure symbol~$\perp$,
that is not used by the encoding function, but may occur in the output
of a channel. The decoding algorithm extends \emph{verbatim\/} to
received words with erasure symbols: it outputs a message sequence 
whose tree encoding is closest in Hamming distance to the received word.

Note that the decoding function is not uniquely defined for a given tree 
code: we could avoid ambiguity by outputting a special failure symbol for 
$\msD(\bar{e})$ whenever $|\{a: a \in [d]^t  \mathrm{\ minimizes\ } 
\Delta(\bar{\msE}(a), \bar{e}) \}| > 1$. Also note that we can view 
tree codes in the following alternative way, connecting them with the history 
tree representation defined above. Starting with a history tree, we can 
label the arcs out of each node by a symbol from $\Sigma$ corresponding 
to the encoding of that path in the tree code. The encoding function 
$\bar{\msE}$ represents the concatenation of the symbols on the path from 
root to node $a$, and the distance property is related to the distance of 
$a, a^\prime$ to their least common ancestor in the history tree, and to 
the number of errors during these corresponding $L(a, a^\prime)$ last 
transmissions. The following was proved in Ref.~\cite{Sch96} for the 
existence of tree codes.
Let $\rH(\alpha) = -\alpha \cdot \log{\alpha} - 
	(1-\alpha) \cdot \log{(1-\alpha)} $ denote the binary entropy 
	function.
\begin{lemma}
	\label{lem:tccode}
	Given a message set $[d]$ of size $d > 1$, a number of rounds of 
	communication $N^\prime \in \mathbb{N}$, and a distance parameter $0 < 
	\alpha < 1$, taking transmission alphabet $\Sigma$ with $|\Sigma| 
	= 2 \lfloor ( 2 \cdot 2^{\rH(\alpha)} \cdot d 
	)^{\tfrac{1}{1-\alpha}} \rfloor -1$ suffices to label the arcs of 
	some tree code, i.e., there exists an encoding function $\msE$ 
	satisfying the tree code property, and the required alphabet size 
	is independent of $N^\prime$, the number of rounds of 
	communication.
\end{lemma}

In fact, the result due to Schulman is even stronger: there exists an 
unbounded depth tree code with $\Sigma$ of the size discussed above. This 
stronger result could be useful in the case in which the number of rounds 
$N^\prime$ is not bounded at the beginning of the protocol, and it has been 
used to authenticate streams of classical data in Ref.~\cite{FGOS12}.

The distance property of tree codes ensures the following: if in 
round $t$ the decoding is good for the first $t-\ell$ messages sent 
$(\ell \geq 0)$, but wrong for the message sent in round $t-\ell+1$ (and 
possibly also for some other messages), then the re-encoding of the 
sequence of decoded messages must be distinct from the transmitted one in 
at least $\alpha \cdot \ell$ positions in the last $\ell$ rounds. Then, 
incorrect decoding (i.e., decoding to a message different from the one
encoded) implies that there were at least $\tfrac{1}{2} 
\cdot \alpha \cdot \ell$ transmission errors during those rounds, 
independent of what was sent in the first $t-\ell$ rounds. More 
precisely, given a transmitted message $\bar{a} \in [d]^t$, encoded as 
$\bar{e} = \bar{\msE}(\bar{a}) \in \Sigma^t$, received as $\bar{e}^{\prime 
\prime} \in \Sigma^t$, and decoded as $\bar{a}^\prime = \msD(\bar{e}^{\prime 
\prime}) \in [d]^t$, with $\bar{e}^\prime = \msE(\bar{a}^\prime)$, if we 
have $a_1 \cdots a_{t-\ell} = a_1^\prime \cdots a_{t-\ell}^\prime$ but 
$a_{t-\ell+1} \not= a_{t-\ell+1}^\prime$, i.e., $L(a, a^\prime) = \ell$, 
then $\Delta(\bar{e}, \bar{e}^\prime) \geq \alpha \cdot \ell$ and 
$\Delta(e_{t-\ell+1} \cdots e_t, e_{t-\ell+1}^{\prime \prime} 
\cdots e_t^{\prime 
\prime}) \geq \tfrac{1}{2} \cdot\alpha \cdot \ell$. (Note that $e_1 \cdots 
e_{t-\ell} = e_1^\prime \cdots e_{t-\ell}^\prime$). This property is
extremely useful for interactive communication: even if the decoding of a 
message is incorrect in some round, if there are sufficiently many error-free 
subsequent transmissions, we can later correct that error. This self-healing
property is essential to our analysis of the simulation protocol, and to our 
proof of 
%Lemma~
\cref{lem:optcor}.

%------------------------------------------------------------
%------------------------------------------------------------
\subsubsection{Blueberry Codes}
	\label{sec:bbc}

Another kind of online code we need in order to withstand the highest possible 
error rates are randomized error detection codes called blueberry codes 
in Ref.~\cite{FGOS12}. To use these, Alice and Bob encode and decode messages 
with a shared secret key in a way that weakly  authenticates and encrypts 
each message, and in this way adversary Eve cannot apply a corruption 
of her choosing. Such codes unknown to the adversary were termed private 
codes in Ref.~\cite{Lan04}. At best, with some small (but constant) 
probability Eve is able to corrupt a message in such a way that Alice and 
Bob do not detect it, and this results in an effective decoding error, but 
most of the time a corruption of Eve results in an effective erasure 
decoding. Since the tree code property, and hence also its decoding, is 
defined in terms of Hamming distance, transmission errors are twice as harmful 
as erasures in the tree decoding. (We can view the erasure flag $\perp$ as 
a special symbol in $\Sigma$; although never used in the encoding, this
symbol helps 
in decoding.) When incorrect decoding occurs, the two parties might
perform  operations on the quantum registers that need to be corrected
later. On the other hand, when an erasure occurs, it is 
visible to the recipient and this prevents him from performing such incorrect
operations. Hence, concatenating a blueberry code with the tree 
code enables significant improvement in the allowed error rates.

These blueberry codes were defined in Ref.~\cite{FGOS12} for the purpose of 
authenticating streams of classical messages and for the simulation of 
interactive classical protocols. Below we summarize their
definition and important properties.

\begin{definition}
	\emph{(blueberry codes \cite{FGOS12})} For $i \geq 1$ let $\msB_i : \Gamma 
	\rightarrow \Gamma$ be a random and independent permutation. The 
	\emph{blueberry code} maps a string $e \in \Sigma^t \subset 
	\Gamma^t$ of arbitrary length $t$ to $\msB(e) = \msB_1(e_1) \msB_2(e_2) 
	\cdots \msB_t(e_t)$. We denote such a code as $\msB: \Sigma^* 
	\rightarrow \Gamma^*$, and define the erasure parameter of this 
	code as $\beta = 1 - \tfrac{|\Sigma| - 1}{|\Gamma| - 1}$, and its 
	complement $\epsilon_\beta = 1 - \beta = \tfrac{|\Sigma| - 
	1}{|\Gamma| - 1}$.
\end{definition}

\begin{definition}
	Assume that at some time $i$, $d_i = \msB_i (e_i)$ is transmitted and 
	$d_i^\prime \not= d_i$ is received. If $d_i^\prime 
	\not\in \msB_i (\Sigma)$, we mark the transmission as an erasure, and the 
	decoding algorithm (for the Blueberry code) outputs $\perp$. 
	Otherwise, this event is called an error.
\end{definition}

\begin{corollary}
	Let $e \in \Sigma^t$ and assume $\msB(e)$ is communicated over a 
	noisy channel. Every symbol corrupted by the channel causes 
	either an error with probability $\epsilon_\beta$, or an erasure 
	with probability $\beta$.
\end{corollary}

\begin{lemma}
\label{lem:bbc}
	Assume a blueberry code $\msB : \Sigma^* \rightarrow \Gamma^*$ is 
	used to transmit a string $e \in \Sigma^t$ over a noisy channel. 
	For any constant $0 \leq c \leq 1$, if the channel's corruption 
	rate is $c$, then with probability $1 - 2^{- \Omega(t)}$ at least 
	a $(1 - 2 \epsilon_\beta)$-fraction of the~$ct$ \emph{corrupted\/} transmissions are 
	marked as erasures.
\end{lemma}

\begin{corollary}
	If out of $t$ received transmissions, $ct$ were marked as 
	erasures while decoding a blueberry code $\msB : \Sigma^* \rightarrow \Gamma^*$, 
	then except with probability $2^{- \Omega(t)}$ over the shared 
	randomness, the adversarial corruption rate is at most $c / (1 - 
	2 \epsilon_\beta)$.
\end{corollary}

\section{Basic Simulation Protocol}
	\label{sec:bas}

We start by describing a basic simulation protocol, which achieves our
first goal of simulating quantum protocols with asymptotically positive 
communication and tolerable error rates, and with entanglement consumption 
linear in the communication.
This provides an interactive analogue of a family of good quantum 
codes. This protocol contains the essential ideas of the optimal protocol of 
%section
\cref{sec:opt}, but the description and analysis are simplified
because we do not have the additional blueberry code layer. Moreover,
this protocol succeeds with perfect fidelity, provided the number of
errors is below a certain threshold.

%------------------------------------------------------------
%------------------------------------------------------------

\subsection{Result}

We focus on the shared entanglement model. Techniques to distribute 
entanglement in both random \cite{Lloyd97, Shor02, Dev05}
 and adversarial \cite{CRSS98, FM04, Rai99} error models are 
well studied. We can combine our findings with these entanglement 
distribution techniques to translate results in the shared entanglement model
to the quantum model. 
We first focus on an adversarial error model, and then adapt 
these results to a random error model. Such extensions to 
other communication models are explored in 
%Section~
\cref{sec:oth}. For 
the basic simulation protocol described in this section, entanglement is 
only used to teleport the quantum information back and forth between the 
two parties. In 
%Section~
\cref{sec:opt}, we show how to tolerate maximum 
error rates by also using entanglement to generate a shared secret key 
unknown to the adversary, thus enabling the two honest parties to detect 
most adversarial errors as effective erasures.

Given an adversarial channel in the shared entanglement model with 
low enough error rate, we show how to simulate perfectly any noiseless 
protocol of length $N$ over this channel using a number of transmissions 
linear in $N$, and consuming a linear number of EPR pairs. 
More precisely, we prove the following.
(See~\cref{sec:nqucommsh} for the definition 
of~$\A_{\delta, q, N^\prime}^\rS \; $ which is mentioned in the theorem.)
\begin{theorem}
\label{th:bas}
There exist a constant error rate $\delta > 0$, communication 
rate $R_\rC > 0$, 
transmission alphabet size $q \in \mathbb{N}$, and entanglement 
consumption rate $R_\rE \in \mathbb{R}^+$ such that for all noiseless 
protocol lengths $N \in 2 \mathbb{N}$, there exists a universal simulator 
$S$ in the shared entanglement model of length $N^\prime$, with 
communication rate at least $R_\rC$, transmission alphabet size $q$, 
entanglement consumption rate at most $R_\rE$, which succeeds with zero 
error at simulating all noiseless protocols of length $N$ against 
all adversaries in $\A_{\delta, q, N^\prime}^\rS$.
\end{theorem}
Specific values for the constants posited in the theorem are given at
the end of 
%Section~
\cref{sec:anal}.

%------------------------------------------------------------
%------------------------------------------------------------
\subsection{Intuition for the Simulation Protocol}
	\label{sec:int}

Before describing in detail the basic simulation protocol, first we
give some intuition on how it succeeds in simulating a noiseless quantum 
protocol over a noisy channel. The strategy to avoid losing the quantum 
information in the communication register over the noisy channel is to 
teleport the $C$ register of the noiseless protocol back and forth into 
Alice's $C_\sA$ register and Bob's $C_\sB$ register, creating a virtual $C$ 
register which is either in Alice's or in Bob's hand. They use the shared 
entanglement in $T_\sA T_\sB$ to do so, and use the noisy 
classical channels to transmit their teleportation measurement outcomes. 
Whenever Alice possesses the virtual $C$ register she can try to evolve 
the simulation of the noiseless protocol by applying one of her noiseless 
protocol unitary operators on the virtual $AC$ register, and this
applies similarly for Bob on 
the virtual $BC$ register. If they later realize that there has been 
some error in the teleportation decoding, they might have to apply 
inverses of these operations, but overall, everything acting on the 
virtual $ABC$ quantum register can be described as an intertwined 
sequence of Pauli operators acting on the $C$ register and noiseless 
protocol unitary operators (and their inverses) acting on the $AC$ and the $BC$ 
registers. There are two important points to notice here. First, the 
sequence of operations acting on the joint register is a sequence of 
reversible unitary operators. 
Hence, if the parties keep track of the sequence of operations on the 
joint register, then at least one of the parties can reverse any of his/her 
operations when he/she is in possession of the virtual $C$ register. Second, 
both parties know the order in which these operators have been applied 
while only one knows exactly which operator was applied: for Pauli operators, 
both parties know $\pm \rX^x \rZ^z$ is applied at some point, but only one 
knows the correct value of $x z \in \{0, 1 \}^2$, and similarly both 
know that $U_j^M$ (with $U_j^{+1} = U_j$, $U_j^{-1} = U_j^\dagger$, $U_j^0 = \rI$) 
is applied at some point, but only one knows the correct values of $j 
\in \{1, \dotsc, N^\prime + 1 \}$ and $M \in \{-1, 0, +1 \}$. This is the 
classical information they try to transmit to each other so that both
know exactly the sequence of operations that have been applied on the joint 
register. The tree codes due to Schulman are particularly 
well suited for protecting against noise in this interactive scenario.

More concretely, in each round the parties first need to decode the 
teleportation before trying to evolve the simulation of the quantum 
protocol and finally teleporting back the communication register to the 
other party. The goal is for each party to know his/her exact position
in the simulation of the protocol (i.e., the sequence of unitary
operators
that have been applied to the virtual protocol registers) 
when they are able to correctly 
decode the classical messages sent by the other party. 
To enable a party to learn exactly what action was taken by the other 
party in the earlier rounds, the message sent in each round is in $\{0, 1 
\}^2 \times \{-1, 0, +1\} \times \{0, 1 \}^2$, encoded with a tree code. 
The first pair of bits corresponds to 
the teleportation decoding operation done at the beginning of a party's 
turn. The trit is associated with the evolution in the noiseless 
protocol: $+1$ stands for going forward with the protocol, i.e., for a
unitary operator of 
the noiseless protocol that was applied to the joint state of the party's 
local register and the communication register; $-1$ stands for going 
backward with the protocol, i.e., for the inverse of a unitary operation of 
the noiseless protocol that was
applied by that party to the joint state; and $0$ stands for 
holding the protocol idle, i.e., no action is taken by that party to evolve the 
protocol in that round. Note that the index $j$ of the unitary operator $U_j^M$ 
that a party applies can be computed solely from the sequence of trits sent by 
that party, and such an explicit calculation is defined in the simulation 
description. Finally, the last pair of bits corresponds to the outcome of 
the measurement in the teleportation of the communication register, which
enables the other party to correctly decode the teleportation.

For each party, we call his/her \emph{history} at some point the sequence of 
these triplets of messages that he/she transmitted up to that point (see 
%section 
\cref{sec:clcomm}). If a party succeeds in correctly decoding the history of 
the other party, he/she then possesses all the information about the 
operations that were applied on the joint quantum register and can choose his/her 
next move accordingly.

Note that the information about which Pauli 
operator was used to decode the teleportation might appear 
redundant, it is not when there are decoding errors. 
This is a subtle and important point, 
so let us explain in more detail what we mean. In the case of decoding errors,
the wrong Pauli operator might be 
applied to do the teleportation decoding. Even though the party who 
applied the wrong Pauli operator will later realize his/her mistake (when the self-healing property of the
tree code eventually enables him/her to decode this message correctly), the other 
party still needs to be informed of this previous error in decoding.  Sending the 
information about which Pauli operator was used to do the teleportation 
decoding accomplishes this and even enables the other party to 
correct this wrong teleportation decoding if needed. 
Indeed this property has an essential use, especially in the simulation for
maximal error tolerance in 
%Section~
\cref{sec:opt}. In more detail, when a corruption is detected as an erasure,
the teleportation decoding operation applied is the trivial one.
This is wrong three-quarters of the time on average.
Another approach that would also work would be to 
let the other party know what information was received, 
and then let each party correct for his/her own previous decoding error. 
The problem with this is that the tolerable error rate would 
have to be much lower than $\frac{1}{2} - \epsilon$: in the terms used in the analysis, 
we would need a good round to recover from an erasure round, which is undesirable.

%------------------------------------------------------------
%------------------------------------------------------------
\subsection{Description of the Simulator}

All communication is done with a tree encoding over some alphabet 
$\Sigma$. To later simplify the analysis, we fix the distance parameter 
to $\alpha = \tfrac{39}{40}$. The message set consists of $\{0, 1 \}^2 
\times \{-1, 0, +1 \} \times \{0, 1 \}^2 \cong [4] \times [3] \times [4] 
\cong [48]$, so we take arity $d=48$. Also, taking $N^\prime = 4 (1 + 
\tfrac{1}{N}) N$ is sufficient. By 
%Lemma~
\cref{lem:tccode}, we know that 
there exists a $q \in \mathbb{N}$ independent of $N^\prime$ such that an 
alphabet $\Sigma$ of size $q$ suffices to label the arcs of a tree code 
of any depth $N^\prime \in \mathbb{N}$.  Before the protocol begins,
both parties agree on such a tree code of depth $N^\prime$ with 
corresponding encoding and decoding functions $\msE$ and $\msD$ (each party 
uses a separate instance of the same tree code to transmit her/his messages 
to the other party). The goal is to tolerate error rates up to $\delta = 
\tfrac{1}{80}$.

We use the following convention for the variables describing the protocol.
On Alice's side, in round $i$, $x_{i}^{\sAD} z_{i}^{\sAD} \in \{0, 1 \}^2$ 
correspond to the bits she uses for the teleportation decoding on the $\rX$ 
and $\rZ$ Pauli operators, respectively; $x_{i}^{\sAM} z_{i}^{\sAM} \in \{0, 1 \}^2$ 
correspond to the bits of the teleportation measurement on the 
corresponding Pauli operators; $j_{i}^{\sA} \in \mathbb{Z}$ and $M_{i}^{\sA} \in 
\{-1, 0, +1 \}$ correspond, respectively, to the index of the unitary
operator she 
uses in round $i$ and to whether she uses $U_{j_{i}^{\sA}}^{+1} = U_{j_{i}^{\sA}}$
or its 
inverse $U_{j_{i}^{\sA}}^{-1} = U_{j_{i}^{\sA}}^\dagger$ or simply applies the 
identity channel $U_{j_{i}^{\sA}}^0 = \rI$ on the $AC$ quantum register; and the 
counter $c_{i}^{\sA}$ keeps track of the sum of all previous messages $M_{\ell 
}^{\sA}$, $l \leq i$. On Bob's side, we use a similar set of variables, 
with superscript $\sB$ instead of $\sA$. All Pauli operators are applied on 
the virtual $C$ register.
When discussing variables obtained 
from decoding in round~$i$, a superscript $i$ is added to account for the 
fact that this decoding might be wrong and could be corrected in later 
rounds. Similarly, a superscript $i$ is used when discussing other variables that are 
round-dependent.

\subsubsection{Representations of the Joint State}

The actions taken by Alice and Bob round~$i$ are based on their best
guesses for the state $\ket{\psi_i}$ of the joint register at the 
beginning of round~$i$. (Note that $\ket{\psi_1} = 
\ket{\psi_\mathrm{init}}$ is the initial state in the protocol being
simulated.) The state~$\ket{\psi_i}$ can be classically 
computed from the information in Alice's and Bob's histories; due to noise, it 
is generally unknown, at least in part, to Alice and Bob.
The analysis rests on the following two representations for the 
state~$\ket{\psi_i}$. The first can be directly 
computed, up to irrelevant operations of Eve on the $E$ register, as
\begin{align}
\label{eq:psii}
%	\ket{\psi_i}^{ABCE} & = \notag \\
%		\prod_{\ell=1}^{i-1}  ( \rX^{x_{\ell BM}} \rZ^{z_{\ell BM}} &
%		U_{j_{\ell B}}^{M_{\ell B}} \rZ^{z_{\ell BD}} \rX^{x_{\ell 
%	BD}} \notag \\
%		\rX^{x_{\ell AM}} \rZ^{z_{\ell AM}} &
%		U_{j_{\ell A}}^{M_{\ell A}} \rZ^{z_{\ell AD}} \rX^{x_{\ell AD}} )
%		\ket{\psi_{\mathrm{init}}}^{ABCE}.
	\ket{\psi_i}^{ABCER}  = 
		\prod_{\ell=1}^{i-1}  \left( \rX^{x_{\ell}^{ \sBM}} \rZ^{z_{\ell}^{ \sBM}} 
		U_{j_{\ell}^{ \sB}}^{M_{\ell}^{ \sB}} \rZ^{z_{\ell}^{ \sBD}} \rX^{x_{\ell}^{ \sBD}} 
		\rX^{x_{\ell}^{ \sAM}} \rZ^{z_{\ell}^{ \sAM}} 
		U_{j_{\ell}^{ \sA}}^{M_{\ell}^{ \sA}} \rZ^{z_{\ell}^{ \sAD}} \rX^{x_{\ell}^{ \sAD}} \right)
		\ket{\psi_{\mathrm{init}}}^{ABCER}.
\end{align}
Here, from the history $s_\sA$ of Alice's history tree, we can directly obtain 
from the $\ell$th message sent by Alice, for $\ell = 1 \cdots i-1$, the 
two bits $x_{\ell}^{ \sAD} z_{\ell}^{ \sAD}$ used to decode the teleportation, the 
trit $M_{\ell}^{ \sA}$ corresponding to the evolution of the protocol 
performed in round $\ell$, and then the two bits $x_{\ell}^{ \sAM} z_{\ell 
}^{\sAM}$ corresponding to the outcome of the teleportation measurement. We 
then use counters $c_{\ell}^{ \sA}$'s that maintain the sums of the $M_{\ell 
}^{\sA}$'s to compute the indices $j_{\ell}^{ \sA}$'s of the noiseless protocol 
unitary operators used by Alice in round~$\ell$: $c_{0}^{\sA} = 0, c_{\ell}^{ \sA} = 
c_{(\ell-1)}^{\sA} + M_{\ell}^{ \sA}, j_{\ell}^{ \sA} = 2 c_{(\ell-1)}^{\sA} + M_{\ell}^{ \sA}$. 
Note that $j_{i}^{\sA}$ depends only on the sequence of messages $M_{1}^{\sA}, 
M_{2}^{\sA}, \dotsc , M_{(i-1)}^{\sA}, M_{i}^{\sA}$. Similarly, the history $s_\sB$ of Bob's 
history tree is used to obtain $x_{\ell}^{ \sBD} z_{\ell}^{ \sBD}, x_{\ell}^{ \sBM} 
z_{\ell}^{ \sBM}$, as well as $M_{\ell}^{ \sB}$, and to compute $c_{0}^{\sB} = 0, 
c_{\ell}^{ \sB} = c_{(\ell-1)}^{\sB} + M_{\ell}^{ \sB}, j_{\ell}^{ \sB} = 2 c_{(\ell-1)}^{\sB} + 
M_{\ell}^{ \sB} + 1$. We define $U_j^M = \rI$ whenever $j \leq 0$ or $M=0$. Note 
that if $M_\ell^\sA \not= 0$, $j_{\ell}^{ \sA}$ is odd and $U_{j_{\ell}^{ \sA}}^{M}$ 
acts on Alice's side. Similarly, if $M_\ell^\sB \not= 0$, $j_{\ell}^{ \sB}$ is even and $U_{j_{\ell}^{ \sB}}^{M}$ 
acts on Bob's side. Also note that $j \leq N^\prime + 1$, so the $U_j$'s 
are well-defined, by the noiseless protocol embedding described in
%Section~
\cref{sec:nslss}.

This first representation of the form of the state $\ket{\psi_i}$ is not too informative in itself, but from it
we can classically compute a second representation by recursively cleaning 
it up.  The cleanup is performed by combining
as many of the operators as possible as follows: we multiply all consecutive
Pauli operators acting on the~$C$ register, and simplify consecutive pairs of
operators $U_\ell, U_\ell^{-1}$ acting on the same set of qubits,
to obtain a state of the form
\begin{align}
	\label{eq:coll}
%	\ket{\psi_{i}}^{ABCE} & = \notag \\
%	 \hat{\sigma}^i \cdot \tilde{U}_{t_i}^i & \cdot 
%	\tilde{\sigma}_{t_i}^i \cdot \tilde{U}_{t_i-1}^i \cdot 
%	\tilde{\sigma}_{t_i-1}^i \cdots \notag \\
%		\tilde{U}_2^i \cdot \tilde{\sigma}_2^i \cdot \tilde{U}_1^i 
%	& \cdot \tilde{\sigma}_1^i \cdot U_{r_i} \cdot U_{r_i-1} \cdots 
%	U_2 \cdot U_1 \ket{\psi_{\mathrm{init}}}^{ABCE}
\suppress{
	\ket{\psi_{i}}^{ABCE}  =
	 \hat{\sigma}^i \cdot \tilde{U}_{t_i}^i & \cdot 
	\tilde{\sigma}_{t_i}^i \cdot \tilde{U}_{t_i-1}^i \cdot 
	\tilde{\sigma}_{t_i-1}^i \cdots
		\tilde{U}_2^i \cdot \tilde{\sigma}_2^i \cdot \tilde{U}_1^i 
		& \cdot \tilde{\sigma}_1^i \cdot U_{r_i} \cdot U_{r_i-1} 
		\cdots U_2 \cdot U_1 \ket{\psi_{\mathrm{init}}}^{ABCE}
}
	\ket{\psi_{i}}^{ABCER}  =
	 \hat{\sigma}^i \; \tilde{U}_{t_i}^i 
	\; \tilde{\sigma}_{t_i}^i  \; \tilde{U}_{t_i-1}^i  
	\; \tilde{\sigma}_{t_i-1}^i \; \cdots
		\; \tilde{U}_2^i  \; \tilde{\sigma}_2^i  \; \tilde{U}_1^i 
		 \; \tilde{\sigma}_1^i \:   U_{r_i}  U_{r_i-1} 
		\cdots U_2  U_1 \ket{\psi_{\mathrm{init}}}^{ABCER}
\end{align}
with $\hat{\sigma}^i = \pm \rX^{\hat{x}^i} \rZ^{\hat{z}^i}$, and
for~$\ell \in \{1, \dotsc, t_i\}$,
$\tilde{\sigma}_\ell^i = \rX^{x_\ell^i} \rZ^{z_\ell^i}$ for $\hat{x}^i 
\hat{z}^i, x_\ell^i z_\ell^i \in \{0, 1 \}^2$, and $\tilde{U}_\ell^i = 
U_{\ell^\prime}^{\pm 1}$ for some $r_i - 2t_i \leq \ell^\prime \leq r_i + 
2t_i$. The rules used recursively to perform the cleanup are the 
following: in the case when $\tilde{\sigma}_\ell^i = \rI$, 
for two consecutive unitary operators acting on the same set of
qubits we require that if 
$\ell > 1$, then $\tilde{U}_{\ell}^i \not= (\tilde{U}_{\ell-1}^i)^{-1}$, 
and if $\ell=1$, then $\tilde{U}_1^i \not= U_{r_i+1}$ and~$\tilde{U}_1^i
\not= U_{r_i}^{-1}$. This last rule is 
what determines the cut between $U_{r_i}$ and $\tilde{U}_1^i 
\tilde{\sigma}_1^i$. The parameter $r_i$ determines the number of 
noiseless protocol unitary operators the parties have been able to successfully 
apply on the joint register before errors  arise, and the 
parameter $t_i$ determines the number of errors the parties have to 
correct before being able to evolve the state as in the noiseless protocol.
Note that this is 
well defined: there is a unique representation in the form 
\cref{eq:coll} corresponding to any in the form \cref{eq:psii}.
This second representation is thus powerful: it is the analogue in our
setting of the protocol tree representation of classical protocols, and
it enables us to precisely keep track of the evolution of the noiseless
protocol simulation. This is why Alice and Bob will always base their
actions on their best estimates of this representation.

\subsubsection{Choosing the Next Step}

To decide which action to take in round~$i$, Alice starts by decoding the 
possibly corrupted messages $f_1^\prime, \dotsc, f_{i-1}^\prime \in 
\Sigma$ received from Bob up to this point to obtain her best guess 
$s_\sB^i = \msD(f_1^\prime, \dotsc, f_{i-1}^\prime)$ for the history $s_\sB$ of 
his history tree. Along with the history $s_\sA$ of her history tree,
she uses this 
to compute her best guess of the form \cref{eq:coll} of the joint state. 
If her decoding of Bob's history is \emph{good} (error-free), then she has all the information 
she needs to compute the joint state $\ket{\psi_i}$. She can 
then choose the correct actions to evolve the simulation. She takes 
the following actions based on the assumption that her decoding is good. 
If it is not, errors might accumulate on the joint register $ABC$, which 
she will later have to correct.

Alice's next move depends on whether $t_i=0$ in her 
best guess for the state~$\ket{\psi_i}$.  If 
$t_i=0$, then she wishes to evolve the protocol one round further, if it 
is her turn to do so. That is, if $r_i$ is even, then Alice sets $M_{i}^{\sA} = 
+1$ to apply $U_{r_i + 1}^{AC}$, but if $r_i$ is odd, Bob should be the 
next to apply a unitary operator of the protocol, so she sets $M_{i}^{\sA} = 0$. If 
$t_i \not= 0$, then she wishes to correct the last error not yet 
corrected if she is the one who applied it. That is, if $\tilde{U}_{t_i} 
= U_{\ell^\prime}^{M^\prime}$ for $\ell^\prime$ odd, then she sets 
$M_{i}^{\sA}= - M^\prime \in \{ \pm 1 \}$ (note that in this case it holds 
that $j_{i}^{\sA} = \ell^\prime$); otherwise, she sets $M_{i}^{\sA} = 0$ and 
hopes that Bob will correct $\tilde{U}_{t_i}$. In all cases, with 
$\hat{\sigma}_i^C = \pm \rX^{\hat{x}_i} \rZ^{\hat{z}_i}$, she sets $x_{i}^{\sAD} = 
\hat{x}_i, z_{i}^{\sAD}=\hat{z}_i$ and computes $c_{i}^{\sA} = c_{(i-1)}^{\sA} + M_{i}^{\sA}, 
j_{i}^{\sA} = 2 c_{(i-1)}^{\sA} + M_{i}^{\sA}$. Note that she does 
not care about the 
global phase factor $\pm 1$ appearing in $\hat{\sigma}_i$ 
during the clean-up from the form \cref{eq:psii} to the form 
\cref{eq:coll}. This phase arises because the Pauli operators $\rX$ and $\rZ$ 
anticommute, and it is irrelevant.

After this classical preprocessing, she can now perform her quantum 
operations on the $AC$ registers: she first decodes the teleportation 
operation (and possibly some other Pauli errors remaining on the $C$ 
register) by applying $\rZ^{z_{i}^{\sAD}} \rX^{x_{i}^{\sAD}}$ on 
the $T_\sA^{2(i-1)}$ register before swapping registers $T_\sA^{2(i-1)}$ 
and $C_\sA$, effectively putting the virtual $C$ register into $C_\sA$. 
(Note that in round $1$, Alice already possesses the $C$ 
register so this part is trivial: we let $T_\sA^0 = C_\sA$ and 
set $x_{1}^{\sAD} z_{1}^{\sAD} = 00$.) She then 
performs $U_{j_{i}^{\sA}}^{M_{i}^{\sA}}$ on the virtual $AC$ register to try to 
evolve the protocol (or correct a previous error) before teleporting 
back the virtual $C$ register to Bob using the half of the entangled state in 
the $T_\sA^{2i-1}$ register, obtaining measurement outcome $x_{i}^{\sAM} z_{i}^{\sAM} 
\in \{0, 1 \}^2$. She updates her history $s_\sA$ by following the edge $a_i 
= (x_{i}^{\sAD} z_{i}^{\sAD}, M_{i}^{\sA}, x_{i}^{\sAM} z_{i}^{\sAM})$ in the history tree, and 
transmits message $e_i = \msE(a_1 \cdots a_i)$ over the noisy classical 
channel, with $\msE$ the encoding function of the tree code.

Upon receiving the message $e_i^\prime$, a possibly corrupted version 
of $e_i$, Bob obtains his best guess $s_\sA^i$ for Alice's history $s_\sA$ by 
computing, with previous messages $e_1^\prime \cdots e_{i-1}^\prime$, 
$s_\sA^i = \msD(e_1^\prime \cdots e_i^\prime)$. He uses this, along with his own 
history $s_\sB$, to compute his best guess of the representation of the
state
\begin{align}
 	\left( \rX^{x_{i}^{\sAM}} \rZ^{z_{i}^{\sAM}} U_{j_{i}^{\sA}}^{M_{i}^{\sA}} \rZ^{z_{i}^{\sAD}} 
	\rX^{x_{i}^{\sAD}} \right) \ket{\psi_i}
\end{align}
analogous to that in \cref{eq:psii}. He then cleans this up to obtain a 
representation analogous to that in \cref{eq:coll} and, based on this latest 
representation, chooses in the same way as Alice his $x_{i}^{\sBD} z_{i}^{\sBD}, 
M_{i}^{\sB}$, and then uses $M_{i}^{\sB}$ and $c_{i-1}^{\sB}$ to compute $c_{i}^{\sB}, j_{i}^{\sB}$. After this 
classical preprocessing, he can then perform his quantum operations: he 
first decodes the teleportation operation by applying $\rZ^{z_{i}^{\sBD}} 
\rX^{x_{i}^{\sBD}}$ on the $T_\sB^{2i-1}$ register and by swapping it with $C_\sB$, 
creating a virtual $C$ register, then performs $U_{j_{i}^{\sB}}^{M_{i}^{\sB}}$ on 
the virtual  $BC$ register to try to evolve the protocol before 
teleporting back the virtual $C$ register to Alice using the half of 
the entangled state in the $T_\sB^{2i}$ register, and obtains measurement 
outcome $x_{i}^{\sBM} z_{i}^{\sBM}$. He updates his history $s_\sB$ by following the 
edge $b_i = (x_{i}^{\sBD} z_{i}^{\sBD}, M_{i}^{\sB}, x _{i}^{\sBM} z_{i}^{\sBM})$, and transmits 
message $f_i = \msE(b_1 \cdots b_i)$ over the channel. The round is completed
when Alice receives message $f_i^\prime$, a possibly corrupted version of 
$f_i$. After the $\tfrac{N^\prime}{2}$ rounds, Alice and Bob take the 
particular registers $\tilde{A}, \tilde{B}$, and $\tilde{C}$ specified by 
the noiseless protocol embedding (see 
%section 
\cref{sec:nslss}) and use 
them as their respective outcomes for the protocol. If the simulation is 
successful, the output quantum state corresponds to the $ABC$ subsystem of 
$\ket{\psi_{\mathrm{final}}}^{ABCE}$ specified by the original noiseless 
protocol.  We later prove that the protocol is successful if the error 
rate is below $\tfrac{1}{80}$.

\subsubsection{Summary of Protocol}

We summarize the protocol below. Alice and Bob start with the
state~$\ket{\psi_\mathrm{init}}$ in the registers~$A B C_\sA E$, the
register~$C_\sB$ initialized to~$\ket{0}$, the registers~$T_\sA
T_\sB$ initialized to~$N^\prime$ EPR pairs~$\left[\tfrac{1}{\sqrt{2}}(\ket{00}
+ \ket{11})\right]^{\otimes {N^\prime}}$, with one qubit each from each
EPR pair held by Alice and
Bob, and the qubits in registers~$\tilde{A}, \tilde{B}, \tilde{C}$ initialized
to~$\ket{0}$ (cf. the noiseless protocol embedding described in
%Section~
\cref{sec:nslss}). They also have access to a suitable amount of 
classical workspace for local computations required for the simulation.
They repeat the following for \mbox{$i = 1$, $\dotsc$, $\tfrac{N^\prime}{2}$:} 

\begin{enumerate}
 	\item If~$i > 1$, Alice computes $s_\sB^i = \msD(f_1^\prime \cdots 
		f_{i-1}^\prime)$, and  for $\ell=1, \dotsc, i-1$ she extracts $b_\ell^i = (x_{\ell}^{i \sBD} 
		z_{\ell}^{i \sBD}, M_{\ell}^{i \sB}, x_{\ell}^{i \sBM} z_{\ell 
		}^{i \sBM})$. These are her best 
		guesses for Bob's messages. She computes the corresponding 
                $c_{\ell}^{i \sB}, j_{\ell }^{i \sB}$. 
		For~$i = 1$, the
                values of the parameters Alice needs for the simulation
                are straightforward.

	\item Also using $s_\sA$, she computes her best guess for the form 
		\cref{eq:coll} of the state $\ket{\psi_{i}}$ of the 
		joint register and of the corresponding $x_{i}^{\sAD} z_{i}^{\sAD}$, $
		M_{i}^{\sA}$, $c_{i}^{\sA}$, $j_{i}^{\sA}$, described
                earlier in this section.
	\item If~$i > 1$, she completes the teleportation operation by applying $\rZ^{z_{i}^{\sAD}} 
		\rX^{x_{i}^{\sAD}}$ to register $T_\sA^{2(i-1)}$ and swaps this 
		with the $C_\sA$ register.
	\item She applies
		$U_{j_{i}^{\sA}}^{M_{i}^{\sA}}$ to the $AC_\sA$ register,
              in an attempt to evolve the original protocol.
	\item She teleports the $C_\sA$ register to Bob using 
		entanglement in register $T_\sA^{2i-1}$ and gets outcomes 
		$x_{i}^{\sAM} z_{i}^{\sAM}$.
	\item Alice updates her state $s_\sA$ by following edge $a_i = 
		(x_{i}^{\sAD} z_{i}^{\sAD}, M_{i}^{\sA}, x_{i}^{\sAM} z_{i}^{\sAM})$ and transmits 
		message $e_i = \msE(a_1 \cdots a_i)$ using the noisy
classical channel to 
		Bob, who receives $e_i^\prime$, a possibly corrupted 
		version of $e_i$.
	\item Bob computes $s_\sA^i = \msD(e_1^\prime \cdots e_i^\prime)$ and 
		also using $s_\sB$, performs actions analogous 
		to Alice's. He completes the teleportation
operation, swaps register $T_\sB^{2i-1}$ with 
		$C_\sB$, applies the appropriate unitary operation
to the register~$C_\sB B$, uses the $T_\sB^{2i}$ register to teleport 
		the $C_\sB$ register to Alice, and finally transmits $f_i$. 
		Round $i$ is completed when Alice receives
		$f_i^\prime$, a possibly corrupted version of $f_i$.
\end{enumerate}
After these $\tfrac{N^\prime}{2}$ rounds, both Alice and Bob extract their protocol 
outcomes from the $\tilde{A} \tilde{B} \tilde{C}$ registers specified by 
the noiseless protocol embedding.

%------------------------------------------------------------
%------------------------------------------------------------
\subsection{Analysis}
	\label{sec:anal}

The analysis is done conditioned on some overall classical state (and in particular, some respective views of Alice and Bob) at each round. By a view of Alice or Bob, we mean the transcript of messages sent and received. Moreover, if the adversary Eve has an adaptive, probabilistic strategy, we condition on some strategy based on the outcome of her previous measurements. We return to this issue later.

The total number of rounds is 
$\tfrac{N^\prime}{2}$, with two transmissions per round, for a total of 
$N^\prime$ transmissions.
We define two kinds of rounds: \emph{good} rounds in which both parties
correctly decode 
the each other's history, and \emph{bad\/} rounds, in which at least one 
party makes a decoding error. To analyse the protocol, 
we define a ``potential function'' $P(i) \in \mathbb{Z}$,
which increases at least by some (strictly positive) amount in good 
rounds, and decreases by at most some other (bounded) amount in bad 
rounds. The potential function is such that we know the simulation
succeeds whenever $P(\tfrac{N^\prime}{2} +1) \geq N + 1$.
Hence, it is sufficient to bound the ratio of 
good to bad rounds as a function of the error rate to prove the success 
of the simulation. 

Let us now define $P(i)$ more formally. To do so, we use the 
representation \cref{eq:coll} for the form of the quantum state of the 
joint registers at the beginning of round $i$ (or equivalently, at the 
end of round $i-1$). Recall that $r_i$ determines the number of noiseless 
protocol unitary operators that the parties have been able to successfully apply on 
the joint register before errors arise, and $t_i$ 
determines the number of errors that the parties have to correct before being 
able to resume the simulation. Define
\begin{equation}
	P(i) = r_i - 2t_i.
\end{equation}
The factor of $2$ in front of $t_i$ accounts for the worst-case
scenario for the simulation in round~$i$.
As will be apparent from our analysis below, in the worst case, 
all remaining $\tilde{U}_\ell^i$'s are applied by the same 
party who applied $U_{r_i-1}$ and $\tilde{U}_{t_i}^i = 
U_{r_i-1-2(t_i-1)}^{-1}$. 
Then, if $P(\tfrac{N^\prime}{2} + 1) \geq N + 1$,
the operators $\tilde{U}^i_\ell$ in \cref{eq:coll} at the end of the 
simulation (i.e., with~$i = N^\prime+1$) may only be equal to the identity 
operator, as ensured by the noiseless protocol embedding.
%and the $\tilde{A} \tilde{B} \tilde{C}$ registers contain the correct output.
Thus the output of the simulation is correct.
We now prove the following technical lemma 
which bounds $P(i)$ as a function of the number of good and bad rounds.

\begin{lemma}
\label{lem:pi}
	At the end of round $i$, define
\begin{align*}
	N_\rg^i & \quad = \quad |\{j: j \leq i, \text{ round } j \text{ was good}\}|, \\
	N_\rb^i & \quad = \quad |\{j: j \leq i, \text{ round } j \text{ was bad}\}|.
\end{align*}
Then $P(i + 1) \geq N_\rg^i - 4 N_\rb^i$.
\end{lemma}

\begin{proof}
We prove 
%Lemma 
\cref{lem:pi} by induction. For the 
base case, $\ket{\psi_1} = \ket{\psi_\mathrm{init}}$, so $P(1) = 0$, and 
the statement holds. 

To get a flavor of the induction step, let us look 
at $P(2)$ at the end of round $1$. In round $1$, Alice applies $U_1$ 
and then teleports the virtual~$C$ register.
If Bob decodes the message correctly, he applies $U_2$ and 
teleports back the virtual register~$C$, leading to a joint 
state of the form~$\hat{\sigma} U_2 U_1 \ket{\psi_{\mathrm{init}}}$.
In this case $N_\rb^1 = 0$, so $P(2) = 2 \geq 1 = N_\rg^1$.
If there is a decoding error, at worst Bob applies the incorrect
Pauli operation to complete the teleportation step, and he still applies $U_2$.
The joint state is then of the form $\hat{\sigma} U_2 \tilde{\sigma} U_1 
\ket{\psi_{\mathrm{init}}}$. In this case $N_\rg^1 = 0$, and $P(2) 
= 1 - 2 = -1 \geq -4 = -4 N_\rb^1$.

For the induction step, given the state $\ket{\psi_{i}}$ at the end of round 
$i-1$, we consider two cases.  First, suppose that the $i$th round is good, so
that $N_\rg^{i} = N_\rg^{i-1} + 1$ and~$ N_\rb^{i} = N_\rb^{i-1}$.
Both Alice and Bob correctly reconstruct the state as
in~\cref{eq:coll}. If~$t_i = 0$, by the simulation rules, at least one 
of Alice or Bob can advance the original noiseless protocol, and
$t_{i+1} = t_{i} = 0$ and $r_{i+1} \geq r_{i} 
+ 1$. (If~$r_i$ is odd, only Bob advances the protocol, otherwise
both do.) If $t_{i} \geq 1$, again, at least one of Alice or Bob can
invert the unitary operation~$\tilde{U}_{t_i}^i$ (depending on the
parity of~$\ell$, where~$\tilde{U}_{t_i}^i = U_\ell^{\pm 1}$). 
Then $t_{i+1} \leq t_{i} - 1$, and $r_{i+1} \geq r_{i}$.
So in all cases 
\begin{align*}
	P(i + 1) &= r_{i+1} - 2 t_{i+1} \\
		& \geq r_{i} - 2 t_{i} +1 \\
		& = P(i) + 1  \\
		& \geq N_\rg^{i-1} - 4 N_\rb^{i-1} + 1  \\
		& = N_\rg^{i} - 4 N_\rb^{i} \enspace.
\end{align*}

In the second case, the~$i$th round is bad, so
that~$N_\rg^{i} = N_\rg^{i-1}$ and~$N_\rb^{i} = N_\rb^{i-1} + 1$. At
worst, both Alice and Bob decode the received messages incorrectly.
With an incorrect guess for the state in~\cref{eq:coll}, Alice's
actions in this round either decrease~$r_i$ by one, 
increase~$t_i$ by one, or leave both unchanged. The same holds for Bob.
At worst, $t_{i+1} = t_{i} + 2$ and~$ r_{i+1} = r_{i}$. The other eight
possibilities such as $t_{i+1} = t_{i} + 1, r_{i+1} 
= r_{i} - 1$, or $t_{i+1} = t_{i}, r_{i+1} = r_{i} -2$, lead to a smaller
decrease in the potential function~$P$. So
\begin{align*}
	P(i+1) & = r_{i+1} - 2 t_{i+1} \\
		&\geq r_{i} - 2 t_{i} - 4  \\
		&= P(i) - 4 \\
		&\geq N_\rg^{i-1} - 4 N_\rb^{i-1} - 4 \\
		& = N_\rg^{i} - 4 N_\rb^{i}.
\end{align*}
In all cases, $P(i+1) \geq N_\rg^i - 4 N_\rb^i$ which proves the claim. 
\end{proof}

\begin{corollary}
\label{cor:pi}
	If $P(\tfrac{N^\prime}{2} + 1) \geq N + 1$, then the simulation 
	succeeds with zero error.
\end{corollary}

\begin{proof}
	For notational convenience, in this proof let $r = 
	r_{\tfrac{N^\prime}{2} + 1}$, $t = t_{\tfrac{N^\prime}{2} + 1}$. 
	We also let the superscript $\tfrac{N^\prime}{2} + 1$ be implicit in 
	all of the operators $\tilde{U}_\ell^{\pm 1}$ that occur in the proof below.

	The only unitary operations from the original protocol that Alice 
applies are of the form~$U_\ell^{\pm 1}$ for odd~$\ell$. Moreover,
Alice knows her history at all times. Thus, even in a bad round~$i$, she
applies either~$U_{\ell+2}$, $\rI$, or~$U_\ell^{-1}$, where~$U_\ell$ is the
last unitary operation she applied in the representation~\cref{eq:coll}.
A similar statement holds for Bob. Thus, the subscripts in the original
protocol of two consecutive unitary operators applied by the same 
party in~\cref{eq:coll} do not differ by more than~$2$.

	We have $P(\tfrac{N^\prime}{2} + 1) = r - 2t \geq N + 1$, so $r 
	\geq N + 1 + 2t$ with $t \geq 0$. In particular, we have $r 
	\geq N + 1$. Once $U_{r}$ has been applied, the noiseless 
	protocol embedding ensures that the final state of the noiseless 
protocol  in registers $ABC$ is safely stored in local registers 
$\tilde{A} 
	\tilde{B} \tilde{C}$ that are never changed by $U_{N + 2} 
	\cdots U_{N^\prime + 1}$ or by the Pauli operations on the
virtual~$C$ register. It remains to be verified that
	all of the operators $\tilde{U}_\ell$, $0 \le \ell \le t$,
have indices strictly higher than $N + 1$.

The indices (in the original protocol) of the operators~$\tilde{U}_\ell$
applied by Alice may decrease 
	by at most two at once, and similarly for Bob. So the 
	worst case is if all of the operators $\tilde{U}_\ell$ are applied
by the same party, and 
	are inverses of the noiseless protocol unitary operators. Without loss of 
	generality, we consider only this case. If the party 
	who applied $U_r$ also applies all the operators $\tilde{U}_\ell$, then 
	$\tilde{U}_1 = U_r^{-1}$, $\tilde{U}_2 = U_{r-2}^{-1}$, $\dotsc$, $
	\tilde{U}_t = U_{r-2(t-1)}^{-1}$ and $r-2(t-1) > r-2t = 
	P(\tfrac{N^\prime}{2} + 1) \geq N + 1$. Thus the simulation
generates the correct output. Similarly,
	if the party who applied $U_{r-1}$ also applies all the operators
	$\tilde{U}_\ell$, then $\tilde{U}_1 = U_{r-1}^{-1}$, $\tilde{U}_2 = 
	U_{r-3}^{-1}$, $\dotsc$, $ \tilde{U}_{t} = U_{r-2t+1}$, and $r-2t+1 
	>r-2t = P(\tfrac{N^\prime}{2} + 1) \geq N + 1$. In all cases, 
	the safe registers $\tilde{A} \tilde{B} \tilde{C}$ 
	to be outputted by the parties contain the $ABC$ subsystem of 
	$\ket{\psi_\mathrm{final}}$ at the end of round 
	$\tfrac{N^\prime}{2}$ whenever $P(\tfrac{N^\prime}{2} + 1) \geq N + 
	1$.
\end{proof}

We now show that if the number of errors as a fraction of 
$N^\prime$, which is the total number of classical symbols transmitted over the 
adversarial channel, is bounded by a particular constant $\delta > 0$,
then we 
are guaranteed that the simulation succeeds. We do this in two steps:  we 
first give a bound on the fraction of bad rounds as a function of the 
error rate, and then use it to show that below a certain error rate, the 
simulation succeeds.

The bound on the fraction of bad rounds as a function of the error rate 
we use follows from the more general result in 
%Lemma 
\cref{lem:optcor}, which we prove in the next section when studying a
protocol designed to
tolerate the highest possible error rate. The implication we use here 
is the following: if the error rate is bounded by $\delta$ (so there are 
at most $\delta N^\prime$ errors) and the tree code distance of both 
Alice and Bob's tree code is at least $\alpha$, then 
the number of bad rounds $N_\rb$ is bounded as $N_\rb \leq
% 2 \delta N^\prime + \epsilon_\alpha N^\prime =
(2 \delta + \epsilon_\alpha) N^\prime$, where~$\epsilon_\alpha = 1 -
\alpha$.

We are now ready to prove that the simulation succeeds with the 
parameters chosen for our protocol. We have $\epsilon_\alpha = \tfrac{1}{40}$, $
\delta = \tfrac{1}{80}$, $N^\prime = 4(N + 1)$, so 
\begin{align*}
	P\!\left(\frac{N^\prime}{2} + 1 \right) &\geq N_\rg - 4 N_\rb \\
		& = \frac{N^\prime}{2} - 5 N_\rb \\
		& \geq \frac{N^\prime}{2}  - 5(2 \delta + 
			\epsilon_\alpha)N^\prime  \\
		& = N^\prime \left(\frac{1}{2} - \frac{10}{80} - \frac{5}{40}
\right)  \\
		& = \frac{1}{4} N^\prime \\
		& = N + 1 \enspace.
\end{align*}
Here, the first inequality is from 
%Lemma 
\cref{lem:pi}, the first 
equality is by definition of $N_\rg$, $N_\rb$, i.e., $\tfrac{N^\prime}{2} =
N_\rg + N_\rb$, and the second inequality is from our bound on $N_\rb$ 
due to 
%Lemma 
\cref{lem:optcor}. The fact that the simulation succeeds is 
then immediate from 
%Corollary 
\cref{cor:pi}.

Note that the form of the simulation protocol does not depend on the particular 
protocol to be simulated but only on its length $N$ and the noise 
parameter of the adversarial channel we want to tolerate. Also note that 
even if the adversary is adaptive and probabilistic (with adaptive, 
random choices depending on her measurement outcomes and her view 
of the transcript, as allowed by the 
model), the simulation succeeds regardless of her choice of action. As long 
as the corruption rate is bounded by $\delta$, our analysis holds in each 
branch of the adversary's probabilistic computation . We use the 
definition of the class $\A_{\delta, q, N^\prime}^\rS$ to prove that,
indeed, the simulation succeeds with zero error. (See \cref{sec:nqucommsh}
for the definition of~$\A_{\delta, q, N^\prime}^\rS \; $.)

For $\ket{\psi} \in \H(A \otimes B \otimes C \otimes E \otimes R)$, with 
$R$ a purifying system of the same size as $A \otimes B \otimes C \otimes 
E$, we have that
\begin{align*}
 	(\Pi \otimes \rI^R) (\ket{\psi}) = \Tr{E}{U_N \cdots U_1 
	\kb{\psi}{\psi} U_1^\dagger \cdots U_N^\dagger} \enspace,
\end{align*}
where~$\Pi$ is the protocol being simulated.
For any adversary in $\msA \in \A_{\delta, q, N^\prime}^\rS$, the
simulation yields state
\begin{align*}
%(S^\Pi (\msA) \otimes \rI^D ) (\ket{\psi}) &= \\
%\Tr{ \neg (\tilde{A} \tilde{B} \tilde{C} D) }{\M_{N^\prime + 1}^\Pi 
%N_{N^\prime} &\M_{N^\prime}^\Pi \cdots \M_2^\Pi \N_1 %\M_1^\Pi 
%(\kb{\psi}{\psi})},
	(S^\Pi (\msA) \otimes \rI^R ) (\ket{\psi}) =
 	\Tr{ \neg (\tilde{A} \tilde{B} \tilde{C} R) }{\M_{N^\prime + 
	1}^\Pi \N_{N^\prime} \M_{N^\prime}^\Pi \cdots \M_2^\Pi \N_1 
	\M_1^\Pi (\kb{\psi}{\psi})},
\end{align*}
in which the $\neg (\tilde{A} \tilde{B} \tilde{C} R)$ subscript 
for the partial trace means that we trace all except the 
$\tilde{A} \tilde{B} \tilde{C} R$ registers, and the
instrument~$\M^\Pi_\ell$ is the simulation step for the~$\ell$th
local computation by the corresponding party. Then we can rewrite 
\begin{align*}
%(S^\Pi (\msA) \otimes \rI^D ) (\ket{\psi}) &= \\
%\sum_{x_\rT y_\rT z} p_{X_\rT Y_\rT Z} (x_\rT, y_\rT, z) &\kb{x_\rT}{x_\rT}^{X_\rT} \otimes \\
%\kb{y_\rT}{y_\rT}^{Y_\rT} \otimes &\kb{z}{z}^{Z} \otimes \rho(x_\rT, y_\rT, z)
	(S^\Pi & (\msA)  \otimes \rI^R )  (\ket{\psi}) \\
	& = \sum_{x_\rT y_\rT z} p_{X_\rT Y_\rT Z} (x_\rT, y_\rT, z | \ket{\psi}) \; \kb{x_\rT}{x_\rT}^{X_\rT} 
	\otimes 
	\kb{y_\rT}{y_\rT}^{Y_\rT} \otimes &\kb{z}{z}^{Z} \otimes \rho(x_\rT, y_\rT, z)
\end{align*}
where $X_\rT$, $Y_\rT$ are the registers containing the views $x_\rT$, $y_\rT$ of the 
transcript as seen by Alice and Bob, respectively, $Z$ is the adversary's
classical register, $\rho(x_\rT, y_\rT, z)$ are some quantum 
states, and $p_{X_\rT Y_\rT Z}$ is a probability distribution
conditional on the input $\ket{\psi}$. By definition of the class $\A_{\delta, q, N^\prime}^\rS$, we have that, 
conditioned on some classical state $z$ of Eve, $\rho(x_\rT, y_\rT, z)$ 
suffers at most $\delta N^\prime$ corruptions by Eve for any possible 
transcript views $x_\rT$, $y_\rT$. So, by the above analysis, the
$\tilde{A} \tilde{B} \tilde{C} R$ subsystems contains $\Tr{E}{U_N \cdots 
U_1 \kb{\psi}{\psi} U_1^\dagger \cdots U_N^\dagger}$, a perfect copy 
of $(\Pi \otimes \rI^R)(\ket{\psi})$ for any views $x_\rT, y_\rT$ of the 
transcripts of Alice and Bob, respectively. Hence, tracing over all 
subsystems but $\tilde{A} \tilde{B} \tilde{C} R$, we obtain $(\Pi \otimes 
\rI^R) (\ket{\psi})$, and the simulation protocol succeeds with zero probability
of error 
at simulating any noiseless protocol of length $N$ against all 
adversaries in $\A_{\delta, q, N^\prime}^\rS$.

We have thus established the following. We use a tree code of arity $d=48$ 
and distance parameter $\alpha = 1 - \epsilon_\alpha = \tfrac{39}{40}$.
With $q = |\Sigma|$ chosen according to 
%Lemma 
\cref{lem:tccode},
$R_\rC = \tfrac{N}{N^\prime 
\log{q}} = \tfrac{1}{4 (1 + \tfrac{1}{N}) \log{q}} \geq \tfrac{1}{8 
\log{q}}$, $R_\rE = \tfrac{1}{\log{q}}$, and $\delta = \tfrac{1}{80}$, we have 
that for all $N$, there exists a universal simulation protocol in the 
shared entanglement model that, given black-box access to any two-party 
quantum protocol of length $N$ in the noiseless model, succeeds with zero 
probability of  
error at simulating the noiseless protocol on any input (independent of 
the contents of the purifying register held by Eve) while transmitting 
$\tfrac{1}{R_\rC \log{q}} N$ symbols from an alphabet $\Sigma$ of size $q$ 
over any adversarial channel with error rate $\delta$, and consuming 
$\tfrac{R_\rE}{R_\rC} N$ EPR pairs. This proves 
%Theorem~
\cref{th:bas}.

\section{Tolerating Maximal Error Rates}
	\label{sec:opt}

We show how we can modify the basic protocol described in the last 
section such that it tolerates an error rate up to $\tfrac{1}{2} -
\epsilon$,
for arbitrarily small $\epsilon > 0$, in the shared entanglement model.
In particular, we show that given an adversarial channel in the shared 
entanglement model with error rate strictly smaller than $\tfrac{1}{2}$,
we can simulate any noiseless protocol of length $N$ with negligible
error 
over this channel using a linear in $N$ number of constant-size
transmissions and consuming a linear number of EPR pairs.

\begin{theorem}
\label{th:optach}
There exists a constant $c > 0$ such that for arbitrarily small 
constant~$\epsilon > 0$, there exist a communication rate 
$R_\rC > 0$, an alphabet size $q \in \mathbb{N}$, and an entanglement 
consumption rate $R_\rE \geq 0$ such that for all~$N \in 2 \mathbb{N}$,
there exists a universal simulator $S$ for noiseless quantum protocols
of
length~$N$ with the following properties. The simulator~$S$ is in the 
shared entanglement model, has length~$N^\prime$, communication 
rate $R_\rC$, transmission alphabet size $q$, and entanglement
consumption 
rate $R_\rE$. Further, the simulation succeeds with error at 
most~$2^{- c N}$ for all noiseless protocols of length $N$ against 
all adversaries in~$\A_{\tfrac{1}{2} - \epsilon, q, N^\prime}^\rS$~.
\end{theorem}

This is optimal since we also prove that no interactive protocol can 
withstand an error rate of $\tfrac{1}{2}$ in this model. 
In particular, given any two-party quantum protocol of length $N$ in the 
noiseless model, no simulation protocol in the shared entanglement model
can 
tolerate an error rate of $\tfrac{1}{2}$ and succeed in simulating the
noiseless protocol with worst-case error lower than the worst-case error
of the best uni-directional protocol. 

\begin{theorem}
\label{th:optobl}
For all noiseless protocol lengths 
$N \in \mathbb{N}$, communication rates $R_\rC > 0$, transmission 
alphabet sizes $q \in \mathbb{N}$, entanglement consumption rates $R_\rE 
\geq 0$, and simulation protocols $S$ in the shared entanglement 
model of length $N^\prime$ 
with the above parameters, there exists an adversary $\msA \in 
\A_{\tfrac{1}{2}, q, N^\prime}^\rS$ and a unidirectional protocol $U$
such that for all 
noiseless protocols $\Pi$ of length $N$, $\|S^\Pi(\msA) - \Pi
\|_\diamond 
\geq \|U - \Pi \|_\diamond$.
This result holds in the oblivious model as well as the alternating 
communication model.
\end{theorem}

%------------------------------------------------------------
%------------------------------------------------------------
\subsection{Proof of Optimality}

To prove 
%Th. 
\cref{th:optobl}, we observe that the argument of Ref.~\cite{FGOS12} in the 
classical case applies here as well; we need only note that if the 
error rate is $\tfrac{1}{2}$ with alternating communication in the shared 
entanglement model, then an adversary can completely corrupt all of the 
transmissions of either Alice or Bob, at his choosing.
For example, the adversary could replace all of Bob's transmissions by a fixed message
and leave Alice's messages unchanged. Effectively, Bob does not 
transmit any information to Alice, and this protocol can be simulated in 
the uni-directional model. Indeed, suppose that for a fixed register $E$, transmission 
alphabet $\Sigma$ of size $q$, noiseless protocol length $N$, and simulation 
protocol length $N^\prime$, the adversary $\msA_{\tfrac{1}{2}}$ 
maps all transmissions from Bob to Alice to a fixed 
symbol $e_0 \in \Sigma$ for any simulator $S$ of length $N^\prime$ that 
tries to simulate a noiseless protocol $\Pi$ of length $N$. We construct
$\M_1^U$, which is the composition of all operations of Alice in $S$ while 
replacing all messages of Bob by $e_0$. In the unidirectional
protocol~$U$, Alice applies the instrument~$\M_1^U$ to Alice's share of
the joint state in the simulation protocol.
The quantum communication from Alice to Bob is the concatenation of 
all the messages from Alice in the simulation protocol, along 
with Bob's share of the initial joint state. Bob would then apply the
instrument $\M_2^U$, 
which is the sequential application of all his operations in the
simulation protocol $S$. 
This unidirectional protocol simulates $S$ running against 
the adversary $\msA_{\tfrac{1}{2}}$ for any noiseless protocol and any input
and then produces the same output.

The above proof also applies in an oblivious model for noisy 
communication. In an oblivious model, the order in which the 
parties speak is fixed by the protocol and does not depend on the input 
or the actions of the adversary. An adversary can choose to 
disrupt all of the messages of the party who communicates at most half the  
number of bits. Hence, the proof also extends to the case of oblivious, 
but not necessarily alternating, communication. In such a case, the 
simulation protocol would also define a function $\mathrm{Speak}: 
[N^\prime] \rightarrow \{A, B \}$ known to all (Alice, Bob, and Eve) which 
specifies whose turn it is to speak and is independent of both the input and 
the action of Eve.

We can further extend the argument to the case of a $\mathrm{Speak}$ 
function, which 
depends on some secret key and is unknown to Eve, so Eve does not always 
know who is going to speak more often. In that case, Eve can flip a 
random bit to decide which party's communication she is going 
to corrupt. If the communication is classical, then a reasonable 
assumption is that Eve can see who speaks \emph{before\/} she decides 
whether or not to corrupt a message. In this case, the statement is 
changed to ``$\|S^\Pi (\msA) - \Pi \|_\diamond$ is bounded away from zero'', 
as can be seen by considering, for increasing $N$, some family of 
protocols computing, for example, the bitwise parity function of 
$\tfrac{N}{2}$ bits output by both parties or the swap function in which 
Alice and Bob want to exchange their $A, B$ registers. An extension of 
the argument of the proof of 
%Theorem 
\cref{th:iidshentopt} shows that the 
fidelity is also bounded away from $1$ for the case of protocols 
computing the inner product binary function. To reach the $\tfrac{1}{2}$ 
bound on the tolerable error rate, the parties would then need an 
adaptive strategy that depends on the sequence of errors applied by the 
adversary. However, this is dangerous in a noisy model: depending on the 
error pattern, the parties might not agree on whose turn it is to speak, 
and they could run into synchronisation problems.

%------------------------------------------------------------
%------------------------------------------------------------
\subsection{Proof of Achievability}

%------------------------------------------------------------
\subsubsection{Description of the Simulation}

The proof of achievability is somewhat more involved. It follows ideas 
similar to those of the basic simulation, but the protocol must be carefully 
analysed and optimized. We start by setting up new notation that enables
us to do so. The intuition given in 
%section 
\cref{sec:int} still 
applies here, but parameters which were fixed in the basic case 
now depend on the parameter $\epsilon$ when we wish to tolerate
an error rate of~$\tfrac{1}{2} - \epsilon$. In particular, the distance 
parameter $\alpha = 1 - \epsilon_\alpha$, as well as the 
length of the protocol $N^\prime = \ell N$, now changes. Since the parties have access 
to shared entanglement, they do not need to distribute it at the 
beginning of the protocol, and they can also use it to generate a secret 
key unknown to the adversary Eve. The secret key is used to generate a 
blueberry code with erasure parameter \mbox{$\epsilon_\beta = (|\Sigma| - 
1)/ (|\Gamma| - 1)$}, with $\Sigma$ the tree code alphabet and $\Gamma$ the 
blueberry code alphabet. Each of the tree code transmission alphabet 
symbols is further encoded with the blueberry code before transmission 
over the noisy channel. A corruption caused by the adversary is detected 
as an erasure with probability $1 - \epsilon_\beta$. When an erasure is 
detected by either party in a round, that party does not attempt to
continue the simulation (as in the previous section) in that round. 
The corresponding trit sent is $0$, and the teleportation decoding bits 
are $00$. 
Otherwise, the structure of the protocol is mainly unchanged. 

We summarize the optimized protocol below.  Alice and Bob start with the
state~$\ket{\psi_\mathrm{init}}$ in the registers~$A B C_\sA E$, the
register~$C_\sB$ initialized to~$\ket{0}$, the registers~$T_\sA
T_\sB$ initialized to~$N^\prime$ EPR pairs~$\left[\tfrac{1}{\sqrt{2}}(\ket{00}
+ \ket{11})\right]^{\otimes {N^\prime}}$, with one qubit each from each EPR
pair held by Alice and
Bob, and the qubits in registers~$\tilde{A}$, $\tilde{B}$, $\tilde{C}$
initialized to~$\ket{0}$ (cf. the noiseless protocol embedding described in
%Section~
\cref{sec:nslss}). They measure a suitable number of additional 
EPR pairs to produce a secret key unknown to the adversary.
Using this, they generate common blueberry codes~$\msB_1$, $\msB_2$, $\dotsc$, $
\msB_{N^\prime}$ uniformly and independently from the set of permutations 
over~$\Gamma$.  They also have access to a suitable amount of
classical workspace for local computations required for the simulation.

Alice and Bob repeat the 
following for $i = 1 \cdots \tfrac{N^\prime}{2}$:

\begin{enumerate}
 	\item For~$i = 1$, there is no message to be decoded, and the
                values of the parameters needed for the simulation
                are straightforward. Alice continues with 
step~\ref{step:state}.
		If~$i > 1$, Alice decodes the blueberry encoding of Bob's possibly 
		corrupted last transmission. If she detects an erasure, 
		she sets $M_{i}^{\sA} = 0$, $x_{i}^{\sAD} =  z_{i}^{\sAD} = 0$ and 
		$f_{i-1}^\prime = \perp$ and skips to 
step~\ref{step:tp}. 
		Otherwise, she decodes the transmission as $f_{i-1}^\prime \in 
		\Sigma$, a possibly corrupted version of Bob's last tree 
		encoding $f_{i-1}$, and continues with 
step~\ref{step:next}.

	\item \label{step:next} Alice computes $s_\sB^i = \msD(f_1^\prime \cdots 
		f_{i-1}^\prime)$, and  for~$\ell=1, \dotsc, i-1$ she extracts $b_\ell^i = (x_{\ell }^{i \sBD} 
		z_{\ell}^{ i \sBD}, M_{\ell}^{i  \sB}, x_{\ell}^{i \sBM} z_{\ell 
		}^{i \sBM})$, her best guess for Bob's 
		messages, and the corresponding $c_{\ell}^{i \sB}$, $j_{\ell}^{ 
		i \sB}$.
	\item \label{step:state} Using $s_\sA$, $s_\sB$, she computes her best guess for the state 
		$\ket{\psi_{i}}$ of the joint register, and the 
		corresponding $x_{i}^{ \sAD} z_{i}^{ \sAD}$, $M_{i}^{\sA}$, $c_{i}^{\sA}$, $
		j_{i}^{\sA}$.

	\item \label{step:tp} She completes the teleportation by applying $\rZ^{z_{i}^{ \sAD}} 
		\rX^{x_{i}^{ \sAD}}$ to register $T_\sA^{2(i-1)}$ and swaps this 
		with the $C_\sA$ register.
	\item She tries to make progress in the simulation by applying 
		$U_{j_{i}^{\sA}}^{M_{i}^{\sA}}$ to the $A C_\sA$ register.
	\item She teleports the $C_\sA$ register to Bob using 
		entanglement in register $T_\sA^{2i-1}$ and gets outcomes 
		$x_{i}^{\sAM} z_{i}^{\sAM}$.
	\item Alice updates her history $s_\sA$ by following edge $a_i = 
		(x_{i}^{\sAD} z_{i}^{\sAD}, M_{i}^{\sA}, x_{i}^{\sAM} z_{i}^{\sAM})$, computes 
		$e_i = \msE(a_1 \cdots a_i)$, and transmits the blueberry 
		encoding~$\msB_{2i-1}(e_i)$ of $e_i$ over the noisy channel to Bob.

	\item Upon receiving a possibly corrupted version of Alice's 
		last transmission, Bob decodes the blueberry code layer: 
		he either detects an erasure and sets $e_i^\prime = 
		\perp$, or else decodes the transmission as $e_i^\prime 
		\in \Sigma$, a possibly corrupted version of $e_i$.
	\item Bob computes $x_{i}^{\sBD} z_{i}^{\sBD}$, $M_{i}^{\sB}$ in
the same way as
		Alice, depending on whether or not he detects an erasure. 
		In more detail, if Bob does not detect an erasure, he decodes $s_\sA^i = \msD(e_1^\prime \cdots 
		e_i^\prime)$ and also uses $s_\sB$ to compute the above
parameters. He 
		then performs actions on his registers analogous to Alice's:
		he completes the teleportation step,
		swaps register $T_\sB^{2i-1}$ with $C_\sB$, applies the
operator~$U_{j_i^\sB}^{M_i^\sB}$ to the registers~$BC_\sB$, 
		uses the $T_\sB^{2i}$ register to teleport back the $C_\sB$ 
		register to Alice, computes $f_i$, and transmits the 
		blueberry encoding~$\msB_{2i}(f_i)$ of $f_i$ to Alice. Round $i$ is completed
when Alice receives a possibly corrupted version 
		of this message.
\end{enumerate}

After these $\tfrac{N^\prime}{2}$ rounds, both Alice and Bob extract the output of
the simulation from the $\tilde{A} \tilde{B} \tilde{C}$ registers specified by 
the noiseless protocol embedding.

%------------------------------------------------------------
\subsubsection{Analysis}
	\label{sec:analopt}

As in the proof in 
%Section 
\cref{sec:anal}, the analysis is first carried 
conditioned on some respective views of Alice and Bob of the transcript
at each round. An additional component is the conditioning on some
classical state $z$ of the $Z$ register of the adversary, Eve, and 
the averaging over the shared secret key used for the 
blueberry code. In 
particular, if the adversary has an adaptive and probabilistic strategy, 
we condition on some strategy consistent with the transcript on which we
have already conditioned. We return to this issue later.

We again define a function $P(i)$ such 
that the simulation succeeds whenever $P(\tfrac{N^\prime}{2} + 1) 
\geq N + 1$. Using the notation and the form of the state 
$\ket{\psi_i}$ as in \cref{eq:coll} on the joint register $A B C E$ at the beginning of round 
$i$ (or at the end of round $i-1$), we
let $P(i) = r_i- 2t_i$ (i.e., the same potential function works for the
enhanced simulation as well). We now have three kinds of rounds: \emph{good\/}
rounds, in which both parties decode correctly the other party's 
history; \emph{bad\/} rounds in which at least one party makes a decoding error; and 
\emph{erasure\/} rounds, in which no party makes a decoding error,
but at least one party decodes an erasure from the blueberry 
code. (In an erasure round, the party
detecting an erasure applies the identity operator on
the quantum register before teleporting it back.)

We state an analogue of the technical 
%Lemma 
\cref{lem:pi} 
and its corollary.
\begin{lemma}
	\label{lem:erapi}
	At the end of round $i$, define
\begin{align*}
	N_\rg^i &= |\{j: j \leq i, \text{ round } j \text{ was good}\}|, \\
	N_\rb^i &= |\{j: j \leq i, \text{ round } j \text{ was bad}\}|, \\
 	N_\re^i &= |\{j: j \leq i, \text{ round } j \text{ was an 
		erasure round}\}|.
\end{align*}
Then $P(i + 1) \geq N_\rg^i - 4 N_\rb^i$.
\end{lemma}
The proof of this lemma and its corollary below are omitted since they 
are nearly identical 
to the proofs in the basic simulation. The only difference is if 
at least one party detects an erasure in some 
round, which may be a bad round or an erasure round.
We sketch the argument in the case that round~$i$
is an erasure round. The only unitary operation applied by a party that
detects an erasure, is a Pauli operator on the virtual communication
register~$C$. If both parties detect an erasure, 
$r_{i+1} = r_i$ and~$t_{i+1} = t_i$. If any one party decodes 
correctly and the other detects an erasure, we have $r_{i+1} \ge r_i$
and~$t_{i+1} \le t_i$, so~$P(i+1) \ge P(i)$. (The function increases only if
the party that decoded correctly can apply~$U_{r_i + 1}$ 
or~$\tilde{U}_{t_i}^{-1}$ as defined by the simulation; i.e., that party
holds the registers on which the said unitary operation acts.)
In both cases, the quantity~$N_\rg^i - 4 N_\rb^i = N_\rg^{i-1}
- 4 N_\rb^{i-1} \le P(i)$, so~$P(i+1) \ge N_\rg^i - 4 N_\rb^i$.

\begin{corollary}
	\label{cor:erapi}
	If $P(\tfrac{N^\prime}{2}+1) \geq N + 1$, then the simulation succeeds
with zero error.
\end{corollary}

Hence, it suffices to bound the ratio of bad to good rounds as a 
function of the corruption rate in order to prove the success of the 
simulation.
To do so, we show that depending on a given tolerable error rate 
$\tfrac{1}{2} - \epsilon$, we can vary the distance parameter $\alpha = 
1-\epsilon_\alpha$ of the tree codes used by Alice and Bob, as well as the 
erasure parameter $\beta = 1-\epsilon_\beta$ of the blueberry codes they 
use, and make this ratio as low as desired (except with negligible 
probability in the random choice of the shared secret key used for the 
blueberry code). However, there is now a third kind of round, and we 
would also want to ensure that the ratio of good rounds versus erasure rounds 
does not become arbitrarily low and that 
$P(\tfrac{N^\prime}{2}+1) \geq N + 1$. 

We focus on the numbers~$N_\rg = 
N_\rg^{\tfrac{N^\prime}{2} + 1}$, $N_\rb = N_\rb^{\tfrac{N^\prime}{2} + 1}$ and 
$N_\re = N_\re^{\tfrac{N^\prime}{2} + 1}$ of good, bad and erasure rounds in 
the whole simulation, respectively. To bound the fraction of bad rounds 
as a fraction of the corruption rate, we appeal to a corollary of the 
following technical lemma. The lemma derives a new bound on tree codes with 
an erasure symbol. Since this result only pertains to the structure of such 
codes independent of our application, it might have applications to
classical interactive coding and other settings as well.

\begin{lemma}
\label{lem:optcor}
 	If there is a bound $\delta$ on the fraction of the total number 
of transmissions $N^\prime$ that are corrupted and not detected as erasure 
errors by the blueberry code, then the number $N_\rb$ of bad rounds in the whole 
simulation is bounded as $N_\rb \leq (2 \delta + \epsilon_\alpha) N^\prime$,
where $\epsilon_\alpha = 1- \alpha$, and~$\alpha$ is the distance parameter 
of the tree code with an erasure symbol used by Alice and Bob.
\end{lemma}

\begin{proof}
 	For any $1 \leq i \leq j \leq \tfrac{N^\prime}{2}$, let $I_\re^\sA(i, 
j), I_\rb^\sA(i, j), I_\rg^\sA(i, j)$ be the subset of rounds $i, i+1, \dotsc, 
j-1, j$ in which the symbol that Alice gets from the blueberry decoding is an 
erasure, an error (i.e., an incorrect symbol), or the original encoded 
symbol,
respectively. Note that these are disjoint sets satisfying $I_\re^\sA(i, j) 
\cup I_\rb^\sA(i, j) \cup I_\rg^\sA(i, j) = [i,j]$, where~$[i,j]$
denotes the set~$\{i, i+1, \dotsc, j-1, j \}$. 
Similarly, let
% $J_\re^\sA(i, j),
$J_\rb^\sA(i, j)$ and~$J_\rg^\sA(i, j)$ be the subsets of 
$[i, j]$, respectively, in which the sequence of messages Alice receives
from the tree decoding corresponds to
% a failure (so that $I_\re^\sA (i, j) \subseteq J_\re^\sA (i, j)$),
a decoding error and the correct decoding.
Again note that
%$J_\re^\sA(i, j) \cup
$I_\re^\sA(i, j) \cup J_\rb^\sA(i, j) \cup J_\rg^\sA(i, j) = [i,j]$, 
a disjoint union.
We define analogous subsets for Bob with $\sA$'s 
replaced by $\sB$'s in the notation. Using this notation, we have
\begin{align*}
%  N_\rb & \quad = \quad |J_\rb^\sA (1, \frac{N^\prime}{2}) \cup J_\rb^\sB (1, \frac{N^\prime}{2})| , \quad \text{and} \\
%  |I_\rb^\sA (1, \frac{N^\prime}{2})| &+ |I_\rb^\sB (1, \frac{N^\prime}{2})|  
%		\leq \delta N^\prime,
  N_\rb \quad = \quad \left|J_\rb^\sA \!\!\left(1, \frac{N^\prime}{2}\right)  \cup J_\rb^\sB \!\!\left(1, \frac{N^\prime}{2}\right)\right|,
\quad \text{and} \\
  \left|I_\rb^\sA \!\!\left(1, \frac{N^\prime}{2}\right)\right| + \left|I_\rb^\sB \!\!\left(1, \frac{N^\prime}{2}\right)\right|  \quad \leq \quad
	\delta N^\prime \enspace.
\end{align*}
The statement we wish to prove is
\begin{align*}
  \left|J_\rb^\sA \!\!\left(1, \frac{N^\prime}{2}\right) \cup J_\rb^\sB \!\!\left(1, \frac{N^\prime}{2}\right)\right| 
	& \quad \leq \quad 2 \delta N^\prime + \epsilon_\alpha N^\prime.
\end{align*}
We prove the following stronger statements, which claim that the number
of rounds in which a party makes a tree code decoding error is only
slightly larger than the number of rounds in which that party makes a blueberry code 
decoding error:
\begin{align}
\label{eqn:alice-bound}
  \left|J_\rb^\sA \!\!\left(1, \frac{N^\prime}{2}\right)\right| & \quad \leq
\quad  2 \left|I_\rb^\sA \!\!\left(1, \frac{N^\prime}{2}\right)\right| 
		+ \frac{1}{2} \epsilon_\alpha N^\prime \enspace,
\end{align}
and 
\begin{align*}
  \left|J_\rb^\sB \!\!\left(1, \frac{N^\prime}{2}\right)\right| \quad  \leq
\quad  2 \left|I_\rb^\sB \!\!\left(1, \frac{N^\prime}{2}\right)\right| 
		+ \frac{1}{2} \epsilon_\alpha N^\prime \enspace.
\end{align*}

The proofs of the two statements are similar,
so we only prove the statement for Alice's subsets. To simplify notation, we 
drop the $\sA$ superscripts. For any subset $K$ 
of $[\tfrac{N^\prime}{2}]$ and any two strings~$\bar{e}, \bar{e}^\prime
\in \Sigma^t$ with $\bar{e} = e_1 \cdots e_t$ and $\bar{e}^\prime = e_1^\prime 
\cdots e_t^\prime$, and~$t \le N^\prime/2$, define 
$\Delta_K(\bar{e}, \bar{e}^\prime) = |\{i \in K : i \leq t, e_i \not= 
e_i^\prime \}|$. Note that with $\bar{K} = [\tfrac{N^\prime}{2}] \setminus 
K$, $\Delta (\bar{e}, \bar{e^\prime}) = \Delta_K (\bar{e}, 
\bar{e}^\prime) + \Delta_{\bar{K}}(\bar{e}, \bar{e}^\prime)$, and 
$\Delta_K(\bar{e}, \bar{e}^\prime) \leq |K|$. 

We are now ready to prove the statement~\cref{eqn:alice-bound}. We prove 
by strong induction on the number of rounds~$t$ 
that $|J_\rb (1, t)| \leq 2|I_\rb(1, t)| + \epsilon_\alpha t$. The base case,
$t = 1$, is immediate: in the first round, Alice does not decode any message,
so that the two sets~$J_\rb (1, 1), I_\rb(1,1)$ are empty.
\suppress{if there is no transmission error during the first 
round, then there is no decoding error, and otherwise $1 \leq 2 + 
\epsilon_\alpha$.}

For~$t > 1$, assume that 
\begin{align*}
 |J_\rb (1, j)| \quad \leq \quad 2|I_\rb(1, j )| +
\epsilon_\alpha j \enspace,
\end{align*}
for all~$j$ with~$0 \le j < t$, where we define $J_\rb (1, 0) 
= I_\rb (1, 0) = \emptyset$. 
If in round $t$, $t > 1$, Alice detects an erasure or decodes 
correctly, then the induction step is immediate. Hence, for the induction 
step, we consider the case of incorrect decoding.
Let $\bar{a} \in [d]^t$ be the 
sequence of transmitted messages, $\bar{e} = \bar{\msE} (\bar{a}) \in 
\Sigma^t$ the corresponding sequence of transmissions, $\bar{e}^\prime 
\in \Sigma^t$ the sequence of possibly corrupted receptions, 
$\bar{a}^\prime = \msD(\bar{e}^\prime) \in [d]^t$ the sequence of decoded 
messages, and $\bar{e}^{\prime \prime} = \bar{\msE} (\bar{a}^\prime)$ the
encoding of~$\bar{a}^\prime$ in the tree code. Then, by the decoding condition, $\Delta(\bar{e}^{\prime 
\prime}, \bar{e}^\prime ) \leq \Delta(\bar{e}, \bar{e}^\prime)$. Let 
$\ell = L(\bar{a}, \bar{a}^\prime)$ be the distance of $\bar{a}, 
\bar{a}^\prime$ to their least common ancestor. Then $\Delta_{[1, t - 
\ell]}(\bar{e}^{\prime \prime}, \bar{e}) = 0$, as the encodings have the
same prefix as well. Since~$\bar{e}^{\prime \prime} \neq \bar{e}$, note
that $1 \leq \ell \leq t$. By the induction hypothesis, 
\begin{align*}
 |J_\rb (1, t - \ell)| \quad \leq \quad 2|I_\rb(1, t - \ell)| + \epsilon_\alpha (t - \ell)
\enspace.
\end{align*}
By definition
\begin{align*}
|J_\rb (1, t)| & \quad = \quad |J_\rb (1, t - \ell)| + |J_\rb (t - \ell+1, t)|, \\
|I_\rb(1, t)| & \quad = \quad  |I_\rb(1, t - \ell)| + |I_\rb(t - \ell + 1, t)|,
\end{align*}
so it suffices to prove 
\begin{align}
\label{eqn:alice2}
|J_\rb (t - \ell+1, t)| \quad  \leq \quad  2 |I_\rb(t - \ell + 1, t)| + \epsilon_\alpha \ell
\end{align}
to complete the proof.

Let~$K = I_\re(t - \ell + 1, t)$, the set of rounds in which Alice
detects an erasure. Since codewords in the tree code, in
particular~$\bar{e}^{\prime \prime}$ and~$\bar{e}$, do
not contain the erasure symbol, the decoding
condition~$\Delta(\bar{e}^{\prime \prime}, \bar{e}^\prime) \le
\Delta(\bar{e}^\prime, \bar{e})$ is equivalent
to~$\Delta_{\bar{K}}(\bar{e}^{\prime \prime}, \bar{e}^\prime)
\le \Delta_{\bar{K}}(\bar{e}^\prime, \bar{e})$.  We therefore have
\begin{align}
\nonumber
\Delta(\bar{e}^{\prime \prime}, \bar{e})
    & \quad = \quad \Delta_{K}(\bar{e}^{\prime \prime}, \bar{e}) +
                \Delta_{\bar{K}}(\bar{e}^{\prime \prime}, \bar{e}) \\
\nonumber
    & \quad \le \quad \size{I_\re(t - \ell + 1, t)} +
                \Delta_{\bar{K}}(\bar{e}^{\prime \prime}, \bar{e}) \\
\nonumber
    & \quad \le \quad \size{I_\re(t - \ell + 1, t)} +
                \Delta_{\bar{K}}(\bar{e}^{\prime \prime}, \bar{e}^\prime)
                + \Delta_{\bar{K}}(\bar{e}^\prime, \bar{e}) \\
\nonumber
    & \quad \le \quad \size{I_\re(t - \ell + 1, t)} 
                + 2\, \Delta_{\bar{K}}(\bar{e}^\prime, \bar{e}) \\
\label{eqn:distance}
    & \quad = \quad \size{I_\re(t - \ell + 1, t)} 
                + 2\, \size{I_\rb(t - \ell + 1, t)} \enspace.
\end{align}
On the other hand, the tree code 
distance condition stipulates that~$\Delta(\bar{e}^{\prime \prime},
\bar{e}) \ge \alpha \ell = (1 - \epsilon_\alpha) \ell$ since~$\bar{a} \neq \bar{a}^\prime$. Along
with~\cref{eqn:distance}, this gives
\begin{align}
\label{eqn:length}
\ell \quad \le \quad \Delta(\bar{e}^{\prime \prime}, \bar{e}) +
               \epsilon_\alpha \ell 
     \quad \le \quad \size{I_\re(t - \ell + 1, t)} 
                + 2\, \size{I_\rb(t - \ell + 1, t)} 
                + \epsilon_\alpha \ell \enspace.
\end{align}
We use this to bound the number of bad rounds for Alice, in terms of the
number of blueberry decoding errors she encounters. We have
\begin{align}
\nonumber
\ell & \quad = \quad \size{I_\re(t - \ell + 1, t)}
                + \size{J_\rb(t - \ell + 1, t)} 
                + \size{J_\rg(t - \ell + 1, t)}  \\
\label{eqn:length2}
    & \quad \ge \quad \size{I_\re(t - \ell + 1, t)}
                + \size{J_\rb(t - \ell + 1, t)} \enspace.
\end{align}
Combining~\cref{eqn:length} and~\cref{eqn:length2}, we get the
claimed bound, as in~\cref{eqn:alice2}.
\end{proof}

\begin{corollary}
\label{cor:optcor}
 	If the corruption rate~$c$ of the channel satisfies $0 \leq c < \tfrac{1}{2}$, then 
except with probability smaller than $2^{- \Omega(N^\prime)}$, where 
$N^\prime$ is the length of the simulation protocol, the total number of bad 
rounds in the simulation is bounded as $N_\rb \leq (2 \epsilon_\beta + 
\epsilon_\alpha) N^\prime$, where $\epsilon_\alpha = 1 - \alpha$,
$\alpha$ is the 
distance parameter of the tree code, $\epsilon_\beta = 1- \beta$,
and~$\beta$ is the erasure parameter of the blueberry code.
\end{corollary}

\begin{proof}
 	Suppose that the transmitted symbol is $g_i \in \Gamma$ after a blueberry 
encoding $\msB_j$ (where~$j \in \set{2i-1, 2i}$) and that conditional on her
classical state and some measurement outcomes $z_k$ until round~$i$, Eve
chooses to corrupt $g_i$ into a different $g_i^\prime \in \Gamma$. This 
action is independent of the randomness used in $B_j$, and it 
holds that Pr$[B_i^{-1} (g_i^\prime) \in \Sigma | z_1, \dotsc, z_i] = 
\epsilon_\beta$. This is independent of the classical state and any 
measurement outcome $z_i$ of Eve. We consider two cases. First, suppose
the corruption rate $c$ is bounded as $\epsilon_\beta \leq c < \tfrac{1}{2}$
(so that the corruption rate is at least a constant).
By 
%Lemma 
\cref{lem:bbc}, with probability $1 - 2^{- \Omega 
(N^\prime)}$ at least a $(1 - 2 \epsilon_\beta)$-fraction of the~$c
N^\prime$ corrupted transmissions are detected as erasures. 
So the blueberry decoding gives at most $c N^\prime - c (1 - 
2 \epsilon_\beta) N^\prime = 2 c \epsilon_\beta N^\prime < \epsilon_\beta 
N^\prime$ transmission errors, except with probability negligible 
in $N^\prime$. Taking $\delta = \epsilon_\beta$ in the statement of 
%Lemma 
\cref{lem:optcor} gives us the corollary. If $0 \leq c \leq \epsilon_\beta$, 
then the corollary is immediate from 
%Lemma 
\cref{lem:optcor}, 
with $\delta = \epsilon_\beta$.
\end{proof}

With the above result in hand, we can show that if the corruption rate is 
$\tfrac{1}{2} - \epsilon$ with~$\epsilon > 0$, and we take $\epsilon_\alpha = \tfrac{1}{20} \epsilon, 
\epsilon_\beta = \tfrac{1}{40} \epsilon, N^\prime \geq 
\tfrac{2}{\epsilon} (N + 1)$, then except with negligible 
probability, the simulation succeeds:
\begin{align*}
    P\!\!\left(\frac{N^\prime}{2} + 1 \right)
          \quad &\geq N_\rg - 4 N_\rb \\
		& = \frac{N^\prime}{2} - N_\re - 5 N_\rb 
                    & \textrm{(By 
%Lemma~
\cref{lem:erapi})} \\
		& \geq \epsilon N^\prime  - 5 N_\rb
                    & (\textrm{since } N^\prime/2 = N_\rg +
                       N_\rb + N_\re) \\
 		& \geq \epsilon N^\prime  - 5(2 \epsilon_\beta + 
			\epsilon_\alpha)N^\prime 
                    & (\textrm{since } N_\re \le (1/2 -
                       \epsilon) N^\prime) \\
 		& = N^\prime \left(\epsilon - \frac{10}{40} \epsilon - 
			\frac{5}{20} \epsilon \right) 
                    & \textrm{(By 
%Corollary~
\cref{cor:optcor})} \\
		& = \frac{1}{2} \epsilon N^\prime \\
		& \geq N + 1 \enspace.
\end{align*}
\suppress{
The first inequality is from 
%Lemma 
\cref{lem:erapi}, the first equality is 
by definition of $N_\rg, N_\rb, N_\re$, i.e., $\tfrac{N^\prime}{2} = N_\rg + N_\rb + 
N_\re$, the second inequality is from the fact that the number of erasure 
rounds is bounded by the number of corruption, i.e., $N_\re \leq 
(\tfrac{1}{2} - \epsilon) N^\prime$, and the third inequality is from our 
bound on $N_\rb$ due to 
%Corollary 
\cref{cor:optcor}, which holds except with 
negligible probability.
}
That the simulation succeeds is now
immediate from 
%Corollary 
\cref{cor:erapi}.
 
The above statement holds conditional on some classical state $z$ of the 
$Z$ register of Eve and on some respective views of Alice and Bob of the 
transcript at each round. To prove 
%Theorem~
\cref{th:optach}, we 
argue as in 
%Section 
\cref{sec:anal} in order to translate 
these results into the output state produced by the protocols, even when we consider 
inputs entangled with some reference register $R$. We do not repeat the
whole analysis here, since it is nearly identical to the analysis in
\cref{sec:anal} once we make the 
following observation. An arbitrary adversary Eve fitting the framework of the shared 
entanglement model could have adaptive, probabilistic behaviour based on 
previous measurement outcomes. However, these probabilistic choices are
independent of the secret key generated by Alice and Bob for the 
blueberry code. As in 
%section 
\cref{sec:anal}, the above result holds for each 
probabilistic choice of Eve. Summing over all 
such choices, we obtain the same result, proving 
%Theorem~
\cref{th:optach}.

\section{Results in Other Models}
\label{sec:oth}

By adapting the results in the shared entanglement model 
for an adversarial error model, we can obtain several other interesting 
results. We first complete our study of the shared entanglement model 
with results in a random error setting. We then consider the quantum 
model and obtain results for both adversarial and random error settings. 
We also prove that the 
standard forward quantum capacity of the quantum channels used does not 
characterize their communication capacity in
the interactive communication scenario. Finally, we 
consider a variation on the shared entanglement model in which, along 
with the noisy classical communication, the shared entanglement is also 
noisy.

%------------------------------------------------------------
%------------------------------------------------------------
\subsection{Shared Entanglement Model with Random Errors}

In this section we consider two-party protocols with prior shared
entanglement and classical communication over binary symmetric
channels.
Given a two-party quantum protocol of length $N$ in the noiseless 
model and any $C > 0$, we exhibit a simulation protocol in the shared 
entanglement model that is of length $\rO(\tfrac{1}{C} N)$ and succeeds in 
simulating the original protocol with negligible error over classical 
binary symmetric channels of capacity $C$. More precisely, we have the
following theorem.
\begin{theorem}
\label{th:iidshent}
There exist constants $c, l > 0$ such that given any  $C > 0$ 
and  $N \in 2 \mathbb{N}$, there exists a universal simulator $S$ 
for noiseless quantum protocols of length~$N$ with the following
properties.  The simulator~$S$ is in the shared entanglement model, has
length~$N^\prime$, communication rate $R_\rC\geq l C$, transmission 
alphabet of size $2$, and entanglement consumption rate $R_\rE \leq 6$.
Further, the simulation succeeds with error at most $2^{- c N}$ for all 
noiseless protocols of length $N$ over any classical binary symmetric
channel $\M$ of capacity $C$.
\end{theorem}

We complement this with a lower bound for the communication rate. We
exhibit a sequence of two-party quantum protocols of 
increasing length~$N$ in the noiseless model such that for all $C > 0$, 
any corresponding sequence of simulation protocols of length
${\mathrm o}(\tfrac{1}{C} N)$ 
in the shared entanglement model with classical binary 
symmetric channels of capacity $C$ fails at producing 
the final state with low error on some input. Moreover, the family of quantum 
protocols can be chosen as one that computes a distributed binary function. 
More precisely, we have the next theorem.
\begin{theorem}
\label{th:iidshentopt}
There exists a sequence $\{\Pi_N \}_{N \in 2 \mathbb{N}}$ 
of two-party quantum protocols such that for all~$C > 0$, for any 
simulation protocol $S$ in the shared entanglement model of 
length $N^\prime \in {\mathrm o}(N/C)$ with 
communication rate $R_\rC= \tfrac{N}{N^\prime}$ and arbitrary entanglement 
consumption rate $R_\rE$, the simulation produces an error of at least
$1 - {\mathrm o}(1)$ over binary symmetric channels of capacity $C$.
\suppress{
There exists a sequence $\{\Pi_N \}_{N \in 2 \mathbb{N}}$ 
of two-party quantum protocols and constants $d, \epsilon > 0$ such that 
for all $N_0 \in \mathbb{N}$, there exist $N \geq N_0$ and $C > 0$ such 
that for any $R_\rE \geq 0$ and any simulation protocol $S$ in the shared 
entanglement model of length $N^\prime = \tfrac{d}{C} N$ with 
communication rate $R_\rC= \tfrac{N}{N^\prime}$ and arbitrary entanglement 
consumption rate $R_\rE$, the simulation does not succeed with error 
$\epsilon$ over the binary symmetric channels.
}
\end{theorem}

%------------------------------------------------------------
\subsubsection{Discussion of Optimality}

The above results show that in the regime where we use binary symmetric 
channels of classical capacity close to $0$, we cannot do much 
better than what we achieve, up to a multiplicative constant on top of 
the $\tfrac{1}{C}$ dilation factor. If we want to perform better in that 
regime, we would have to use the specifics of the operations implemented 
by the noiseless protocol instead of using these operations as black-boxes, even 
if we are restricting to protocols computing binary functions. We could,
however, hope to be able to get much better hidden constants, since we 
do not match the case of one-way communication in which the constant can 
be made arbitrarily close to $\tfrac{1}{2}$ as the quantum message size 
increases. Another regime of interest would be one for channels of capacity 
close to $1$, in which our techniques dilate the length of the protocols 
by a large multiplicative constant even when the error rate is low. In 
the classical case, recent results of Kol and Raz \cite{KR13} show how to 
obtain communication rates going to $1$ as the capacity goes to $1$. 
%Using our representation for quantum protocols, we are able to adapt 
%their techniques with ideas similar to those used here to obtain 
%comparable results in the shared entanglement model (up to a factor of 
%$2$ for teleportation), and this result will appear in a forthcoming 
%paper.

%------------------------------------------------------------
\subsubsection{Proof of \cref{th:iidshent}}

In Lemma~2 of Ref.~\cite{Sch96}, it is stated that, given a transmission alphabet 
$\Sigma$, there exists $d > 0$ and $\epsilon \in (0, \tfrac{1}{90})$ such that given a binary 
symmetric channel $\M$ of capacity $C$, there is a $p \in \mathbb{N}$, 
$p \leq d \tfrac{1}{C}$, an encoding function $\msE: \Sigma \rightarrow \{0, 1 
\}^p$ and a decoding function $\msD: \{0, 1 \}^p \rightarrow \Sigma$ such 
that Pr$[\msD( \M (\msE (e))) \not= e] \leq \epsilon$ for all $e \in \Sigma$.
We use this in conjunction with the result of 
%Theorem 
\cref{th:bas} and the Chernoff bound to obtain the following result.
Consider $\epsilon  <  \tfrac{1}{80}$, $\Sigma$ given by 
%Lemma 
\cref{lem:tccode} for a tree code of arity $48$ and distance parameter 
$\alpha = \tfrac{39}{40}$, the corresponding $d > 0$, 
and the length~$N^{\prime \prime} = 4 (1 + \tfrac{1}{N}) N$ of the 
basic simulation protocol over alphabet $\Sigma$
for the length $N$ of the noiseless protocol to be simulated.
Given a binary symmetric channel of capacity $C$ 
and the corresponding $p \in \mathbb{N}$, $\msE$, and $\msD$,
if all the $\Sigma$ 
transmissions in the basic simulation protocol are done by re-encoding 
over $\{0, 1 \}^p$ with $\msE$ (and decoding with $\msD$), then 
$N^\prime = p N^{\prime \prime}$ is the length of the 
oblivious simulation protocol over the binary symmetric channel,
and except with 
probability $2^{- \Omega(N^{\prime \prime})}$,  the error rate for 
transmission of $\Sigma$ symbols is below $\tfrac{1}{80}$.
By 
%Theorem 
\cref{th:bas} the simulation succeeds.

%------------------------------------------------------------
\subsubsection{Proof of \cref{th:iidshentopt}}

It is known that for a classical discrete memoryless channel such as the 
binary symmetric channel, entanglement assistance does not increase the 
classical capacity \cite{BSST02}, and it is also known that allowing for 
classical feedback does not lead to an increase in the classical 
capacity. However, we might hope that allowing for both simultaneously 
might lead to improvements. This is not the case: classical feedback 
augmented by shared entanglement can be seen as equivalent to quantum 
feedback, and it is also known that for discrete memoryless quantum 
channels, the classical capacity with unlimited quantum feedback is equal 
to that with unlimited entanglement assistance \cite{Bow04}. Hence, in 
the shared entanglement model, the classical capacity of the binary 
symmetric channels used is not increased by the entanglement assistance 
and the other binary symmetric channel's feedback. For 
some protocols of length $N$ fitting our general framework in the 
noiseless model, such as those accomplishing a quantum swap function or even 
a classical swap or bitwise XOR functions on inputs of size 
$\tfrac{N}{2}$, the parties effectively exchange their entire inputs 
to produce the correct output. Hence, a dilation factor proportional to the 
inverse of the capacity $\tfrac{1}{C}$ is necessary.
\suppress{since these protocols are equivalent to a communication
of $\tfrac{N}{2}$ bits or qubits in each direction.}
What we wish to prove is even stronger: there exists a family 
of distributed binary functions such that this is necessary. We consider 
the inner product function $\IP_n: \{0, 1 \}^n \times \{0, 1 \}^n 
\rightarrow \{0, 1 \}$, defined as $\IP_n (x, y) = \oplus_{i=1}^n x_i \wedge y_i$, 
which has communication complexity in $\Theta (n)$ in 
both the Yao and the Cleve--Buhrman quantum communication complexity model 
\cite{CvDNT98,NS06}.

By a reduction due to Cleve, van Dam, Nielsen, and Tapp~\cite{CvDNT98},
any protocol evaluating the $\IP_n$ function with small error can be used 
to transmit $n$ classical bits with small probability of error.
Hence, any noise-tolerant simulation of such a protocol over a channel
of classical capacity $C$ can be used to transmit $n$-bit strings
with some small probability of failure. As a consequence,
for small enough error, the simulation requires at least $\tfrac{1}{C} n$ 
uses of the channel. 
\suppress{
Since for any small enough error, the communication 
complexity of $\IP_n$ is $\Theta (n), N \in \Theta (n)$ (if the protocol 
does not waste communication), and $N^\prime \in \Omega (\tfrac{1}{C} n) = 
\Omega (\tfrac{1}{C} N)$ is required for the simulation to succeed.
}
Note 
that we have made the reasonable assumption that we can run the 
simulation backward over the noisy channel at the same communication cost 
or else that we can start with a coherent protocol for the inner product 
function. The restriction of having protocols compute a function in 
a coherent way is natural if we wish to compose quantum simulation 
protocols; then they may be run on arbitrary superpositions of inputs.

%------------------------------------------------------------
%------------------------------------------------------------
\subsection{Quantum Model with Adversarial Errors}

We turn our attention to two-party protocols where there is no prior
entanglement and the communication is over noisy \emph{quantum\/} channels.
Given an adversarial channel in the quantum model with error rate 
strictly smaller than $\tfrac{1}{6}$, we can simulate any noiseless 
protocol of length $N$ over this channel using a number of transmissions
linear in $N$. More precisely, we show the following.
(See~\cref{sec:nqucommqu} for the definition of 
$\A_{\tfrac{1}{6} - \epsilon, q, N^\prime}^\rQ \; $ mentioned in the theorem.)
\begin{theorem}
\label{th:advstand}
There exists a constant $c > 0$ such that 
for arbitrarily small $\epsilon > 0$, there exist a communication rate
$R_\rC > 0$ and an alphabet size $q \in \mathbb{N}$ such that for 
all~$N \in 2 \mathbb{N}$, there exists a universal simulator 
$S$ for noiseless quantum protocols of length~$N$ with the following
properties.
The simulator~$S$ is in the quantum model, has length $N^\prime$,
communication rate at least $R_\rC$, and transmission alphabet size $q$.
Further, the simulation succeeds with error at most
$2^{- c N}$ for all noiseless protocols of length $N$ against 
all adversaries in $\A_{\tfrac{1}{6} - \epsilon, q, N^\prime}^\rQ \; $.
\end{theorem}

%------------------------------------------------------------
\subsubsection{Proof of  \cref{th:advstand}}

The approach we take in the quantum model is to emulate the simulation in 
the shared entanglement model. First, we use the quantum 
channels available to distribute sufficient entanglement.
Alice and Bob can use entanglement to generate a secret key. They then
use the quantum channels
effectively as classical channels along with the entanglement to run the 
simulation protocol from 
%section 
\cref{sec:opt}. Thus the simulation
consists of an entanglement distribution phase, followed by a protocol
implementation phase.

Specifically, suppose we wish to emulate a simulation protocol of 
length~$N'$ in the shared entanglement model. Alice uses $l N^\prime$
transmissions, for a parameter~$l$ to be specified below, to
distribute sufficient perfect entanglement to Bob through the use 
of a quantum error correcting code (QECC).  
(We refer the reader to Ref.~\cite[Chapter 10]{NC00} for the definition of a 
QECC.) They then run the simulation protocol in 
%section 
\cref{sec:opt}.  During this protocol implementation
phase, before transmission and after reception of a quantum register
through the channel, both the sender and the 
receiver measure the register. These measurements 
have the effect of transforming all possible quantum actions of Eve into 
classical actions. Conditioned on the results of the 
two measurements, the corresponding branches of the simulation proceed 
exactly as if the sender and the receiver had transmitted and received 
information over a classical channel. If the size~$q$ of the
communication register is larger than the alphabet size~$\Gamma$ of the
transmissions, and Eve maps some of these classical messages 
outside of $\Gamma$, Alice and Bob mark these as 
erasures. So Eve does not gain anything by introducing errors outside
$\Gamma$. 

We start by pinning down the parameters of the 
QECCs needed to distribute the necessary amount of entanglement.
In the interest of simplicity, we do not attempt to optimize the 
parameters involved.

For a given $\epsilon > 0$, let $s = 
\tfrac{(|\Gamma|!)}{(|\Gamma|-|\Sigma|)!}$ be the size of the shared 
secret key used to do the blueberry encoding in each round of the
simulation in 
%Section~
\cref{sec:opt}.
Two maximally entangled states of size $2 s$, i.e., states of 
the form $\sum_{j=0}^{2 s - 1} \ket{j}^{T_\sA} \ket{j}^{T_\sB}$, are used to 
generate the secret keys and to create the EPR pairs 
required for teleportation in every round. 
For a given size $q$ for the communication register, and for a simulation 
protocol in the shared entanglement model of length $N^\prime$, we 
distribute a maximally entangled state over $ N^\prime \log_q (2s)$ 
registers of size~$q$.

In the entanglement distribution phase of the simulation in the quantum
model, we encode the~$N^\prime \log_q (2s)$ registers into $l N^\prime$ 
registers of size $q$. For the encoding, we use a quantum error correcting 
code with alphabet size $q$, transmission rate $R_\rQ \ge \tfrac{1}{l}
\log_q (2 s)$, and maximum tolerable error rate $\delta$ to be determined
shortly.
We only consider exact QECCs, but the analysis 
extends to approximate ones. (Approximate error correction allows for
some deviation from perfect transmission.) 

To determine the relationship
between $q, l,$ and $\delta$ required for the simulation to succeed,
we first note that in the protocol implementation phase (the second
phase of the simulation), we
transmit classical messages chosen from a set of size $|\Gamma|$ over 
the quantum channel. For simplicity, we choose~$q \geq |\Gamma|$.
To ensure that this second phase succeeds, the number of corruptions 
in it should be bounded by $(\tfrac{1}{2} - \epsilon) N^\prime$.
An adversary could choose to put all of the allowed corruptions in the 
first (entanglement distribution) phase, so the QECC should be able 
to recover from the same 
number of errors. In other words, we require~$\delta l N' \ge 
\tfrac{N^\prime}{2} - \epsilon N^\prime$. The length of the message
in the entanglement distribution phase satisfies~$l \geq \tfrac{1 - 2 
\epsilon}{2 \delta}$. In summary, the entire simulation tolerates
$\tfrac{N^\prime}{2} - \epsilon N^\prime$ adversarial errors during a 
total of $(l + 1) N^\prime$ transmissions of size~$q$ registers
provided a suitable QECC exists. The error rate tolerated is
$\tfrac{1 - 2\epsilon}{2 (l + 1)}$.

The above analysis applies to the oblivious communication model.
If we restrict ourselves to the alternating communication model,
we have twice as much communication, i.e., $2 l N'$ size-$q$
registers, in the entanglement transmission phase.
The adversary can choose to corrupt the transmissions of one party
alone, so~$l \geq \tfrac{1 - 2  \epsilon}{2 \delta}$ as before.
The total number of transmissions is, however, $(2l + 1)N^\prime$, so the
error rate tolerated is $\tfrac{1 - 2\epsilon}{2 (2l + 1)}$.

We now appeal to a high-dimensional quantum Gilbert--Varshamov bound \cite{AK01, 
FM04} stating that for arbitrarily small $\epsilon^\prime > 0$, there 
exist strictly positive communication rate $R_\rQ > 0$ and large enough 
transmission alphabet size such that families of quantum codes of 
arbitrarily large length exist which can tolerate a fraction $\tfrac{1}{4} 
- \epsilon^\prime$ of errors and allow for perfect decoding of the quantum 
state. Using these codes with~$\epsilon^\prime = \epsilon$, we 
get~$\delta = \tfrac{1}{4} - \epsilon$, $l \ge
\tfrac{1-2 \epsilon}{2\delta} = \tfrac{2(1-2 \epsilon)}{1- 4\epsilon}$
and net error rate~$\tfrac{1 - 2\epsilon}{2 (2l + 1)} = \tfrac{(1-2
\epsilon)(1- 4\epsilon)}{6 - 16\epsilon} \ge
\tfrac{1}{6} - \epsilon$ that the simulation protocol can tolerate 
in an oblivious model of communication. In an alternating model 
of communication, we are able to tolerate an error rate of $\tfrac{1}{10} - 
\epsilon$.

The above choice of parameters ensures that the error rate in the 
entanglement distribution phase is bounded by~$\tfrac{1}{4} - \epsilon$,
and the received quantum state can be decoded perfectly.
This establishes a shared maximally entangled
state of the required dimension. Moreover, 
the corruption rate of the adversary during the 
protocol implementation phase is lower than $\tfrac{1}{2} - \epsilon$. 
Recall that Alice and Bob measure the states received over the quantum
channel in the standard basis to convert it to a classical channel.
Given any strategy of the adversary, which is necessarily independent 
of the secret key used for the blueberry codes, for any choice of
measurement outcomes for Alice and Bob, the simulation 
succeeds with probability exponentially close to~$1$ (in terms
of~$N^\prime$). The remainder 
of the analysis follows that in  \cref{sec:analopt}, proving  
\cref{th:advstand}.

%------------------------------------------------------------
\subsubsection{Discussion of Optimality}

If we consider only perfect QECCs for quantum 
data transmission, it is known that we cannot tolerate error rates of 
more than $\tfrac{1}{4}$ asymptotically. With the approach of first 
distributing entanglement and then using the $\tfrac{1}{2} - \epsilon$ 
error rate simulation protocol in the shared entanglement model, we get
an overall tolerable error rate for the 
simulation of less than $\tfrac{1}{6}$. Cr\'epeau, Gottesman, and 
Smith \cite{CGS05} showed how we can tolerate an error rate
up to $\tfrac{1}{2}$ asymptotically for data transmission if we 
consider approximate QECCs.  Using these, we could get a tolerable 
error rate of $\tfrac{1}{4} - \epsilon$ for a two phase simulation protocol as 
described above. However, their register size, as well as the number of 
communicated registers, is linear in the number of transmitted qubits in
the original protocol.  This would lead to a communication rate of $0$ 
asymptotically in the simulation.
It would be interesting to see whether we can do 
something similar with register size independent of the transmission 
size, but possibly dependent on the fidelity we want to reach and how 
close to $\tfrac{1}{2}$ (or some other fraction strictly larger than 
$\tfrac{1}{4}$) we would like the tolerable error rate to be. Using this
kind of code, if we break up the simulation into two phases---an 
entanglement distribution part and then a protocol implementation 
part---the above is the best we can do. We might hope to develop a fully 
quantum analogue of tree codes that does not entail the two phase 
simulation, in order to achieve higher error rates. The putative quantum 
codes would require some properties for fault-tolerant computation, 
so that we may coherently apply the noiseless protocol unitary operations 
in the simulation.  This issue does not occur in the fully classical
setting, since we can copy classical information and perform the computation 
on the copy. 

Finally, we note that the proof of 
%Theorem 
\cref{th:optobl} applies here as 
well. It establishes a bound of $\tfrac{1}{2}$ 
on the maximum error rate tolerable in an oblivious communication model,
that is, no simulation protocol in the quantum model can 
succeed with arbitrarily small error against all adversaries in 
$\A_{\tfrac{1}{2}, q, N^\prime}^\rQ$ for any $q, N^\prime \in \mathbb{N}$.
(See~\cref{sec:nqucommqu} for the definition of  
$\A_{\tfrac{1}{2}, q, N^\prime}^\rQ \; $.)

%------------------------------------------------------------
%------------------------------------------------------------

\subsection{Quantum Model with Random Errors}

We shift our focus to quantum communication over depolarizing channels.
Given a two-party quantum protocol of length $N$ in the noiseless 
model and any $C_Q > 0$, we devise a simulation protocol in the quantum 
model that is of length $\rO(\tfrac{1}{C_Q} N)$ and succeeds in simulating the 
original protocol with arbitrarily small error over quantum depolarizing 
channels of quantum capacity $C_Q$. (We refer the reader to Ref.~\cite[Chapter 23]{Wilde11} for the definition of quantum capacity $C_Q$.) More precisely,
we state the following theorem.
\begin{theorem}
\label{th:iidstand}
There exist a constant $l > 0$ and a function $f: \mathbb{N} 
\rightarrow \mathbb{R}^+$ with $\lim_{N \rightarrow \infty} f(N) 
= 0$ such that given any  $C_Q > 0$ and  
$N \in 2 \mathbb{N}$, there exists a universal simulator~$P$ 
for noiseless quantum protocols of length~$N$ with the following
properties.
The simulator~$P$ is in the 
quantum model, has length $N^\prime$, communication rate $R_\rQ \geq l 
C_Q$, and transmission alphabet size $2$.
Further, the simulation succeeds with error at most $f(N)$ in
simulating all noiseless protocols of length $N$ over depolarizing 
channel $\M$ of quantum capacity $C_Q$.
\end{theorem}

We point out that quantum capacity with feedback is a lower bound on the
dilation needed to simulate protocols over depolarizing channels.
There exist a sequence of two-party quantum protocols of 
increasing length $N$ in the noiseless model such that for all $C_Q^\rB > 0$, 
any corresponding sequence of simulation protocols of length 
${\mathrm o}(\tfrac{1}{C_Q^\rB} N)$ in the quantum model with 
quantum depolarizing channels of 
quantum capacity~$C_Q^\rB$ with classical feedback fails at producing the final 
state with low error on some input. (We refer the reader to Refs~\cite{BDS97, LLS09} for definitions of quantum capacity with classical feedback $C_Q^\rB$ and quantum capacity with free assistance by two-way classical communication $C_Q^2$.) Moreover, the family of 
quantum protocols can be chosen as one computing a distributed binary 
function. 
\begin{theorem}
\label{th:iidstandconv}
There exists a sequence $\{\Pi_N \}_{N \in 2 
\mathbb{N}}$ of two-party quantum protocols 
%and constants $d, \epsilon > 0$ 
such that for all 
%$N_0 \in \mathbb{N}$, there exists $N \geq N_0$ and 
$C_Q^\rB > 0$,
% such that 
for any simulation protocol $P$ in the quantum model 
of length $N^\prime \in {\mathrm o}(N/C)$
%= \tfrac{d}{Q_\rB} N$ 
with communication rate $R_\rQ = 
\tfrac{N}{N^\prime}$, the simulation produces an error of at least
%$\epsilon$ 
$\Omega (1)$
over quantum depolarizing channels of quantum capacity~$C_Q^\rB$ with classical feedback.
\end{theorem}
 
It turns out that quantum capacity does not capture the ability to
transmit information in an interactive setting.
Given a two-party quantum protocol of length $N$ in the noiseless 
model, there exist a quantum depolarizing channel of unassisted forward 
quantum capacity $C_Q=0$ and a simulation protocol in the quantum model 
with asymptotically positive rate of communication which succeeds in 
simulating the original protocol with arbitrarily small error over that 
quantum channel.
\begin{theorem}
\label{th:iidstandopt}
There exist constants $c, R_\rQ > 0$ such that given  any $N \in 2  
\mathbb{N}$, there exists a universal simulator $P$ 
for noiseless quantum protocols of length~$N$ with the following
properties.
The simulator~$P$ is in the quantum model, has length $N^\prime$,
communication rate at least $R_\rQ$, and transmission 
alphabet size $2$.
Further, the simulation succeeds with error at most~$2^{- cN}$ at
simulating all noiseless protocols of length $N$ over a 
particular depolarizing quantum channel $\M_0$ of forward 
quantum capacity $C_Q = 0$.
\end{theorem}

%------------------------------------------------------------
\subsubsection{Proof of  \cref{th:iidstand}}

For the case of random error in the quantum model, we use techniques similar 
to the case of adversarial error. Indeed, we split the protocol into two 
phases: an entanglement distribution phase and a protocol implementation
phase. 

It suffices to adapt the result from 
%section 
$\cref{sec:bas}$ for a basic simulation protocol of length 
$N^{\prime \prime}$ over some large alphabet $\Sigma$. We then need only 
distribute $N^{\prime \prime}$ maximally entangled states of the
appropriate size. For any depolarizing channel 
of quantum capacity $C_Q > 0$, we use standard coding results from 
quantum Shannon theory~\cite{Wilde11} 
to distribute entanglement at a rate of $\tfrac{d}{C_Q}$ for some $d 
> 0$ with low error. Then, for the protocol implementation phase, we appeal
to two properties.
First, the classical capacity $C$ of a quantum channel is at least as 
large as its quantum capacity. Second, a classical 
capacity achieving strategy for the depolarizing channel is to 
simulate a binary symmetric channel (BSC) of capacity $C$ for each 
transmission by measuring the output in the computational basis, and then 
to block code over the corresponding BSC (see, e.g., 
Ref.~\cite{Wilde11} for details). We can then translate the proof of 
%Theorem 
\cref{th:iidshent} in order to design our classical strategy. This succeeds 
with overwhelming probability assuming perfect entanglement, 
and the output is arbitrarily close to the noiseless protocol output. 
Combining the bound on the error from the two phases,
the simulation can be made to succeed with error less 
than $f(N)$ over the depolarizing channel of quantum capacity $C_Q$, for 
some function $f: \mathbb{N} \rightarrow \mathbb{R}^+$ which 
asymptotically goes to zero. 

%------------------------------------------------------------
\subsubsection{Proof of  \cref{th:iidstandconv}}

The idea for this proof is to use the fact that distributing an EPR pair over a quantum depolarizing channel
produces a Werner state, which is symmetric under the interchange of Alice and Bob (see 
%section 
\cref{sec:nsent} for a definition of Werner states). 
Moreover, if Bob uses the free classical feedback to teleport to Alice with these Werner states, this creates a virtual
depolarizing channel from him to Alice, with the same parameter as the actual channel from Alice to him.
Hence, a quantum depolarizing channel from Alice to Bob along with free classical feedback is sufficient
to simulate depolarizing channels in both directions, and the total number of uses of the depolarizing channel is the same in both cases. 

Similar to what was argued in the proof of 
%Theorem 
\cref{th:iidshentopt} for classical communication, there exist protocols of 
length $N$ that fit our general framework in the noiseless model and
can be used to communicate up to $\tfrac{N}{2}$ qubits in each direction.
Hence, since our simulation protocols of length $N^\prime$ can be simulated by 
$N^\prime$ uses of a depolarizing channel from Alice to Bob supplemented 
by classical feedback from Bob to Alice, we cannot have a rate of 
communication better than $\tfrac{N}{2 C_Q^\rB}$ for small enough error. 

To prove that a protocol to compute a binary function is sufficient, we once 
again consider the inner product function $\IP_n$. 
We apply a coherent version of the idea to use 
the inner product protocol to communicate, as in the proof of 
%Theorem 
\cref{th:iidshentopt}. This allows us  to use the depolarizing channel to distribute 
quantum entanglement, and then also to teleport
%further use the depolarizing channel 
(again with the inner product protocol used this time to communicate 
classical information). 
For this, it is sufficient to note that what we 
achieved in the proof of 
%Theorem 
\cref{th:iidshentopt} using the protocol 
for $\IP_n$ is actually stronger than $\Theta (N)$ bits of classical 
communication: we had a coherent bit channel \cite{Har04} for $\Theta 
(N)$ cobits (coherent bits), which can be used to distribute $\Theta (N)$ 
ebits (EPR pairs).
Note 
that we once again make the reasonable assumption that we can run the 
simulation backward over the noisy channel at the same communication cost 
or that we can start with a coherent protocol for the inner product 
function.

%------------------------------------------------------------
\subsubsection{Proof of \cref{th:iidstandopt}}

The case of the depolarizing channel requires some technical work, so for 
simplicity we first consider the case of the quantum erasure channel. For 
the quantum erasure channel, we use the property that, for erasure 
probability $\tfrac{1}{2} \leq p < 1$, the (forward, unassisted) quantum 
capacity is $0$ while both the classical capacity  and the 
entanglement generation capacity with classical feedback equal
$1-p$~\cite{BDS97}. Moreover, the feedback required to achieve this bound is only one 
message of length linear in the size of the quantum communication. The 
strategy we use is the following: for a basic simulation protocol of 
length $N^{\prime \prime}$ over $\Sigma$, Alice distributes $N^{\prime 
\prime}$ EPR pairs
to Bob by sending $\tfrac{4 N^{\prime \prime}}{(1-p)}$ 
halves of such states over the quantum erasure channel. Then, except with 
negligible probability, at least $N^{\prime \prime}$ of them are received 
intact, and Bob knows which ones these are. The feedback consists of 
informing Alice which $N^{\prime \prime}$ pairs were received intact and
can be used in the protocol. This can be done over the 
quantum erasure channel, with probability negligibly smaller than 1, with a 
classical message of length linear in $N^{\prime \prime}$. 

Then, given a message set $\Sigma$ we can use the quantum erasure channel 
a constant number of times to decrease the probability of error in a 
classical transmission of any symbol $e \in \Sigma$ below $\tfrac{1}{90}$. 
Except with negligible probability, the fraction of $N^{\prime \prime}$ 
transmissions of symbols of $\Sigma$ transmitted in this way is below 
$\tfrac{1}{80}$. We can then use a reasoning similar to that in the proof of 
%Theorem 
\cref{th:advstand} to argue that the output is arbitrarily close 
to the noiseless protocol output. 

Now for the depolarizing channel, the reasoning is mostly the same, but we 
have to work harder to obtain (almost) noiseless entanglement. The 
unassisted forward capacity of the depolarizing channel is shown in 
Ref.~\cite{BDSW96} to be equivalent to one-way entanglement distillation 
yield. To separate one-way and two-way entanglement distillation, they 
use a combination of the recurrence method of Ref.~\cite{BBPSSW95} along 
with their hashing method. The recurrence method is an explicitly two-way 
entanglement distillation protocol, which can purify highly noisy 
entanglement but does not have a positive yield in the limit of high 
fidelity distillation.  The hashing method is a 
one-way protocol with positive yield in the perfect fidelity limit, but 
which does not work on highly noisy entanglement. We cannot hope 
to use this strategy to distill near-perfect EPR pairs in our scenario 
since the hashing method as they describe it requires too much 
communication. (We could probably use a derandomization 
argument to avoid communicating the random strings in this protocol.)
To reduce the 
communication cost, we instead use a hybrid approach of entanglement 
distillation followed by quantum error correction. 

Starting with a depolarizing channel with depolarizing parameter as high 
as possible, but still low enough to have $C_Q = 0$, we use it to 
distribute imperfect EPR pairs. This yields (rotated) Werner states with 
the highest possible fidelity to perfect EPR pairs, but such that one-way 
entanglement distillation protocols cannot have a positive yield of EPR 
pairs while two-way entanglement distillation protocols can. (See 
%section 
\cref{sec:nsent} for a definition of Werner states.) We then do one round 
of the recurrence method for entanglement distillation to obtain a lesser 
number of Werner states of higher fidelity to perfect EPR pairs, and so 
we could now use one-way distillation protocols on these to obtain a 
positive yield of near-perfect EPR pairs. The amount of 
classical communication required up to this point is one message from 
Alice to Bob of linear length informing him of her measurement outcomes, 
and then one classical message of linear length from Bob to Alice 
informing her which states to keep as well as which rotation to apply to 
these. (The rotation takes the states back to the symmetric Werner form;
$\log{12}$ bits of 
information per pair is sufficient for this purpose \cite{BDSW96}.)
We now use these EPR pairs along with teleportation to effectively 
obtain a depolarizing channel of quantum capacity $C_Q > 0$. We use 
standard coding from quantum Shannon theory~\cite{Wilde11}
over this quantum channel to 
distribute $N^{\prime \prime}$ near-perfect EPR pairs. This new step
only requires a linear amount of classical communication. After 
the initial very noisy entanglement distribution step, we thus only have three 
classical messages to send over the depolarizing channel of classical 
capacity $C > 0$. We generate near-perfect entanglement using 
the depolarizing channel a linear number of times, and then go on to 
the protocol implementation phase as before. Note that we are not yet 
guaranteed
an exponential decay of the error at this point, but only that the error 
tends to zero in the limit of large $N$. To get exponential decay in
error, we adapt 
the above protocol. Before using teleportation and QECCs to 
distribute near-perfect entanglement, we perform a few more rounds of the 
recurrence method until the Werner states reach fidelity parameter above 
$0.82$. Except with negligible probability, starting with some linear 
number of noisy EPR pairs, after a constant number of rounds of the 
recurrence method, we are left with sufficiently many less noisy EPR 
pairs for our next step. At this point, it is known that there exist 
stabilizer codes achieving the hashing bound (which has strictly positive 
yield for this noise parameter) and which have negligible error. Using 
the property that some classical capacity achieving strategy for the 
depolarizing channel also has negligible error, we get the stated 
exponential decay in the error. 

%------------------------------------------------------------
\subsubsection{Discussion of Optimality}

It is known that for some range of the depolarizing parameter, the 
quantum capacity~$C_Q^\rB$ with classical feedback of the depolarizing 
channel is strictly larger than its unassisted forward quantum capacity 
$C_Q$ \cite{BDSW96}. In particular, there exist values for which $C_Q=0$ but 
$C_Q^\rB > 0$. A careful analysis of the related two-way entanglement 
distillation protocols (in particular their communication cost and their 
amount of interaction) reveals that there is some range of the 
depolarizing parameter for which we can achieve successful simulation 
even though $C_Q=0$, by using the depolarizing channels in each direction 
to transmit classical information. 
This proves that the 
standard forward quantum capacity of the quantum channels used does not 
characterize their communication capacity in
the interactive communication scenario.
Note that $C_Q^\rB > 0$ if and only if 
the depolarizing parameter $\epsilon^\prime < \tfrac{2}{3}$, and so $C_Q^\rB > 
0$ if and only if the quantum capacity assisted by two-way classical 
communication $C_Q^2 > 0$. In the case where we are given a depolarizing 
channel with $C_Q^\rB > 0$, we can modify the method used in the proof of 
%Theorem 
\cref{th:iidstandopt}. We iteratively use the recurrence method a 
constant number of times on the noisy distributed EPR pairs, until the 
depolarizing channels induced through teleportation over the noisy 
distilled EPR pairs have non-zero forward quantum capacity.
(Here the constant depends on the depolarizing parameter, but not on~$N$.)
Then we distribute entanglement over the induced channels 
using standard QECCs. We achieve asymptotically positive rates of 
communication for our simulation protocols. It is an interesting open 
question whether we can close the gap between our lower and upper bounds 
and always achieve successful simulation at a rate $\rO(\tfrac{1}{C_Q^\rB} N)$. 
The separation result regarding the forward, unassisted quantum capacity 
of the depolarizing channel requires some technical work, but the case of 
the erasure channel already makes it clear that in general for discrete 
memoryless quantum channels, the unassisted forward quantum capacity is 
not the most suitable quantity to consider in the setting of interactive 
quantum communication. 
%A tighter characterization of the interactive
%quantum capacity of the depolarizing channel will be studied in a forthcoming paper.

%------------------------------------------------------------
%------------------------------------------------------------
\subsection{Noisy Entanglement}
	\label{sec:nsent}

The last model we consider is a further variation on the shared 
entanglement model, in which, along with the noisy classical links 
between the honest parties, the entanglement these parties share is also 
noisy.

There are many possible models for noisy entanglement; we consider a 
simple one in this section, in which parties share noisy EPR pairs 
instead of perfect pairs. Following Ref.~\cite{BBPSSW95}, we consider the 
so-called (rotated) Werner states $W_F = F \kb{\Phi_{00}}{\Phi_{00}} + 
\tfrac{1 - F}{3} (\kb{\Phi_{01}}{\Phi_{01}} + \kb{\Phi_{10}}{\Phi_{10}} + 
\kb{\Phi_{11}}{\Phi_{11}})$, which are mixtures of the four Bell states 
parametrized by $0 \leq F \leq 1$. Note that these are the result of 
passing one qubit of an EPR pair through a $\T_{\epsilon^\prime}$ 
depolarizing channel, for $F = 1 - \tfrac{3 \epsilon^\prime}{4}$. The 
purification of these noisy EPR pairs is given to Eve. We use the result 
of Ref.~\cite{BBPSSW95} to show that for any $F > \tfrac{1}{2}$, simulation 
protocols with asymptotically (in $N \rightarrow \infty$, not in $F 
\rightarrow \tfrac{1}{2}$) positive communication rates and which can 
tolerate a positive error rate can succeed with asymptotically zero 
error. This is optimal since at $F = \tfrac{1}{2}$, Werner states are 
separable, so there is no way to use them in conjunction with classical 
communication to simulate quantum communication.

%------------------------------------------------------------
\subsubsection{Adversarial Errors in the Classical Channel}

We first consider the case of adversarial errors. Let $l_c$ be the number 
of rounds of the recurrence method~\cite{BBPSSW95} for entanglement distillation 
necessary to reach the $F = 0.82$ bound. This number is independent of 
$N$, and depends only on the initial value of the parameter~$F$. As described in the proof of 
%Theorem 
\cref{th:iidstandopt}, each round of the recurrence method only 
requires a linear length classical message in each direction. After this bound is 
reached, one last linear length classical message is sufficient to 
generate a linear amount of entanglement through teleportation via an 
induced depolarizing channel of non-zero quantum capacity~$C_Q$. Standard 
quantum error correction techniques enable us to extract near-perfect 
entanglement at this point. Once we have near-perfect entanglement, we 
can use techniques from the basic simulation protocol to perform 
successful simulation of noiseless protocols and hence achieve our goal. 
The protocol sketched above requires the communication of $2 l_c + 1$ 
messages to distill near-perfect entanglement, independent of $N$, 
followed by a phase of simulating the message transmissions from the
original protocol. The simulation protocol
tolerates a constant error rate, though inversely proportional to $l_c$.
It requires a constant rate of noisy entanglement consumption, which is 
exponential in $l_c$ since each round of the recurrence method consumes 
at least half of the noisy EPR pairs. The protocol has a constant, positive 
rate 
of communication, though inversely proportional to the number of consumed 
noisy EPR pairs.

%------------------------------------------------------------
\subsubsection{Random Errors in the Classical Channel}

The case of noisy communication through binary symmetric channels once 
again is immediate from the adversarial error case by a concentration 
of measure argument. The communication rate of the resulting protocol 
is inversely proportional to the 
classical capacity $C$, and also to the number of noisy EPR pairs 
consumed.

\section{Conclusion: Discussion and Open Questions}

In this work,
we proposed a simulation of interactive quantum protocols intended 
for noiseless communication over noisy channels. Our approach 
is to replace irreversible measurements by reversible pseudo-measurements 
in the Cleve--Buhrman model, i.e., the model with shared entanglement and classical 
communication. Then, in the noisy version of the model, we teleport back and forth the 
corresponding quantum communication register to avoid losing quantum 
information. 
We develop a representation for such noisy quantum protocols that gives
an analogue of Schulman's protocol tree representation for classical protocols.
We prove that with this approach,
it is possible to simulate the evolution of 
quantum protocols designed for noiseless quantum channels over noisy classical
channels with only a linear dilation factor.

In the case of 
adversarial channel errors in which the parties are allowed to pre-share a 
linear amount of entanglement, we prove that the error rate 
of $\tfrac{1}{2} - \epsilon$ that our simulation tolerates is 
optimal unless we allow adaptive protocols. (An adaptive protocol is a 
generalization of the noisy communication model wherein the
order in which the parties take turns speaking can be adapted to 
the errors.) In a noisy setting, restricting to non-adaptive (oblivious) 
protocols 
seems natural. Adaptive protocols run the risk of entering a deadlock:
depending on the particular view of each party of
the evolution of the protocol due to previous errors, the parties could disagree 
on whose turn it is to speak. This would result in protocols that are 
not well defined.

To get the tolerable error rate as high as $\frac{1}{2} - \epsilon$, we develop
new techniques along with a new bound on tree codes with an erasure symbol, 
%Lemma 
\cref{lem:optcor}.
To simplify the exposition, we chose not to optimize the parameters in 
our simulation protocol such as communication and entanglement consumption 
rates, or the size of the communication register.

We adapt our 
findings to a random error model in which parties are allowed to share 
entanglement but communicate over binary symmetric channels of non-zero
capacity $C$. We obtain communication rates proportional to $C$. We 
show that, up to a hidden constant, this is optimal for some family of 
distributed binary functions, for example the inner product functions 
$\IP_n : \{0, 1 \}^n \times \{0, 1 \}^n \rightarrow \{0, 1 
\}$, defined as $\IP_n (x, y) = \oplus_{i=1}^n x_i \cdot y_i$.
\suppress{
\mbox{\textsl{IP}}_n : \{0, 1 \}^n \times \{0, 1 \}^n \rightarrow \{0, 1 
\}$, defined as $\mbox{\textsl{IP}}_n (x, y) = \oplus_{i=1}^n x_i \cdot y_i$.
}
Our 
findings can also be adapted to obtain similar (though not optimal) 
results for the quantum model (the noisy version of Yao's model). Here, 
the simulation protocols run in two phases. In the first, 
a preprocessing phase,
a linear amount of entanglement is distributed with standard techniques 
from quantum Shannon theory for random noise and from quantum coding 
theory for adversarial noise. This is followed by a simulation phase in which 
the actions of the parties parallel those in the shared entanglement 
model. In the case of adversarial noise, we show that we can tolerate an 
error rate of $\tfrac{1}{6} - \epsilon$ in the quantum model. In the case
of random noise in which the parties communicate over depolarizing 
channels of capacity $C_Q > 0$, we obtain rates proportional to $C_Q$. 
Perhaps surprisingly, we show that the use of depolarizing channels 
in both directions enables the simulation to succeed even for some
quantum channels of unassisted forward quantum capacity $C_Q = 0$.
This proves that $Q$ does not characterize a quantum channel's 
capacity for interactive quantum communication.
We extend our ideas to perform simulation in an extension of the 
shared entanglement model in which not only the classical 
communication is noisy but also the entanglement.

A direction of research that immediately grows out of this work is 
characterizing the communication rates in all of the models 
discussed. In particular, 
the precise interactive capacity of the depolarizing channel with a specified 
noise parameter remains open. The question of 
interactive capacity for the binary symmetric channel was raised  in the 
classical context by Schulman \cite{Sch96} and brought to attention 
recently by Braverman in a survey article on the topic of 
interactive coding \cite{Bra12}. Recent developments provide
tight lower and upper bounds for this quantity \cite{KR13}. In the 
classical setting, a particular problem with worst-case interaction of 
one-bit transmissions to which all classical interactive protocols can be 
mapped was proposed for the study of such a quantity. Since every 
interactive quantum protocol can be mapped onto our general problem, 
it would be natural to
study such a quantity in the quantum domain. Would the interactive 
capacity of the binary symmetric channel (with entanglement assistance) 
for quantum protocols be the same as that for classical protocols \cite{KR13}, 
up to a factor of~$2$ for teleportation?
%We show in upcoming articles
%that for bit flip probability $\epsilon$, the lower bound of $\tfrac{1}{2} 
%- \rO(\sqrt{\rH(\epsilon)} \, )$ holds, and even extends to a lower bound 
%of $1- \rO(\sqrt{\rH(\epsilon)} \, )$ for depolarizing channels.
Do the techniques developed in Ref.~\cite{KR13} adapt to the quantum 
setting to obtain an upper bound 
of $\tfrac{1}{2} - \Omega(\sqrt{\rH(\epsilon)} \, )$? What about the 
depolarizing channel and other channels?

Another question that remains open is 
finding the highest adversarial error rate that can be 
withstood in the quantum model. To study this question, it is likely
that a ``fully quantum'' approach with new kinds
of quantum codes is needed. In particular, ideas from 
fault-tolerant quantum computation might be necessary.
Furthermore, the important question of integrating our 
results into a larger fault-tolerant framework, in which the local 
operations are also noisy, remains open. 
Yet another important question for interactive quantum coding
is what would happen in a shared entanglement setting if, along 
with the noisy classical communication, the entanglement provided were 
also noisy; we investigated this question for a depolarizing noise model 
for the entanglement, but other models would also be interesting 
to study, in particular, adversarial noise on the shared EPR 
pairs above the unidirectional binary error rate limit. Note that below 
that bound, standard quantum error correction for qubits with 
teleportation can be used for distillation. Finally, the question of 
computationally efficient simulation also remains open.
%, and we will show in 
%upcoming works how to adapt the techniques developed by Brakerski and 
%Kalai \cite{BK12} to our setting to efficiently process the classical 
%communication in our simulation protocols.

\section*{Acknowledgments}
The authors are grateful to Louis Salvail, Benno Salwey and Mark M.\ 
Wilde for useful discussions.

\appendix

%------------------------------------------------------------
%\subsubsection{Noisy Communication Model}
\section{Formal Definitions for Noisy Communication Model}

%
%\paragraph{Quantum model.}
\subsection{Quantum Model}
\label{sec:nqucommqu}

For the \emph{quantum model}, Alice possesses a local quantum 
register $ A^\prime$  which contains five subsystems of interest: to 
implement a noiseless protocol $\Pi$ as a black-box, the $A$ and $C_\sA$ parts 
correspond to the registers of the noiseless communication protocol, 
while $\tilde{A}$ and $\tilde{C}_\sA$ are the corresponding registers 
defined by the noiseless protocol embedding, and $A^{\prime \prime}$ is 
some scratch register used for her local quantum computation in the 
simulation. Similarly, Bob possesses a local quantum register 
$ B^\prime$ which  contains four subsystems of interest: to implement 
$\Pi$ as a black-box, the $B$ and $C_\sB$ parts correspond to the registers 
of the noiseless communication protocol, while $\tilde{B}$ is the 
corresponding register defined by the noiseless protocol embedding, and 
$B^{\prime \prime}$ is some scratch register used for his local quantum 
computation in the simulation. Eve possesses a local quantum 
register $ E^\prime$ which  contains two subsystems of interest: the 
$E$ part corresponds to her input register of the noiseless 
communication protocol and $E^{\prime \prime}$ is some scratch register 
used for her local quantum computation in the simulation. 
The input registers $ABC_\sA E$ are purified by a reference register $R$, which remains untouched throughout.
A quantum 
communication register $C^\prime$, of some fixed size $q$ independent of 
the length $N$ of the protocol to be simulated, is exchanged 
back and forth between Alice and Bob, passing through Eve; it is 
held by Alice at both the beginning and the end of the simulation 
protocol. A simulation protocol $Q$ in the quantum model of length $N^\prime$ 
is defined by a 
sequence of quantum instruments $\M_1^{ A^\prime C^\prime}$, $\M_2^{ 
B^\prime C^\prime}$, $\dotsc$, $\M_{N^\prime +1}^{ A^\prime C^\prime}$ such 
that, on state $\ket{\psi_{\mathrm{init}}^\prime}^{A^\prime 
B^\prime C^\prime E^\prime R} = \ket{\psi_{\mathrm{init}}}^{A B C_\sA E R} 
\otimes \ket{0}$ as input, given black-box access to a noiseless protocol $\Pi$ ($\Pi$ is assumed to be known to everyone)
and against an adversary $\msA$
defined by a 
sequence of quantum instruments $\N_1^{ E^\prime C^\prime}$, $\dotsc$, $
\N_{N^\prime}^{ E^\prime C^\prime}$, the protocol outputs the $\tilde{A} 
\tilde{B} \tilde{C}$ subsystems of 
\begin{align}
	\rho_{\mathrm{final}} = \M_{N^\prime + 1}^\Pi \N_{N^\prime} 
	\M_{N^\prime}^\Pi \cdots \M_2^\Pi \N_1 \M_1^\Pi 
	(\kb{\psi_{\mathrm{init}}^\prime}{\psi_{\mathrm{init}}^\prime}).
\end{align}
(Here, the superscript~$\Pi$ emphasizes the black-box access to the
protocol.)
We denote the state of the output registers $\tilde{A} \tilde{B} \tilde{C}$ 
%together with the reference register $R$ 
by $Q^\Pi(\msA(\ket{\psi_{\mathrm{init}}}))$, and 
the induced quantum channel from $ABCE$ to $\tilde{A} \tilde{B} \tilde{C} 
\cong ABC$ by $Q^\Pi (\msA)$.  The success of the simulation is measured 
by how close the simulation output state is to the final state of the 
noiseless protocol on the $ABC$ registers, and is captured by the 
following definition.
\begin{definition}
	A simulation protocol $Q$ in the quantum model of length 
	$N^\prime$ succeeds with error $\epsilon$ at simulating all 
	length $N$ noiseless protocols against all adversaries in some class 
	$\A$ if, for all noiseless protocols $\Pi$ of length 
	$N$, for all adversaries $\msA \in \A$, $\| \Pi - Q^\Pi 
	(\msA) \|_\diamond \leq \epsilon$. The communication rate $R_\rQ$ of 
	$Q$ is $R_\rQ = \tfrac{N}{N^\prime \log{q}}$ for $q \geq 2$ the 
	alphabet size of the communication register $C^\prime$.
\end{definition}
Note that the adversary only has to make the simulation 
fail on some particular protocol, and on some particular input, 
to characterize the simulation protocol as ineffective against her.

In a random error model (analogous to that studied in quantum information 
theory \`a la Shannon), Eve is a non-malicious passive environment, and 
$\N_i = \N^Q$ for some fixed quantum channel $\N^Q$, and the class 
$\A$ contains a single element $\N^{C^{\prime \otimes N^\prime}}$
(with trivial $Z, E^\prime$ registers). For simplicity, we then say that 
the simulation succeeds over $\N^Q$. In an adversarial error model 
(analogous to that studied in quantum coding theory, \`a la Hamming), Eve 
is a malicious adversary who wants to make the protocol fail, and we are 
interested in particular classes of adversaries which we denote 
by $\A_{\delta, q, N^\prime}^\rQ$ for some parameter~$\delta$ such that 
$0 \leq \delta \leq 1$. The class 
$\A_{\delta, q, N^\prime}^\rQ$ contains all adversaries with a bound $\delta$ on the 
fraction of communications of the $C^\prime$ register they corrupt, in 
the following sense. Here, $\F_{q^\prime, 1}, \E_{\delta, q, N^\prime}$ are 
	defined in Eqs.~\cref{eq:FqN} and~\cref{eq:Edelta}, respectively.
\begin{definition}
	The class $\A_{\delta, q, N^\prime}^\rQ$ of adversaries in the 
	quantum model with \emph{error rate} bounded by~$\delta$, 
	$0 \leq \delta \leq 1$, contains adversaries of the following
kind: each adversary is specified by a sequence of instruments
$\N_1^{ E^\prime C_1^{\prime}}$, $\dotsc$, $
	\N_{N^\prime}^{ E^\prime C_{N^\prime}^{\prime}}$ 
	with arbitrary local
	quantum register $E^\prime$ of dimension $q^\prime \in \mathbb{N}$. 
%	and local classical
%	register $Z$ with classical state set $\Z$, with $|\Z| \in \mathbb{N}$.
	All of these adversaries act on a quantum
	communication register $C^\prime$ of dimension $q \in \mathbb{N}$, 
	and on protocols of length $N^\prime \in \mathbb{N}$.
	For any $\rho 
%= \sum_{z_0 \in \Z} p_{Z_0} (z_0)  \kb{z_0}{z_0}^Z \otimes
%	 \rho(z_0)^{E^\prime C^{\prime \otimes N^\prime}}
	\in \D ( E^\prime \otimes C^{\prime \otimes N^\prime})$, 
	the action of such an adversary is
\begin{align*}
	\N_{N^\prime}^{ E^\prime C_{N^\prime}^{\prime}} \cdots
	\N_1^{ E^\prime C_1^{\prime}} (\rho) =
	\sum_{i} 
%p_{Z_0}(z_0) \kb{z}{z}^Z \otimes
	 G_{i}%^{E^\prime C^{\prime \otimes N^\prime}}
	\rho G_{i}^{\dagger}% E^\prime C^{\prime \otimes N^\prime}}
	,
\end{align*}
for $i$ ranging over some finite set and with each $G_{i}$ of the form
\begin{align*}
	G_{i} = \sum_{H \in \E_{\delta, q, N^\prime}, F \in \F_{q^\prime, 1}}
	\alpha_{H, F, i} F^{E^\prime} \otimes 
        H^{C^{\prime \otimes N^\prime}} \enspace,
\end{align*}
	which is also subject to the requirement that  $\sum_{i} 
	G_{i}^\dagger G_{i} = \rI^{ E^\prime C^{\prime \otimes N^\prime}}$.
\end{definition}

This adapts to an interactive communication model the
formal definition of adversarial channel given in Ref.~\cite{LS08} 
in a unidirectional communication model. Note that this 
allows for adaptive, probabilistic, entangled strategies for Eve, but 
such that any Kraus operator $G_{i,z,z_0}$ is a linear combination of
operators which act on at most a $\delta$ fraction of the $C^\prime$ registers
non-trivially. We therefore say that the fraction of errors is bounded 
by $\delta$ for all adversaries in $\A_{\delta, q, N^\prime}^\rQ$.

%
%\paragraph{Shared entanglement model.}
\subsection{Shared Entanglement Model}
\label{sec:nqucommsh}

For the \emph{shared entanglement model}, Alice, Bob and Eve possess 
local classical-quantum registers split analogously to those in the 
quantum model. In addition to the entanglement inherent in 
$\ket{\psi_{\mathrm{init}}}^{ABCER}$, Alice and Bob also share 
entanglement to be consumed during the simulation in the form of a large 
state $\ket{\phi}^{T_\sA T_\sB}$ with the registers $T_\sA$, $T_\sB$ held by Alice 
and Bob, respectively. In general, the entanglement registers have a 
product decomposition $T_\sA = T_\sA^1 \otimes \cdots \otimes T_\sA^{N^\prime}, 
T_\sB = T_\sB^{1} \otimes \cdots \otimes T_\sB^{N^\prime}$. A classical 
communication register $C^{\prime \prime}$, of some fixed size $q$ 
independent of the length $N$ of the protocol to be simulated, is 
exchanged back and forth between Alice and Bob, passing through Eve;
it is held by Alice at both the beginning and the end of the 
simulation protocol. A simulation protocol $S$ in the shared entanglement 
model of length $N^\prime$ is defined 
by a sequence of quantum instruments $\M_1^{ A^\prime 
T_\sA C^{\prime \prime}}$, $\M_2^{ B^\prime T_\sB C^{\prime \prime}}$, $\dotsc$, $
\M_{N^\prime + 1}^{ A^\prime T_\sA C^{\prime \prime}}$ such that, with
state $\ket{\psi_{\mathrm{init}}^\prime}^{A^\prime B^\prime C^{\prime 
\prime} E^\prime R} = \ket{\psi_{\mathrm{init}}}^{A B C_\sA E R} \otimes 
\ket{0}$ as input, given black-box access to a noiseless protocol $\Pi$, and 
against an adversary $\msA$ defined by a sequence of quantum 
instruments $\N_1^{ E^\prime C^{\prime \prime}}$, $\dotsc$, $
\N_{N^\prime}^{ E^\prime C^{\prime \prime}}$, the protocol outputs the 
$\tilde{A} \tilde{B} \tilde{C}$ subsystems of the
state~$\rho_{\mathrm{final}}$ given by
\begin{align}
	\rho_{\mathrm{final}} = \M_{N^\prime + 1}^\Pi \N_{N^\prime} 
	\M_{N^\prime}^\Pi \cdots \M_2^\Pi \N_1 \M_1^\Pi 
	(\kb{\psi_{\mathrm{init}}^\prime}{\psi_{\mathrm{init}}^\prime}).
\end{align}
(Again, the superscript~$\Pi$ emphasizes the black-box access to the
protocol by the simulator.)
We denote the state of the output registers $\tilde{A} \tilde{B} \tilde{C}$ 
%together with the reference register $R$ 
by $S^\Pi ( \msA (\ket{\psi_{\mathrm{init}}}))$, and 
the induced quantum channel from $ABCE$ to $\tilde{A} \tilde{B} \tilde{C} 
\cong ABC$ by $S^\Pi (\msA)$. The success of the simulation is measured by 
how close the simulation output state is to the final state of the 
noiseless protocol on the $ABC$ registers, and is captured by the 
following definition:
\begin{definition}
	A length~$N^\prime$ simulation protocol $S$ in the shared entanglement model of 
	 succeeds with error $\epsilon$ at simulating 
	all length $N$ noiseless protocols against all adversaries in 
	some class $\A$ if, for all noiseless protocols $\Pi$ 
	of length $N$, for all adversaries $\msA \in \A$, $\| 
	\Pi - S^\Pi (\msA) \|_\diamond \leq \epsilon$. The communication 
	rate $R_\rC$ of $S$ is $R_\rC = \tfrac{N}{N^\prime \log{q}}$ for $q 
	\geq 2$, the alphabet size of the classical communication register 
	$C^{\prime \prime}$, and the entanglement consumption rate $R_\rE$ 
	is $R_\rE = \tfrac{\log{(\max{(\dim{T_\sA}, \dim{T_\sB})})}}{N^\prime 
	}$ for $T_\sA, T_\sB$ the entanglement registers used for the 
	simulation by Alice and Bob, respectively.
\end{definition}

In a random error model, Eve is a non-malicious passive environment, 
$\N_i = \N^S$ for some fixed classical channel $\N^S$, and the class 
$\A$ contains a single element $\N^{C^{\prime \prime \otimes N^\prime}}$
(with trivial $Z, E^\prime$ registers). For simplicity, we 
then say that the simulation succeeds over $\N^S$. In an adversarial 
error model, Eve is a malicious adversary who wants to make the protocol 
fail, and we are interested in particular classes of adversaries, which we 
denote by $\A_{\delta, q, N^\prime}^\rS$ for some parameter $0 \leq \delta \leq 1$. 
The class 
$\A_{\delta, q, N^\prime}^\rS$ contains all adversaries with a bound $\delta$ on the 
fraction of communications of the $C^{\prime \prime}$ register they corrupt, in 
the following sense. Here, for two strings $c, c_0$ over a finite alphabet, 
$\Delta$ is the Hamming distance
function counting the number of positions in which $c, c_0$ differ; 
see 
%section 
\cref{sec:tc} for a formal definition.
\begin{definition}
	The class $\A_{\delta, q, N^\prime}^\rS$ of adversaries with
\emph{error rate} bounded by~$\delta$, $0 \leq \delta \leq 1$, in the 
	shared entanglement model contains adversaries of the following
kind: each adversary is specified by instruments
	$\N_1^{ E^\prime C_1^{\prime \prime}}, \dotsc, 
	\N_{N^\prime}^{ E^\prime C_{N^\prime}^{\prime \prime}}$
	with arbitrary local
	quantum register $E^\prime$ of dimension $q^\prime \in \mathbb{N}$. 
%	and local classical
%	register $Z$ with classical state set $\Z$, with $|\Z| \in \mathbb{N}$.
	All these instruments act on a classical
	communication register $C^{\prime \prime}$ of dimension $q \in \mathbb{N}$, 
	and on protocols of length $N^\prime \in \mathbb{N}$.
	For any $\rho 
%= \sum_{z_0 \in \Z} p_{Z_0} (z_0)  \kb{z_0}{z_0}^Z \otimes
%	 \rho(z_0)^{E^\prime C^{\prime \prime \otimes N^\prime}} 
	\in \D ( E^\prime \otimes C^{ \prime \prime \otimes N^\prime})$,
	the action of such an adversary is
\begin{align*}
	\N_{N^\prime}^{ E^\prime C_{N^\prime}^{\prime \prime}} \cdots
	\N_1^{ E^\prime C_1^{ \prime \prime}} (\rho) =
	\sum_{c, c_0} 
	 G_{c, c_0}%^{E^\prime C^{\prime \prime \otimes N^\prime}}
	\rho G_{c, c_0}^{\dagger}% E^\prime C^{\prime \prime \otimes N^\prime}}
	,
\end{align*}
for $c$, $c_0 \in \{0, 1, \cdots, q-1 \}^{N^\prime}$ satisfying 
$\Delta (c, c_0) \leq \delta N^\prime$ and with each $G_{c, c_0}$ of the form
\begin{align*}
	G_{c, c_0} = \sum_{F \in \F_{q^\prime, 1}}
	\alpha_{F, c, c_0}  F^{E^\prime} \otimes \kb{c}{c_0}^{C^{\prime \prime \otimes N^\prime}}
\enspace,
\end{align*}
	also subject to the requirement that for any 
	$c_0 \in \{0, 1, \dotsc, q-1 \}^{N^\prime}$,
\begin{align*}
\sum_{c} 
	G_{c, c_0}^\dagger G_{c, c_0} = \rI^{ E^\prime} \otimes \kb{c_0}{c_0}^{ C^{\prime \prime \otimes N^\prime}}.
\end{align*}
\end{definition}

Note that this allows for adaptive, probabilistic 
strategies for Eve, but such that conditioned on 
%any sequence of 
%measurement outcome $z$ (recorded in the $Z$ registers),
any
final transcript $c$ and input transcript $c_0$ on the communication register, 
%inputs $z_0$,  
at most a 
$\delta$ fraction of the actions of Eve have acted non-trivially on the 
$C^{\prime \prime}$ register, even though she can copy all classical 
transmissions in the $E^\prime$ registers. We therefore say that the fraction of 
error is bounded by $\delta$ for all adversaries in $\A_{\delta, q, N^\prime}^\rS$.

Note that the adversaries in the quantum and the shared entanglement 
models are fundamentally different: in the shared entanglement model, Eve 
can copy all classical messages without inducing any error and gather the corresponding information 
to establish her strategy, but she cannot modify Alice's or Bob's quantum 
information, except for what is possible by corrupting their classical 
communication and by using the information in the quantum register $E$ 
purifying the input state. In contrast, in the quantum model, she cannot 
always ``read'' the quantum messages, but she can apply entangled, fully 
quantum corruptions to the quantum register when she chooses to.

\section{Tree Code Figure}
\label{sec:apptreecode}

Figure~\ref{fig:tree-code} depicts two paths $x = za$ and $y = zb$ in a tree with divergence of length $\ell$, along with the encodings $\bar{\msE} (x)$, $\bar{\msE} (y)$, and $\bar{\msE} (z)$ of the strings $x$, $y$, and $z$. Let $a = a_1 a_2 \ldots a_\ell$ and $b = b_1 b_2 \ldots b_\ell$; then the tree code encoding of $x$ and $y$  are
$\bar{\msE}(x) = \bar{\msE}(z) \circ \bar{\msE}(a|z)$ and $\bar{\msE}(y) = \bar{\msE}(z) \circ \bar{\msE}(b|z)$,
in which $\circ$ is the concatenation operator for strings, $\bar{\msE}(a|z) = \msE (z a_1) \msE (z a_1 a_2) \dots \msE (z a)$ and $\bar{\msE}(b|z) = \msE (z b_1) \msE (z b_1 b_2) \cdots \msE (z b)$.

The main property of the tree code is : $\Delta(\bar{\msE}(x), \bar{\msE}(y)) = \Delta(\bar{\msE}(a|z), \bar{\msE}(b|z)) \geq \alpha \cdot \ell \enspace$ ;
i.e., the suffixes of the codewords are at distance at least~$\alpha \ell$.

		\begin{figure}
		\begin{center}
		\begin{overpic}[width=0.65\textwidth]{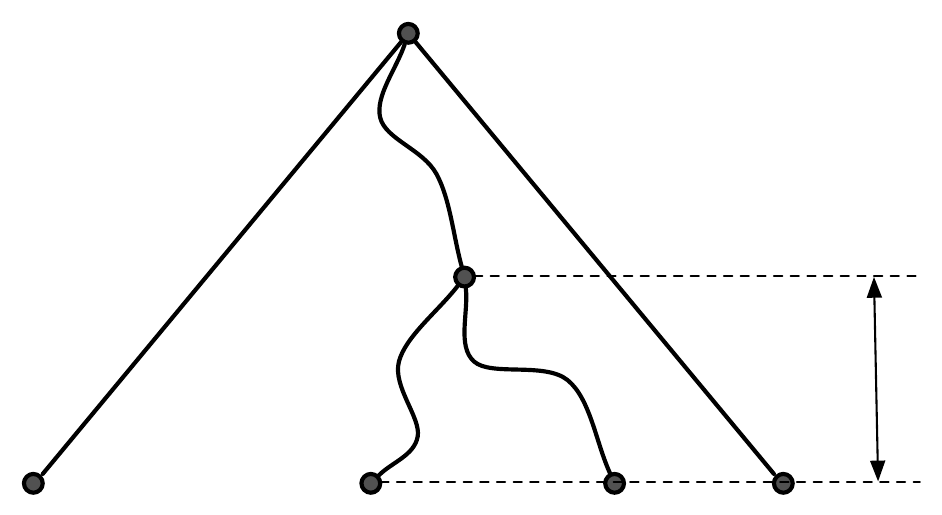}
		  \put(39,15.5){$a$}
		  \put(40,36.7){$z$}
		  \put(50.2,26.7){$\bar{\msE} (z)$}
		  \put(59.5,15.5){$b$}
		  \put(33.6,-2){$\bar{\msE} (x)$}
		  \put(59.5,-2){$\bar{\msE} (y)$}
		  \put(90,14){$\ell$}
		\end{overpic}
		  \caption{Depiction of  paths $x = za$ and $y = zb$ in a tree with divergence of length $\ell$, along with the encodings  $\bar{\msE} (x)$, $\bar{\msE} (y)$, and $\bar{\msE} (z)$ of these strings.}
		  \label{fig:tree-code}
		\end{center}
		\end{figure}

\bibliographystyle{siamplain}
\bibliography{references}
\end{document}